%% file: main.tex
\documentclass{llncs}

\usepackage{xcolor}

\usepackage{booktabs}   
\usepackage{subcaption} 

\usepackage[utf8]{inputenc}
\usepackage{amsthm}
\usepackage[ruled,vlined,linesnumbered]{algorithm2e}
\usepackage{amssymb}
\usepackage{amsmath}
\usepackage{bbm}
\usepackage{hyperref}
\usepackage{cleveref}
\usepackage{verbatim}
\usepackage{changepage}
\let\llncssubparagraph\subparagraph
\let\subparagraph\paragraph
\let\llncssubparagraph\subparagraph
\let\subparagraph\paragraph
\usepackage{titlesec}

\let\subparagraph\llncssubparagraph
\titlespacing{\subsection}{0pt}{6pt plus 4pt minus 2pt}{6pt plus 0pt minus 4pt}
\titlespacing{\subsubsection}{0pt}{6pt plus 0pt minus 2pt}{6pt plus 0pt minus 4pt}

\theoremstyle{plain}
\newtheorem{thm}{Theorem}
\makeatletter
\def\thm@space@setup{%
  \thm@preskip=6pt plus 0cm minus 4pt
  \thm@postskip=\thm@preskip 
}
\makeatother

\usepackage{tikz}
\usetikzlibrary{automata, positioning, arrows}
\tikzset{
->, 
>=stealth, 
node distance=3cm, 
shorten >=1pt,
every state/.style={thick, fill=gray!10}, 
inner sep=0pt,
minimum size=0pt,
initial text=$ $, 
}
\usepackage{lscape}
\usepackage{tabularx}
\usepackage{spverbatim}

\newcommand{\inputs}{{\mathcal I}}
\newcommand{\outputs}{{\mathcal O}}
\newcommand{\inp}{{\sf i}}
\newcommand{\out}{{\sf o}}

\newcommand{\fIO}{f_{\mathbbm{i}\mathbbm{o}}}

\newcommand{\prefs}{{\sf Prefs}}
\newcommand{\MM}{{\mathcal{M}}}
\newcommand{\PP}{{\mathcal{P}}}
\newcommand{\TT}{{\mathcal{T}}}

\newcommand{\atm}{{\mathcal{A}}}
\newcommand{\spec}{{\mathcal{S}}}

\newcommand{\inproj}{{\sf in}}
\newcommand{\outproj}{{\sf out}}
\newcommand{\sem}[1]{{[\!\![ #1]\!\!]}}

\newcommand{\transS}{\textsf{Post}}
\newcommand{\transO}{\textsf{Out}}
\newcommand{\trans}{\Delta}

\newcommand{\finer}{\sqsubseteq}
\newcommand{\finerstrict}{\sqsubset}
\newcommand{\anti}{\mathcal{A}\mathcal{C}}
\let\oldexample\example
\renewcommand{\example}{\oldexample\normalfont}

\usepackage{enumitem}
\setlist[itemize]{itemsep=0.8pt, topsep=2pt}
\setlist[enumerate]{itemsep=0.8pt, topsep=2pt}

\newcommand{\tocheck}[1]{\textcolor{black}{#1}}

\begin{document}
\title{LTL Reactive Synthesis with a Few Hints}

\author{Mrudula Balachander \thanks{Mrudula Balachander: Mrudula Balachander is a Research Fellow at F.R.S-FNRS.} \and Emmanuel Filiot \thanks{Emmanuel Filiot: Emmanuel Filiot is a senior research associate at F.R.S-FNRS.}  \and
Jean-Fran\c{c}ois Raskin} 
\institute{Universit\'e libre de Bruxelles, Brussels, Belgium \thanks{This work was partially supported by the Fonds de la Recherche Scientifique – F.R.S.-FNRS under the MIS project F451019F and by the ASP-Fellowship grant.}}


\maketitle

\input{abstract.tex}

\section{Introduction}


Reactive systems are notoriously difficult to design and even to specify correctly~\cite{DBLP:conf/icalp/AbadiLW89,DBLP:journals/tosem/DIppolitoBPU13}. As a consequence, formal methods have emerged as useful tools to help designers to built reactive systems that are correct. For instance, model-checking asks the designer to provide a model, in the form of a Mealy machine $\MM$, that describes the reactions of the system to events generated by its environment, together with a description of the {\em core correctness properties} that must be enforced. Those properties are expressed in a logical formalism, typically as an LTL formula $\varphi_{{\sf CORE}}$. Then an algorithm decides if $\MM \models \varphi_{{\sf CORE}}$, i.e. if all executions of the system in its environment satisfy the specification. 
Automatic reactive synthesis is more ambitious: it aims at automatically generating a model from a high level description of the ``{\em what}'' needs to be done instead of the ``{\em how}'' it has to be done. Thus the user is only required to provide an LTL specification $\varphi$ and the algorithm automatically generates a Mealy machine $\MM$ such that $\MM \models \varphi$ whenever $\varphi$ is {\em realizable}. Unfortunately, it is most of the time not sufficient to provide the core correctness properties $\varphi_{{\sf CORE}}$ to obtain a Mealy machine $\MM$ that is useful in practice, as illustrated next.

\begin{example}[Synthesis from $\varphi_{{\sf CORE}}$ - Mutual exclusion]\label{ex:mutex}
Let us consider the classical problem of {\em mutual exclusion}. In the simplest form of this problem, we need to design an arbiter that receives requests from two processes, modeled by two atomic propositions $r_1$ and $r_2$ controlled by the environment, and that grants accesses to the critical section, modeled as two atomic propositions $g_1$ and $g_2$ controlled by the system. The core correctness properties (the {\em what}) are: $(i)$ mutual access, i.e. it is never the case that the access is granted to both processes at the same time, $(ii)$ fairness, i.e. processes that have requested access eventually get access to the critical section. These core correctness specifications for mutual exclusion ({\sf ME}) are easily expressed in LTL as follows:
$\varphi^{{\sf ME}}_{{\sf CORE}} \equiv \square ( \neg g_1 \lor \neg g_2) \land \square ( r_1 \rightarrow \lozenge g_1) \land \square ( r_2 \rightarrow \lozenge g_2)$. Indeed, this formula expresses the core correctness properties that we would model check no matter {\em how} $\MM$ implements mutual exclusion, e.g. Peterson, Dedekker, Backery algorithms, etc. Unfortunately, if we submit $\varphi^{{\sf ME}}_{{\sf CORE}}$ to an LTL synthesis procedure, implemented in tools like {\sc Acacia-Bonzai}~\cite{DBLP:journals/corr/abs-2204-06079}, {\sc BoSy}~\cite{DBLP:conf/cav/FaymonvilleFT17}, or {\sc Strix}~\cite{DBLP:conf/cav/MeyerSL18}, we get the solution $\MM$ depicted in~\ref{fig:ME-simple}-(left) (all three tools return this solution). While this solution is perfectly correct and realizes the specification $\varphi^{{\sf ME}}_{{\sf CORE}}$, the solution ignores the inputs from the environment and grants access to the critical sections in a round robin fashion. Arguably, it may not be considered as an \tocheck{efficient} solution to the mutual exclusion problem. This illustrates the limits of the synthesis algorithm to solve the design problem by providing {\em only} the core correctness specification of the problem, i.e. the  {\em what}, only. To produce useful solutions to the mutual exclusion problem, more guidance must be provided.
\end{example}

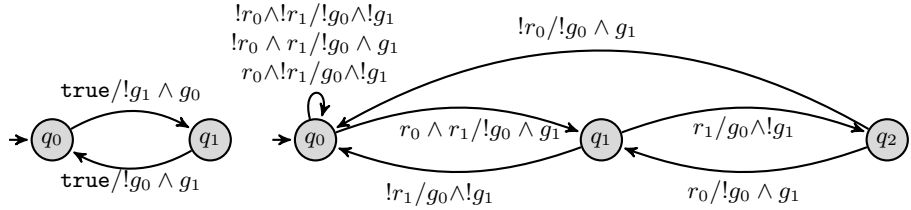
\begin{figure}
    \centering
    \begin{tikzpicture}[->,>=stealth',shorten >=1pt,auto,node distance=3.8cm,
                    thick,inner sep=0pt,minimum size=0pt,scale=0.7]
  \tikzstyle{every state}=[fill=gray!30,text=black,inner
  sep=2pt,minimum size=12pt]
  
        \node[state, initial] (q0) {$q_0$};
        \node[state] (q1) at (3,0) {$q_1$};
        \draw
            (q0) edge[bend left, above] node[yshift=1mm]{{\small
                $\mathtt{true}/!g_1\land g_0$}} (q1)
            (q1) edge[bend left, below] node{{\small $\mathtt{true}/!
                g_0\land g_1$}} (q0);

                    \node[state, initial] (q0x) at (5,0) {$q_0$} ;
        \node[state, right of=q0x] (q1x) {$q_1$};
        \node[state, right of=q1x] (q2x) {$q_2$};
        \draw (q0x) edge[loop above, above, align=center]
        node[yshift=1mm]{$!r_0\land !r_1/!g_0\land !g_1$ \\ $!r_0\land
          r_1/!g_0\land g_1$ \\ $r_0\land !r_1/g_0\land !g_1$} (q0x)
        (q0x) edge[bend left=20, right] node[yshift=-3mm,xshift=-8mm]{$r_0\land r_1/!g_0\land g_1$} (q1x)
        (q1x) edge[bend left=20, left] node[yshift=-3mm,xshift=5mm]{$!r_1/g_0\land !g_1$} (q0x)
        (q1x) edge[bend left=20, below] node[yshift=-1mm]{$r_1/g_0\land !g_1$} (q2x)
        (q2x) edge[bend right=30, left] node[yshift=3mm,xshift=5mm]{$!r_0/!g_0\land g_1$} (q0x)
        (q2x) edge[bend left=20, below] node[yshift=-1mm]{$r_0/!g_0\land g_1$} (q1x);

        \end{tikzpicture}
    \caption{(Left) The solution provided by
      Strix to the mutual exclusion problem  for the high level specification $\varphi_{{\sc LOW}}^{{\sf ME}}$.
    Edge labels are of the form $\varphi /\psi$ where $\varphi$ is a Boolean formula on the input atomic
    propositions (the Boolean variables controlled by the environment)
    and $\psi$ is a maximally consistent
    conjunction  of literals over the set of output propositions (the
    Boolean variables controlled by the system). (Right) A natural solution that we would write by hand, and is automatically produced by our learning and synthesis algorithm for the same specification together with two simple examples.}
    \label{fig:ME-simple}
    \label{fig:nicesolution}
\end{figure}

The main question is now: {\em how should we specify these additional properties ?} Obviously, if we want to use the "plain" LTL synthesis algorithm, there is no choice: we need to reinforce the specification $\varphi^{{\sf ME}}_{{\sf CORE}}$ with additional lower level properties $\varphi^{{\sf ME}}_{{\sf LOW}}$. Let us go back to our running example.

\begin{example}[Synthesis from $\varphi^{{\sf ME}}_{{\sf CORE}}$ and $\varphi^{{\sf ME}}_{{\sf LOW}}$]
To avoid solutions with {\em unsolicited grants}, we need to reinforce the core specification. The Strix online demo website proposes to add the following 3 LTL formulas $\varphi^{{\sf ME}}_{{\sf LOW}}$  to $\varphi^{{\sf ME}}_{{\sf CORE}}$ (see Full arbitrer $n=2$, at \texttt{https://meyerphi.github.io/strix-demo/}):
$(1)$ $\bigwedge_{i \in \{1,2\}} \square((g_i \land \square \neg r_i) \rightarrow \lozenge  \neg g_i)$, $(2)$ $\bigwedge_{i \in \{1,2\}} \square(g_i \land \bigcirc (\neg r_i \land \neg g_i) \rightarrow \bigcirc (r_i {\sf R} \neg g_i))$, and $(3)$ $\bigwedge_{i \in \{1,2\}} (r_i {\sf R} \neg g_i)$.
Now, while the specification $\varphi^{{\sf ME}}_{{\sf CORE}} \land \varphi^{{\sf ME}}_{{\sf LOW}}$ allows Strix to provide us with a better
solution, it is more complex than needed (it has 9 states and can be
seen in App.~\ref{app:nounsoliticed}) and clearly does not look like an optimal solution to our mutual exclusion problem. For instance, the model of Fig.~\ref{fig:nicesolution}-(right) is arguably more natural. How can we get this model without coding it into the LTL specification,  which would diminish greatly the interest of using a synthesis procedure in the first place?
\end{example}

\tocheck{In general, higher level properties are ones that are concerned with safety and are the ones needed to be verified on all implementations. In contrast, lower level properties are more about a specific implementation, i.e., they talk more about expected behaviour and are concerned with the efficiency of the implementation.}
At this point, it is legitimate to question the adequacy of LTL as a specification language for {\em lower level} properties, and so as a way to guide the synthesis procedure towards relevant solutions to realize $\varphi_{{\sf CORE}}$.
In this paper, we introduce an alternative to guide synthesis toward useful solutions that realize $\varphi_{{\sf CORE}}$: we propose to use examples of executions that illustrate behaviors of expected solutions.
We then restrict the search to solutions that {\em generalize} those examples.
Examples, or scenarios of executions, are accepted in requirement engineering as an adequate tool to elicit requirements about complex systems~\cite{DBLP:conf/sigsoft/DamasLL06}.
For reactive system design, examples are particularly well-suited as they are usually much easier to formulate than full blown solutions, or even partial solutions.
It is because, when formulating examples, the user controls {\em both} the inputs {\em and} the outputs, avoiding the main difficulty of reactive system design: having to cope with {\em all} possible environment inputs.
We illustrate this on our running example.

\begin{example}[Synthesis from $\varphi^{{\sf ME}}_{{\sf CORE}}$ and examples] \label{ex:mutexTraces}Let us keep, as the LTL specification, $\varphi^{{\sf ME}}_{{\sf CORE}}$ only, and let us consider the following simple prefix of executions that illustrate how solutions to mutual exclusion should behave: 
  \begin{itemize}
      \item[$(1)$] $\{!r_1,!r_2\} . \{!g_1,!g_2\} \#\{r_1,!r_2\} . \{g_1,!g_2\} \# \{!r_1,r_2\} . \{!g_1,g_2\}$
      \item[$(2)$] $\{r_1,r_2\} . \{g_1,!g_2\} \# \{!r_1,!r_2\} . \{!g_1,g_2\}$
  \end{itemize}
These prefixes of traces prescribe reactions to typical {\em fixed} finite sequences of inputs: $(1)$ if there is no request initially, then no access is granted (note that this excludes already the round robin solution), if process 1 requests and subsequently process 2 requests, process 1 is granted first and then process 2 is granted after, $(2)$ if both process request simultaneously, then process 1 is granted first and then process 2 is granted after. Given those two simple traces together with $\varphi_{{\sf CORE}}$, our algorithm generates the solution of Fig.~\ref{fig:nicesolution}-(right). Arguably, the solution is now simple and natural. 
\end{example}

\paragraph{{\bf Contributions}} First, we provide a synthesis algorithm {\sc SynthLearn} that, given an LTL specification $\varphi_{{\sf CORE}}$ and a finite set $E$ of prefixes of executions, returns a Mealy machine $\MM$ such that $\MM \models \varphi_{{\sf CORE}}$, i.e. $\MM$ realizes $\varphi_{{\sf CORE}}$, and $E \subseteq {\sf Prefix}(L(\MM))$, i.e. $\MM$ is compatible with the examples in $E$, if such a machine $\MM$  exists. It returns {\em unrealizable} otherwise. Additionally, we require {\sc SynthLearn} to {\em generalize} the decisions illustrated in $E$. This learnability requirement is usually formalized in automata learning with a {\em completeness criterium} that we adapt here as follows: for all specifications $\varphi_{{\sf CORE}}$, and for all Mealy machines $\MM$ such that $\MM\models \varphi_{{\sf CORE}}$, there is a small set of examples $E$ (polynomial in $|\MM|$) such that $L(\text{{\sc SynthLearn}}(\varphi_{{\sf CORE}},E))=L(\MM)$. We prove this completeness result in Theorem~\ref{thm:mealycompleteness} for safety specifications and extend it to $\omega$-regular and LTL specifications in Section~\ref{sec:omega-reg}, by reduction to safety. 

Second, we prove that the worst-case execution time of {\sc
  SynthLearn} is {\sc 2ExpTime} (Theorem~\ref{thm:complexityLTL}), and
this is worst-case optimal as the plain LTL synthesis problem (when
$E=\emptyset$) is already known to be {\sc
  2ExpTime-Complete}~\cite{DBLP:conf/icalp/PnueliR89}. {\sc
  SynthLearn} first {\em generalizes} the examples provided by the
user while maintaining realizability of $\varphi_{{\sf CORE}}$. This
generalization leads to a Mealy machine with possibly missing
transitions (called a preMealy machine). Then, this preMealy machine
is extended into a (full) Mealy machine that realizes $\varphi_{{\sf
    CORE}}$ against all behaviors of the environment. During the
completion phase, {\sc SynthLearn} reuses as much as possible
decisions that have been generalized from the examples. The
generalization phase is essential to get the most out of the examples.
Running classical synthesis algorithms on $\varphi_{{\sf CORE}} \land
\varphi_E$, where $\varphi_E$ is an LTL encoding of $E$, often leads
to more complex machines that fail to generalize the decisions taken
along the examples in $E$. While the overall complexity of {\sc
  SynthLearn} is {\sc 2ExpTime} and optimal, we show that it is only
polynomial in the size of $E$ and in a well-chosen symbolic
representation a set of Mealy machines that realize $\varphi_{{\sf
    CORE}}$, see Theorem~\ref{thm:onlypoly}. This symbolic
representation takes the form of an antichain of functions and tends
to be compact in practice~\cite{DBLP:journals/fmsd/FiliotJR11}. It is
computed by default when {\sc Acacia-Bonzai} is solving the plain LTL
synthesis problem of $\varphi_{{\sf CORE}}$. So, generalizing examples
while maintaining realizability only comes at a marginal polynomial
cost. We have implemented our synthesis algorithm in a prototype,
which uses {\sc Acacia-Bonzai} to compute the symbolic antichain
representation. We report on the results we obtain on several examples.

\paragraph{{\bf Related works}}
Scenarios of executions have been advocated by
researchers in requirements engineering to elicite specifications, see
e.g.~\cite{DBLP:conf/sigsoft/DamasLL06,DBLP:journals/aai/DupontLDL08}
and references therein. 
In~\cite{DBLP:conf/tacas/RahaRFN22}, learning techniques are used to transform examples into LTL formulas that generalize them.  
Those methods are complementary to our work, as
they can be used to obtain the high level
specification $\varphi_{{\sf CORE}}$.

\tocheck{In non-vacuous synthesis~\cite{BCES17}, examples are added automatically to an LTL specification in order to force the synthesis procedure to generate solutions that are non-vacuous in the sense of~\cite{KV99}. The examples are generated directly from the syntax of the LTL specification and they cannot be proposed by the user. This makes our approach and this approach orthogonal and complementary. Indeed, we could use the examples generated automatically by the non-vacuous approach and ask the user to validate them as desirable or not. Our method is more flexible, it is semi-automatic and user centric: the user can provide any example he/she likes and so it offers more flexibility to drive the synthesis procedure to solutions that the user deems as interesting. Furthermore, our synthesis procedure is based on learning algorithms, while the algorithm in~\cite{BCES17} is based on constraint solving and it does not offer guarantees of generalization contrary to our algorithm (see Theorem~\ref{thm:mealycompleteness}).}

Supplementing the formal specification with additional user-provided
information is at the core of the  \emph{syntax-guided
synthesis} framework (SyGuS~\cite{DBLP:series/natosec/AlurBDF0JKMMRSSSSTU15}), implemented for instance in \emph{program by sketching}~\cite{DBLP:journals/sttt/Solar-Lezama13}: in SyGuS, the specification is a
logical formula and candidate programs are syntactically restricted by
a user-provided grammar, to limit and guide the search. 
The search is
done by using counter-example guided inductive synthesis techniques (CEGIS)
which rely on learning~\cite{DBLP:conf/asplos/Solar-LezamaTBSS06}. In contrast to our approach,
examples are not user-provided but automatically generated by
model-checking the candidate programs against the specification. The techniques are also orthogonal to ours: SyGuS targets programs syntactically defined by expressions over a decidable background theory, and heavily relies on SAT/SMT solvers. 
Using examples to synthesise programs
(\emph{programming by example}) has been
for instance explored in the context of string processing programs for
spreadsheets, based on learning~\cite{DBLP:conf/popl/SinghG16}, and is
a current trend in AI (see for
example~\cite{DBLP:conf/aistats/NatarajanSDJG19} and the citations
therein). However this approach only relies on
examples and not on logical specifications.

\cite{DBLP:conf/hvc/AlurMRSTU14} explores the use of formal
specifications and scenarios to synthesize distributed
protocols. Their approach also follows two phases: first,
an incomplete machine is built from the scenarios and second, it is
turned into a complete one. But there are two important differences
with our work. First, their first phase does not rely on learning
techniques and does not try to generalize the provided examples. Second,
in their setting, all actions are controllable and there is no adversarial
environment, so they are solving a satisfiability problem and not a
realizability problem as in our case. Their problem is thus
computationally less demanding than the problem we solve: {\sc
  Pspace} versus {\sc 2ExpTime} for LTL specs.

The synthesis problem targeted in this paper extends the LTL synthesis problem. Modern solutions for this problem use automata constructions that avoid Safra's construction as first proposed in~\cite{DBLP:conf/focs/KupfermanV05}, and simplified in~\cite{DBLP:conf/atva/ScheweF07a,DBLP:conf/cav/FiliotJR09}, and more recently in~\cite{DBLP:conf/tacas/EsparzaKRS17}. Efficient implementations of Safraless constructions are available, see e.g.~\cite{DBLP:conf/cav/BohyBFJR12,DBLP:conf/cav/FaymonvilleFT17,DBLP:conf/cav/MeyerSL18,DBLP:journals/corr/abs-2206-11366}. Several previous works have 
proposed alternative approaches to improve on the quality of solutions that synthesis algorithms can offer. A popular research direction, orthogonal and complementary to the one proposed here, is to extend the formal specification with quantitative aspects, see e.g.~\cite{DBLP:conf/cav/BloemCHJ09,DBLP:journals/iandc/BruyereFRR17,DBLP:conf/csr/Kupferman16,DBLP:conf/concur/AlmagorKV16}, and only synthesize solutions that are optimal. 

The first phase of our algorithm is inspired by automata learning techniques based on state merging algorithms like RPNI~\cite{DBLP:series/synthesis/2015Heinz,DBLP:journals/sttt/GiantamidisTB21}. Those learning algorithms need to be modified carefully to generate partial solutions that preserve realizability of $\varphi_{{\sf CORE}}$. Proving completeness as well as termination of the completion phase in this context requires particular care.

\section{Preliminaries on the reactive synthesis problem}\label{sec:prelims}

\paragraph{Words, languages and automata} An alphabet is a finite set of
symbols.   A \emph{word} $u$
(resp. $\omega$-word) over an alphabet $\Sigma$ is a finite
(resp. infinite sequence) of symbols from $\Sigma$. We write
$\epsilon$ for the empty word, and denote by
$|u|\in\mathbb{N}\cup \{\infty\}$ the length of $u$. In particular,
$|\epsilon|=0$. For $1\leq i\leq j\leq |u|$, we let $u[i{:}j]$ be the
infix of $u$ from position $i$ to position $j$, both included, and
write $u[i]$ instead of $u[i{:}i]$. The set of finite (resp. $\omega$-)
words over $\Sigma$ is denoted by $\Sigma^*$ (resp. $\Sigma^\omega$). We let $\Sigma^\infty
=\Sigma^*\cup \Sigma^\omega$. Given two words $u\in\Sigma^*$ and
$v\in\Sigma^\infty$, $u$ is a \emph{prefix} of $v$,
written $u\preceq v$,  if $v = uw$ for some $w\in\Sigma^\infty$. The
set of prefixes of $v$ is denoted by $\prefs(v)$. Finite words are linearly
ordered according to the length-lexicographic order $\preceq_{ll}$,
assuming a linear order $<_{\Sigma}$ over $\Sigma$: $u\preceq_{ll} v$
if $|u|<|v|$ or $|u|=|v|$ and $u=p\sigma_1u'$, $v=p\sigma_2 v'$ for
some $p,u',v'\in\Sigma^*$ and some $\sigma_1<_{\Sigma} \sigma_2$. In
this paper, whenever we refer to the order $\preceq_{ll}$ for words
over some alphabet, we implicitly assume the existence of an arbitrary
linear order over that alphabet. A \emph{language} (resp. $\omega$-language) over an alphabet $\Sigma$
is a subset $L\subseteq \Sigma^*$ (resp. $L\subseteq \Sigma^\omega$).

In this paper, we fix two alphabets $\inputs$ and $\outputs$
whose elements are called inputs and outputs respectively. Given a
word $u\in (\inputs\outputs)^\infty$, we let $\inproj(u) \in
\inputs^\infty$ be the word obtained by erasing all $\outputs$-symbols
from $u$. We define $\outproj(u)$ similarly and naturally extend both
functions to languages. 

\paragraph{Automata over $\omega$-words} A {\em parity automaton}
 is a tuple $\atm=(Q,Q_{{\sf init}},\Sigma,\delta,d)$ where $Q$ is a
finite non empty set of states, $Q_{{\sf init}} \subseteq Q$ is a set
of initial states, $\Sigma$ is a finite non empty alphabet, $\delta :
Q \times \Sigma \rightarrow 2^Q \setminus \{\emptyset\}$ is the transition
function, and $d : Q \rightarrow \mathbb{N}$ is a parity function. The
automaton $\atm$ is {\em deterministic} when $|Q_{{\sf init}}|=1$ and
$|\delta(q,\sigma)|=1$ for all $q\in Q$. The transition function is
extended naturally into a function $\transS^* : Q\times \Sigma^*\rightarrow 2^Q
\setminus \{\emptyset\}$ inductively as follows:
$\transS^*(q,\epsilon)=\{q\}$ for all $q\in Q$ and for all
$(u,\sigma)\in\Sigma^*\times \Sigma$, $\transS^*(q,u\sigma) =
\bigcup_{q'\in \transS^*(q,u)}\delta(q',\sigma)$.

A run of $\atm$ on an $\omega$-word $w=w_0 w_1 \dots$ is an
infinite sequence of states $r=q_0 q_1 \dots$ such that
$q_0 \in Q_{{\sf init}}$, and for all $i \in \mathbb{N}$, $q_{i+1} \in
\delta(q_i,w_i)$. The run $r$ is said to be {\em accepting} if the
minimal colour it visits infinitely often is even, i.e. $\liminf
(d(q_i))_{i\geq 0}$ is {\it even}. We say that $\atm$ is a
{\em B\"uchi automaton} when ${\sf dom}(d)=\{0,1\}$ ($1$-coloured states are
called accepting states), a {\em co-B\"uchi automaton} when ${\sf dom}(d)=\{1,2\}$,
a {\em safety automaton} if it is a B\"uchi automaton such that the set
of $1$-coloured states, called \emph{unsafe states} and denoted
$Q_{\sf usf}$, forms a \emph{trap}: for all $q\in Q_{\sf usf}$, for all
$\sigma\in\Sigma$, $\delta(q,\sigma)\subseteq Q_{\sf usf}$, and 
a \emph{reachability automaton} if it is $\{0,1\}$-coloured and the
set of $0$-coloured states forms a trap.

Finally, we consider the existential and universal interpretations of
nondeterminism, leading to two different notions of $\omega$-word
languages: under the {\em existential (resp. universal) interpretation}, a word $w \in
\Sigma^{\omega}$ is in the language of $\atm$, if there exists a
run $r$ on $w$ such that $r$ is accepting (resp. for all runs $r$ on
$w$, $r$ is accepting). We denote the two languages defined by these
two interpretations $L^{\exists}(\atm)$ and $L^{\forall}(\atm)$ respectively.
Note that if $\atm$ is deterministic, then the existential and
universal interpretations agree, and we write $L(\atm)$ for
$L^\forall(\atm) = L^\exists(\atm)$. Sometimes, for a deterministic automaton $\atm$, we change the initial state to a state $q \in Q$, and note $\atm[q]$ for the deterministic automaton $\atm$ where the initial state is fixed to the singleton $\{q\}$.

For a {\em co-B\"uchi automaton}, we also define a strengthening of
the acceptance condition, called $K$-co-B\"uchi, which requires,
for $K \in \mathbb{N}$, that a run visits at most $K$ times a state
labelled with $1$ to be accepting. Formally, a run $r=q_0 q_1 \dots
q_n \dots$ is {\em accepting} for the $K$-co-B\"uchi acceptance
condition if $|\{ i \geq 0 \mid d(q_i))=1 \}| \leq K$. The language
defined by $\atm$ for the $K$-co-B\"uchi acceptance condition and
universal interpretation is denoted by $L^{\forall}_K(\atm)$. Note
that this language is a {\em safety} language because if a prefix of a
word $p \in \Sigma^*$ is such that $\atm$ has a run prefix on $p$ that
visits more than $K$ times a states labelled with color $1$, then all
possible extensions $w \in \Sigma^{\omega}$ of $p$ are rejected by
$\atm$. 



\paragraph{(Pre)Mealy machines}  Given a (partial) function $f$ from a set $X$ to a set $Y$, we denote by
$\textsf{dom}(f)$ its domain, i.e. the of elements $x\in X$ such that
$f(x)$ is defined. A \emph{preMealy machine} $\MM$ on an input alphabet $\inputs$ and output
alphabet $\outputs$ is a triple $(M,m_{{\sf init}},\trans)$ such that $M$ is a
non-empty set of states, $m_{{\sf init}} \in M$ is the initial state, $\trans : Q
\times \inputs \rightarrow \outputs \times M$ is a partial function. A
pair $(m,\inp)$ is a hole in $\MM$ if $(m,\inp) \not\in {\sf dom}(\trans)$.
A \emph{Mealy machine} is a preMealy machine such that $\trans$ is total,
i.e., ${\sf dom}(\trans)=M \times \inputs$.

We define two semantics of a preMealy machine $\MM = (M,m_{{\sf
    init}},\trans)$ in terms of the languages of finite and infinite
words over $\inputs\cup\outputs$ they define. First, 
    we define two (possibly partial functions) $\transS_\MM : M\times
    \inputs\rightarrow M$ and $\transO_\MM : M\times \inputs\rightarrow
    \outputs$ such that $\trans(m,\inp) =
    (\transS_\MM(m,\inp),\transO_\MM(m,\inp))$ for all $(m,\inp)\in
    M\times\inputs$ if $\trans(m,\inp)$ is defined. We naturally extend these two functions to any
    sequence of inputs $u\in\inputs^+$, denoted $\transS_\MM^*$ and
    $\transO_\MM^*$. In particular, for $u\in\inputs^+$,
    $\transS_\MM^*(m,u)$ is the state reached by $\MM$ when reading
    $u$ from $m$, while $\transO_\MM^*(m,u)$ is the last output in $\outputs$
    produced by $\MM$ when reading $u$. The subcript $\MM$ is ommitted
    when $\MM$ is clear from the context. Now, the language $L(\MM)$
    of finite words in $(\inputs\outputs)^*$ accepted by $\MM$ is
    defined as $L(\MM)=\{ \inp_1\out_1\dots \inp_n\out_n\mid \forall 1\leq
    j\leq n,\ \transS_\MM^*(m_{\sf init},\inp_1\dots\inp_j) \text{ is defined and }
    \out_j = \transO_\MM^*(m_{\sf init},\inp_1\dots\inp_j)\}$. The language
    $L_\omega(\MM)$ of
    infinite words accepted by $\MM$ is the topological closure of
    $L(\MM)$: $L_\omega(\MM) = \{ w\in
    (\inputs\outputs)^\omega\mid \prefs(w)\cap
    (\inputs\outputs)^*\subseteq  L(\MM)\}$.

\paragraph{The reactive synthesis problem} A \emph{specification} is a language $\spec \subseteq
(\inputs \outputs)^\omega$. The \emph{reactive synthesis problem} (or
just synthesis problem for short) is the problem of constructing,
given a specification $\spec$, a Mealy machine $\MM$ such that
$L_\omega(\MM)\subseteq \spec$ if it exists. Such a machine $\MM$ is said to
\emph{realize} the specification $\spec$, also written
$\MM\models \spec$. We also say that $\spec$ is
\emph{realizable} if some Mealy machine $\MM$ realizes it. The induced decision problem is called the \emph{realizability problem}. 

It is well-known that if $\spec$ is
$\omega$-regular \tocheck{(recognizable by a parity automaton \cite{T91})} the realizability problem is decidable~\cite{DBLP:conf/icalp/AbadiLW89} and moreover, a Mealy machine realizing the specification can be effectively constructed.
The realizability problem is {\sc 2ExpTime-Complete} if $\spec$ is given as an LTL formula~\cite{DBLP:conf/icalp/PnueliR89} and {\sc ExpTime-Complete} if $\spec$ is given as a universal coB\"uchi automaton.

\begin{thm}[\cite{DBLP:reference/mc/BloemCJ18}]\label{thm:folklore}
    The realizability problem for a specification $\spec$ given as a universal
    coB\"uchi automaton $\atm$ is {\sc ExpTime-Complete}. Moreover, if $\spec$
    is realizable and $\atm$ has $n$ states, then $\spec$ is realizable by a Mealy machine with $2^{O(n
      log_2 n)}$ states. 
\end{thm}

We generalize this result to the following
realizability problem  which we describe first informally. Given a specification $\spec$ and a preMealy
machine $\PP$, the goal is to decide whether $\PP$ can be completed
into a Mealy machine which realizes $\spec$. We now define this problem
formally. Given two preMealy machines $\PP_1,\PP_2$, we write $\PP_1\preceq
\PP_2$ if $\PP_1$ is a subgraph of $\PP_2$ in the following sense: there
exists an injective mapping $\Phi$ from the states of $\PP_1$ to the
states of $\PP_2$ which preserves the initial state ($s_0$ is the
initial state of $\PP_1$ iff $\Phi(s_0)$ is the initial state of $\PP_2$)
and the transitions ($\Delta_{\PP_1}(p,\inp)=(\out,q)$ iff
$\Delta_{\PP_2}(\Phi(p),\inp)=(\out,\Phi(q))$. As a consequence,
$L(\PP_1)\subseteq L(\PP_2)$ and $L_\omega(\PP_1)\subseteq L_\omega(\PP_2)$. Given a preMealy machine $\PP$, we say that a specification \emph{$\mathcal{S}$ is $\PP$-realizable} if there exists a Mealy machine $\MM$ such that
$\PP\preceq \MM$ and $\MM$ realizes $\spec$. Note that if $\PP$
is a (complete) Mealy machine, $\spec$ is $\PP$-realizable
iff $\PP$ realizes $\spec$.

\begin{thm}
\label{thm:sizeSol}
    Given a universal co-B\"uchi automaton $\atm$ with $n$ states defining a
    specification $\spec = L^\forall(\atm)$ and a preMealy machine $\PP$
    with $m$ states and $n_h$ holes,
    deciding whether $\spec$ is $\PP$-realizable is \textsc{ExpTime}-hard and
    in \textsc{ExpTime} (in $n$ and polynomial in $m$). Moreover, if $\spec$ is $\PP$-realizable, it is
    $\PP$-realizable by a Mealy machine with $m+n_h2^{O(n log_2
      n)}$ states. Hardness holds even if $\PP$ has two states and $\atm$ is a deterministic reachability automaton.
\end{thm} 
\input{DecidingPRealizablity.tex}

\section{Synthesis from safety specifications and examples}\label{sec:learningframework}
In this section, we present the learning framework we use to synthesise Mealy machines from examples, and safety specifications.
Its generalization to any $\omega$-regular specification is described in Section~\ref{sec:omega-reg} and solved by reduction
to safety specifications.
It is a two-phase algorithm that is informally described here:(1) it tries to generalize the examples as much as possible while maintaining realizability of the specification, and outputs a preMealy machine, (2) it completes the preMealy machine into a full Mealy machine.

\subsection{Phase 1: Generalizing the examples}\label{subsec:gen}
\input{phase1.tex}

\subsection{Phase 2: completion of preMealy machines into Mealy machines}\label{subsec:comp}
As it only constructs the PTA and tries to merge its states, the generalization phase might not return a (complete) Mealy machine.
In other words, the machine it returns might still contain some holes (missing transitions).
The objective of this second phase is to complete those holes into a Mealy machine, while realizing the specification. 
More precisely, when a transition is not defined from some state $m$ and some input $\inp\in\inputs$, the algorithm must select an output symbol $\out\in \outputs$ and a state $m'$ to transition to, which can be either an existing state or a new state to be created (in that case, we write $m' = \textsf{fresh}$ to denote the fact that $m'$ is a fresh state).
In our implementation, if it is possible to reuse a state $m'$ that was created during the generalization phase, it is favoured over other states, in order to exploit the examples.
However, the algorithm for the completion phase we describe now does not depend on any particular strategy to pick states.
Therefore, it is parameterized by a \emph{completion strategy} $\sigma_C$, defined over all triples $(\MM, m, \inp, X)$ where $\MM$ is a preMealy machine with set of states $M$, $(m,\inp)$ is a hole of $\MM$, and $X\subseteq \outputs\times (M\cup \{\textsf{fresh}\})$ is a list of candidate pairs $(\out,m')$.
It returns an element of $X$, i.e., $\sigma_C(\MM,m,\inp,X)\in X$. 

In addition to $\sigma_C$, the completion algorithm takes as
input a preMealy
machine $\MM_0$ and a specification $\spec$, and outputs a
Mealy machine which $\MM_0$-realizes
$\spec$, if it exists. The pseudo-code
is given in Algo~\ref{algo:comp}. Initially, it tests whether $\spec$ is
$\MM_0$-realizable, otherwise it returns UNREAL. Then, it keeps on
completing holes of $\MM_0$. The computation of the list of
output/state candidates is done at the loop of
line~\ref{line:candidates}. Note that the \textbf{for}-loop iterates
over $M\cup\{\textsf{fresh}()\}$, where $\textsf{fresh}()$ is a
procedure that returns a fresh state not in $M$. The algorithm maintains the invariant that
at any iteration of the \textbf{while}-loop, $\spec$ is
$\MM$-realizable, thanks to the test at line~\ref{line:test2}, based on Theorem~\ref{thm:sizeSol}. Therefore, the list of candidates is necessarily
non-empty. Amongst those candidates, a single one is selected and the
transition on $(m,\inp)$ is added to $\MM$ accordingly at line~\ref{line:update}.

\begin{algorithm}[ht]
     \DontPrintSemicolon
    \KwIn{A preMealy machine $\MM_0 = (M, m_{\sf init}, \trans)$, a
      specification $\spec\subseteq (\inputs.\outputs)^*$ given as a
      deterministic safety automaton, a completion
      strategy $\sigma_C$}
    \KwOut{A (complete) Mealy machine $\MM$ such that $\spec$ is $\MM_0$-realizable, otherwise UNREAL. }

    \textbf{if} $\spec$ is not
    $\MM_0$-realizable\label{line:test1}\textbf{ then return } UNREAL

    $\MM\gets \MM_0$

    \While{there exists a hole $(m,\inp)\in M\times
      \inputs$}{
      
      $candidates \gets\varnothing$
      
      \For{$(\out,m')\in \outputs\times
        (M\cup\{\textsf{fresh}()\})$\label{line:candidates}}{

        \tcp*{$\textsf{fresh}()$ denotes a new
          state not in $M$}

        $\MM_{\out,m'} \gets (M\cup\{m'\},m_{\sf init}, \trans\cup\{(m,\inp)\mapsto
        (\out,m')\})$

        \If{$\spec$ is $\MM_{\out,m'}$-realizable\label{line:test2}}{
          $candidates\gets candidates\cup \{ (\out,m')\}$
        }
      }

      $(\out,m')\gets \sigma_C(\MM, m,\inp,candidates)$\label{line:selection}

      $(M,\trans)\gets (M\cup \{m'\},\trans\cup\{(m,\inp)\mapsto (\out,m')\})$\label{line:update}
      
      $\MM\gets (M, m_{\sf init}, \trans)$
      }
      \textbf{return} $\MM$

      \caption{\textsc{Comp}($\MM_0$,$\spec$,$\sigma_C$): preMealy machine completion algorithm}
      \label{algo:comp}
  \end{algorithm}

\subsection{Two-phase synthesis algorithm from specifications and examples}
\label{subsec:overview}
\input{phase2.tex}
\section{Synthesis from $\omega$-regular specifications and examples}
\input{SynthesisOmegaRegular.tex}

\section{Implementation and Case study}

We have implemented the algorithm \textsc{SynthLearn} of the previous section in a
prototype tool, in Python, using the tool
\textsc{Acacia-Bonzai}~\cite{DBLP:journals/corr/abs-2204-06079} to
manipulate antichains of counting functions. We first explain the heuristics we have used to define
state-merging and completion strategies, and then demonstrate how our
implementation behaves on a case study whose goal is to synthesize the
controller for an elevator. The interested reader can find in
App.~\ref{app:ebike} other case studies, including a controller for
an e-bike and two variations on mutual exclusion.
\subsection{Merging and completion strategies}
\input{MergeCompletionStrategies.tex}
\paragraph{Merging and completion strategies implemented in our
  prototype} Our tool implements a \emph{merging} strategy $\sigma_G$
where, given an example $e$ that leads in the current preMealy machine
to a state $m$ and a set $\{m_1,m_2, \dots, m_k\}$ of candidates for
merging, as computed in line~7 of Algorithm~\ref{algo:gen}, we choose
state $m_i$ with a $\preceq$-minimal counting function $F^*(m_i)$, as
defined in Lemma~\ref{lem:efficientpprealizability}. Intuitively,
favouring minimal counting functions preserves as much as possible the
set of behaviors that are possible after the example $e$.

Our tool also implements a \emph{completion strategy} $\sigma_C$, where for every hole $(m,\inp)$ of the preMealy machine $\MM$ and out of the list of candidate pairs, selects an element which again favour states associated with $\preceq$-minimal counting functions.

\subsection{Case Studies}

\paragraph{Lift Controller Example} We illustrate how to use our tool to construct a suitable controller for a two-floor
elevator system.


Considering two floors is sufficient enough to
illustrate most of the main difficulties of a more general elevator. Inputs of the controller are given by two  atomic propositions $\texttt{b0}$ and $\texttt{b1}$, which are true whenever the button at
floor 0 (resp. floor 1) is pressed by a user. Outputs are given by the
atomic propositions $\texttt{f0}$ and $\texttt{f1}$, true whenever the
elevator is at floor 0 (resp. floor 1); and 
$\texttt{ser}$, true whenever the elevator is \textit{serving} the
current floor (i.e. doors are opened). This controller should ensure
the following core properties:

\begin{enumerate}
    \item \textbf{Functional Guarantee:} whenever a button of floor 0
      (resp. floor 1) is pressed, the elevator must eventually
      \textit{serve} floor 0 (resp. floor 1): \begin{verbatim}
        G(b0 -> F (f0 & ser)) & G(b1 -> F (f1 & ser)) \end{verbatim}
    \item \textbf{Safety Guarantee:} The elevator is always at one
      floor exactly: \texttt{G(f0<->!f1)}
    \item \textbf{Safety Guarantee:} The elevator cannot transition between
      two floors when doors are opened: \texttt{G((f0 \& ser) -> X(!f1)) \& G((f1 \& ser) -> X(!f0))}
    \item \textbf{Initial State:} The elevator should be in floor 0 initially: \texttt{f0}
\end{enumerate}

Additionally, we make the following \textbf{assumption}: whenever a button of floor 0 (or floor 1) is pressed, it must remain pressed until the floor has been served, i.e., \texttt{G(b0 -> (b0 W (f0 \& ser))) \& G(b1 -> (b1 W (f1 \& ser))).}

Before going into the details of this example, let us explain the methodology that we apply to use our tool on this example. We start by providing only the high level specification $\varphi_{{\sf CORE}}$ for the elevator given above.
We obtain a first Mealy machine from the tool.
We then observe the machine to identify prefix of behaviours that we are unhappy with, and for which we can provide better alternative decisions.
Then we run the tool on $\varphi_{{\sf CORE}}$ and the examples that we have identified, and we get a new machine, and we proceed like that up to a point where we are satisfied with the synthesized Mealy machine.
\begin{figure}[t]
    \resizebox{\textwidth}{!}{
    \begin{tikzpicture}[->,>=stealth',shorten >=1pt,auto,node distance=3.8cm,
                    thick,inner sep=0pt,minimum size=0pt]
  \tikzstyle{every state}=[fill=gray!30,text=black,inner
  sep=2pt,minimum size=12pt]

        \node[state, initial] (q0) {$q_0$};
        \node[state, right of=q0] (q1) {$q_1$};
        \node[state, right of=q1, xshift=-4mm] (q2) {$q_2$};
        \node[state, right of=q2, xshift=-4mm] (q3) {$q_3$};
        \draw
            (q0) edge[loop above, align=center] node[] {$!\mathsf{b0}\;\&\;!\mathsf{b1}/\mathsf{f0}\;\&\;!\mathsf{f1}\;\&\;!\mathsf{ser}$\\$\mathsf{b0}\;\&\;!\mathsf{b1}/\mathsf{f0}\;\&\;!\mathsf{f1}\;\&\;\mathsf{ser}$} (q0)
            (q0) edge[above] node[xshift=1mm, yshift=-0.5mm]{$\mathsf{b1}/\mathsf{f0}\;\&\;!\mathsf{f1}\;\&\;!\mathsf{ser}$} (q1)
            (q1) edge[above] node[yshift=-0.5mm, xshift=-2mm]{$\mathsf{b1}/!\mathsf{f0}\;\&\;\mathsf{f1}\;\&\;\mathsf{ser}$} (q2)
            (q2) edge[loop above] node[]{$!\mathsf{b0}\;\&\;\mathsf{b1}/!\mathsf{f0}\;\&\;\mathsf{f1}\;\&\;\mathsf{ser}$} (q2)
            (q2) edge[bend right=20, above] node[yshift=0.3mm]{$!\mathsf{b0}\;\&\;!\mathsf{b1}/!\mathsf{f0}\;\&\;\mathsf{f1}\;\&\;!\mathsf{ser}$} (q0)
            (q2) edge[above]
            node[yshift=-0.5mm]{$\mathsf{b0}/!\mathsf{f0}\;\&\;\mathsf{f1}\;\&\;!\mathsf{ser}$} (q3)
            (q3) edge[bend left=12, above] node[]{$\mathsf{b0}/\mathsf{f0}\;\&\;!\mathsf{f1}\;\&\;\mathsf{ser}$} (q0);
    \end{tikzpicture}}
    \caption{Machine returned by our tool on the elevator
      specification w/o examples. Here, $q0$ represents the state where \texttt{f0} is served when required, $q1$ represents the state where \texttt{b1} is pending, $q2$ represents state where \texttt{f1} is served, $q3$ represents the state where \texttt{b0} is pending.\label{fig:elevatorcontrollernoexamples}}
\end{figure}
\begin{figure}[t]
\resizebox{\textwidth}{!}{
    \begin{tikzpicture}[->,>=stealth',shorten >=1pt,auto,node distance=3.8cm,
                    thick,inner sep=0pt,minimum size=0pt]
  \tikzstyle{every state}=[fill=gray!30,text=black,inner
  sep=2pt,minimum size=12pt]

        \node[state, initial,fill=red!30] (q0) {$q_0$};
        \node[state, right of=q0, xshift=2cm,fill=red!30] (q1) {$q_1$};
        \node[state, right of=q1, xshift=1cm] (q2) {$q_2$};
        \node[state, below right of=q0, yshift=1.5cm] (q3) {$q_3$};
        \draw
            (q0) edge[loop above, align=center, color=red!70] node[yshift=1mm, color=red!70] {$!\mathsf{b0}\;\&\;!\mathsf{b1}/\mathsf{f0}\;\&\;!\mathsf{f1}\;\&\;!\mathsf{ser}$\\$\mathsf{b0}\;\&\;!\mathsf{b1}/\mathsf{f0}\;\&\;!\mathsf{f1}\;\&\;\mathsf{ser}$} (q0)
            (q0) edge[bend right=20, above, color=red!70] node[yshift=1.5mm, color=red!70]{$!\mathsf{b0}\;\&\;\mathsf{b1}/\mathsf{f0}\;\&\;!\mathsf{f1}\;\&\;!\mathsf{ser}$} (q1)
            (q0) edge[left] node[yshift=-2mm, xshift=1mm]{$\mathsf{b0}\;\&\;\mathsf{b1}/\mathsf{f0}\;\&\;!\mathsf{f1}\;\&\;\mathsf{ser}$} (q3)
            (q1) edge[loop above, align=center, color=red!70] node[color=red!70]{$!\mathsf{b0}\;\&\;!\mathsf{b1}/!\mathsf{f0}\;\&\;\mathsf{f1}\;\&\;!\mathsf{ser}$\\!$\mathsf{b0}\;\&\;\mathsf{b1}/!\mathsf{f0}\;\&\;\mathsf{f1}\;\&\;\mathsf{ser}$} (q1)
            (q1) edge[bend right=20, below, color=red!70] node[yshift=-1.7mm, color=red!70]{$\mathsf{b0}\;\&\;!\mathsf{b1}/!\mathsf{f0}\;\&\;\mathsf{f1}\;\&\;!\mathsf{ser}$} (q0)
            (q1) edge[above]
            node[yshift=1mm]{$\mathsf{b0}\;\&\;\mathsf{b1}/!\mathsf{f0}\;\&\;\mathsf{f1}\;\&\;\mathsf{ser}$} (q2)
            (q2) edge[bend right=30, above] node[yshift=2mm]{$\mathsf{b0}/!\mathsf{f0}\;\&\;\mathsf{f1}\;\&\;!\mathsf{ser}$} (q0)
            (q3) edge[right] node[yshift=-2mm, xshift=-1mm]{$\mathsf{b1}/!\mathsf{f0}\;\&\;\mathsf{f1}\;\&\;!\mathsf{ser}$} (q1);
    \end{tikzpicture}}
    \caption{Mealy machine returned by our tool on the elevator specification with additional examples. The preMealy machine obtained after generalizing the examples and before completion is highlighted in red. This took 3.10s to be generated. \label{fig:elevatorcontroller}}
\end{figure}

Let us now give details. When our tool is provided with this specification without any examples, we get the machine depicted in \cref{fig:elevatorcontrollernoexamples}.
This solution makes the controller switch between floor 0 and floor 1, sometimes unnecessarily.
For instance, consider the trace \texttt{s \# \{!b0 \& !b1\}\{!f0 \& f1 \& !ser\} \# \{!b0 \& !b1\}\{f0 \& !f1 \& !ser\}}, where we let \texttt{s = \{!b0 \& b1\}\{f0 \& !f1 \& !ser\} \# \{!b0 \& b1\}\{!f0 \& f1 \& ser\}}.
Here, we note that the transition goes back to state $q_0$, where the elevator is at floor 0, when the elevator could have remained at floor 1 after serving floor 1.
The methodology described above allows us to identify the following three examples:
\begin{enumerate}
\item The 1st trace states that after serving floor 1, the elevator must remain at floor 1 as \texttt{b0} is false: \texttt{s \# \{!b0 \& !b1\}\{!f0 \& f1 \& !ser\} \# \{!b0 \& !b1\}\{!f0 \& f1 \& !ser\}}
\item The 2nd trace states that the elevator must remain at floor 0, as \texttt{b1} is false: \texttt{\{!b0 \& !b1\}\{f0 \& !f1 \& !ser\} \# \{!b0 \& !b1\}\{f0 \& !f1 \& !ser\}}
\item The 3rd trace ensures that after \texttt{s}, there is no unnecessary delay in serving floor 0 after floor 1 is served in \texttt{s}: \texttt{s \# \{b0 \& !b1\}\{!f0 \& f1 \& !ser\} \# \{b0 \& !b1\}\{f0 \& !f1 \& ser\}}
\end{enumerate}

\noindent With those additional examples, our tool outputs the machine of \cref{fig:elevatorcontroller}, which generalizes them and now ensures that moves of the elevator occur only when required. For example, the end of the first trace has been generalized into a loop on state $q_1$ ensuring that the elevator does not go to floor $0$ from floor $1$ unless \texttt{b0} is pressed.
We note that the number of examples provided here is much smaller than the theoretical (polynomial) upper bound proved in Theorem~\ref{thm:mealycompleteness}.

\section{Conclusion} In this paper, we have introduced the problem of {\em synthesis with a
  few hints}. This variant of the synthesis problem allows the user to
guide synthesis using examples of expected executions of
high quality solutions. Existing synthesis tools may not provide
natural solutions when fed with high-level specifications only, and as
providing complete specification goes against the very goal of
synthesis, we believe that our algorithm has a greater potential
in practice.

On the theoretical side, we have studied in details the
computational complexity of problems that need to be solved during our
new synthesis procedure. We have proved that our algorithm is
{\em complete} in the sense that any Mealy machine $\MM$ that realizes
a specification $\varphi$ can be obtained by our algorithm from
$\varphi$ and a sufficiently rich example set $E$, whose size is bounded polynomially in the size of $\MM$. On the practical side, we have implemented our
algorithm in a prototype tool that extends
Acacia-Bonzai~\cite{DBLP:journals/corr/abs-2204-06079} with tailored
state-merging learning algorithms. We have shown that only a small number
of examples are necessary to obtain high quality machines from high-level LTL specifications only. The tool is not fully optimized yet. While this is sufficient to demonstrate the relevance of our approach, we will work
on efficiency aspects of the implementation. 

As future works, we will consider extensions of the user interface to interactively and concisely specify sets of (counter-)examples to
solutions output by the tool. In the same line, an interesting future direction is to handle
parametric examples (e.g. elevator with the number of floors given as parameter). This would require to provide a concise syntax to define parametric
examples and to design efficient synthesis algorithm in this setting.  We will also consider the possibility
to formulate negative examples, as our theoretical results readily
extend to this case and their integration in the implementation should
be easy.

\bibliographystyle{splncs04}
\bibliography{main}

\newpage
\appendix

\section{Additonal examples}\label{app:ebike}

In this appendix, we provided the interested reader with three additional examples.

\subsection{Electric Bike Example}

Here, we aim to synthesize a Mealy machine for the controller of an electric
bike, in charge of regulating the braking system as well as the
e-assistance. Its inputs are the following atomic propositions:
\textsf{brake}, which is true whenever the cyclist activates the
handbrake; \textsf{full}, true when the battery sensor indicates the battery
is fully charged; and \textsf{speedy}, true whenever the bike speed is
above 25 km/h. Its outputs are the following atomic propositions:
\textsf{rim}, which is set to true whenever the rim brake is
activated, \textsf{recharge}, true whenever the motor brake is
activated and recharging the battery, and \textsf{assist}, whenever
the motor is assisting the cyclist. 
This controller should ensure the properties:
\begin{enumerate}
  \item whenever the battery is full, it cannot be recharged:

    \begin{verbatim}
G(full -> !recharge)
\end{verbatim}

  \item when the cyclist does not brake, none of the braking system is
    activated:

\begin{verbatim}
G(!brake -> (!rim & !recharge))
\end{verbatim}

    \item whenever the speed of the bike is above 25 km/h, the
      assistance is inactive:
\begin{verbatim}
G(speedy -> !assist)
\end{verbatim}

    \item if the cyclist brakes for at least three cycles, then one of
      the two braking systems should be active until the cyclist does
      not brake anymore:

                  \begin{verbatim}
G((brake & X brake & XX brake) -> XX((recharge | rim) W !brake)) 
\end{verbatim}

    \item  if the cyclist does not brake for at least three cycles,
      then the assistance should be active unless the speed is above
      the limit, and until she brakes again

\begin{verbatim}
G((!brake & X(!brake) & XX(!brake)) -> ((XX((!speedy -> assist) 
                                        W brake)))) 
\end{verbatim}

      \item whenever the motor status changes, it should be idle for
        at least one cycle:
            \begin{verbatim}
G(recharge -> X(!assist)) & G(assist -> X(!recharge))
\end{verbatim}
        \item assistance and brakes are mutually exclusive:
\begin{verbatim}
G((recharge | rim) -> !assist)
\end{verbatim}

\end{enumerate}

\begin{figure}[h]
    \centering
    \begin{tikzpicture}
        \node[state, initial] (q0) {$q_0$};
        \draw
            (q0) edge[loop above, align=center] node{$\mathsf{brk}/!\mathsf{as}\;\&\;!\mathsf{re}\;\&\;\mathsf{ri}$ \\ $!\mathsf{brk}\;\&\;!\mathsf{spd}/\mathsf{as}\;\&\;!\mathsf{re}\;\&\;!\mathsf{ri}$ \\
            $!\mathsf{brk}\;\&\;\mathsf{spd}/\mathsf{idle}$} (q0);
    \end{tikzpicture}
    \caption{Preliminary machine obtained by our tool (without additional examples) and Strix, on the e-bike specification.}
    \label{fig:notraces}
\end{figure}

When provided with this specification without any example to our tool,
or to Strix, we get a solution which never recharges the battery. It
has a single state on which it loops with the labels\footnote{For the sake of readability, we have replaced: 
\begin{itemize}
    \item the output edge label $!\mathsf{as}\;\&\;!\mathsf{re}\;\&\;!\mathsf{ri}$ with the term $\mathsf{idle}$
    \item the labels $\mathsf{assist}$ with $\mathsf{as}$, $\mathsf{recharge}$ with $\mathsf{re}$, $\mathsf{rim}$ with $\mathsf{ri}$, $\mathsf{brake}$ with $\mathsf{br}$ and $\mathsf{full}$ with $\mathsf{ful}$.
\end{itemize}}
$\mathsf{brk}/!\mathsf{as}\;\&\;!\mathsf{re}\;\&\;\mathsf{ri}$,
$!\mathsf{brk}\;\&\;!\mathsf{spd}/\mathsf{as}\;\&\;!\mathsf{re}\;\&\;!\mathsf{ri}$,
and $!\mathsf{brk}\;\&\;\mathsf{spd}/\mathsf{idle}$. It is depicted in 
Fig.~\ref{fig:notraces}. To obtain a better machine, without specifying formally when exactly the battery should or should not be recharged, we 
provide the next two simple scenarios to our tool (that are counter-examples to the first machine):
\begin{enumerate}
  \item The first scenario describes a trace where, when the cyclist activates the handbrake (\texttt{brk} becomes true), and the battery is not fully charged (\texttt{ful} is false), then the motor brake is activated to recharge the battery while braking (\texttt{re} is set to true): \begin{verbatim}  {!brk,spd}{!as,!re,!ri} # {brk,!ful}{!as,re,!ri}
  # {brk,!ful}{!as,re,!ri} \end{verbatim}
  \item The second scenario describes a trace where, when the cyclist activates the handbrake (\texttt{brk} becomes true), and the battery is fully charged (\texttt{ful} is true), then the rim brake is activated (\texttt{ri} is set to true):\begin{verbatim}  {brk,ful}{!as,!re,ri} \end{verbatim}
\end{enumerate}
Note that, unlike outputs, inputs are not complete in those two examples: the first example does not specify whether at the first step, \texttt{ful} is true or not.
Similarly, the notation \texttt{\{brk,!ful\}} does not specify whether \texttt{spd} is true or not.
The set notation here is a syntax provided by the tool allowing to specify non-maximal sets.
So, the first trace actually corresponds to $8$ examples, because there are 3 possible Boolean symbols (\texttt{brk}, \texttt{spd} and \texttt{ful}), and the second trace to $2$ examples.
This offers to the user a way to compactly represent several examples at once, and to focus on relevant Boolean signals she wants to provide as input to the tool.

With that additional information, our tool outputs the machine of Fig.~\ref{fig:machinebike} which recharges the battery when braking and whenever it is possible. Here, we note that the first trace has been generalized to the following: whenever the cyclist activates the handbrake and the batteries are not fully charged, the controller uses the motor brake and recharges the batteries. Likewise, the second trace has been generalized to: whenever the cyclist activates the handbrake and the batteries are fully charged, then the controller uses the rim brake.

\begin{figure}[t]
\resizebox{\textwidth}{!}{
    \begin{tikzpicture}[->,>=stealth',shorten >=1pt,auto,node distance=3.8cm,
                    thick,inner sep=0pt,minimum size=0pt]
  \tikzstyle{every state}=[fill=gray!30,text=black,inner
  sep=2pt,minimum size=12pt]

        \node[state, initial,fill=red!30] (q0) {$q_0$};
        \node[state, right of=q0, xshift=1.5cm] (q1) {$q_1$};
        \node[state, left of=q0, xshift=-1.5cm,fill=red!30] (q2) {$q_2$};
        \draw
            (q0) edge[loop above, align=center] node[color=red!70]{$!\mathsf{brk}\;\&\;\mathsf{spd}/\mathsf{idle}$\\$\mathsf{brk}\;\&\;\mathsf{ful}/!\mathsf{as}\;\&\;!\mathsf{re}\;\&\;\mathsf{ri}$} (q0)
            (q0) edge[bend right=20, above] node[yshift=2mm]{$!\mathsf{brk}\;\&\;!\mathsf{spd}/\mathsf{as}\;\&\;!\mathsf{re}\;\&\;!\mathsf{ri}$} (q1)
            (q0) edge[bend right=20, below,color=red!70] node[yshift=-2mm,color=red!70]{$\mathsf{brk}\;\&\;!\mathsf{ful}/!\mathsf{as}\;\&\;\mathsf{re}\;\&\;!\mathsf{ri}$} (q2)
            (q1) edge[loop above] node{$!\mathsf{brk}\;\&\;!\mathsf{spd}/\mathsf{as}\;\&\;!\mathsf{re}\;\&\;!\mathsf{ri}$} (q1)
            (q1) edge[bend right=20, below] node[yshift=-1mm]{$!\mathsf{brk}\;\&\;!\mathsf{spd}/\mathsf{idle}$} (q0)
            (q1) edge[bend left=20, above]
            node[yshift=1mm]{$\mathsf{brk}/\mathsf{idle}$} node[yshift=-3pt]{$\prec$} (q2)
            (q2) edge[loop above, align=center] node{\textcolor{red!70}{$\mathsf{brk}\;\&\;!\mathsf{ful}/!\mathsf{as}\;\&\;\mathsf{re}\;\&\;!\mathsf{ri}$}\\$\mathsf{brk}\;\&\;\mathsf{ful}/!\mathsf{as}\;\&\;!\mathsf{re}\;\&\;\mathsf{ri}$} (q2)
            (q2) edge[bend right=20, above] node[yshift=1mm]{$!\mathsf{brk}/\mathsf{idle}$} (q0);
    \end{tikzpicture}}
    \caption{Mealy machine returned by our tool on the e-bike
      specification with examples provided by the user as explained. The preMealy machine
      obtained after generalizing the examples and before completion
      is highlighted in red. Remember that our algorithm tries to
      reuse as much as possible states that were created during the generalizing phase, and as a consequence on this example the completion
      phase creates one additional state only.\label{fig:machinebike}}
\end{figure}

\subsection{Mututal Exclusion with a prioritized process (Prioritized Arbiter)}\label{app:mutexprio}

Here, we aim to synthesize a Mealy machine for mutual exclusion with three processes: \texttt{0}, \texttt{1} and \texttt{m} and process \texttt{m} is prioritized.
This could be useful in situations where we would like to prioritize one process over the others, i.e., have a \textit{master} process.
Here is the high level specification for this system:

\begin{enumerate}
    \item \texttt{G(request\_0 -> F grant\_0)} 
    \item \texttt{G(request\_1 -> F grant\_1)}
    \item \texttt{G(request\_m ->  grant\_m)}
    \item \texttt{G(($!$grant\_0 $\land$ $!$grant\_1)$|$($!$grant\_0 $\land$ $!$grant\_m)$|$($!$grant\_m $\land$ $!$grant\_1))}
\end{enumerate}

The above formulas correspond to the typical mutual exclusion specification with the additional constraint that the \textit{master} request must be granted immediately.

Providing these specifications to Strix results in the machine depicted in \cref{fig:StrixPrioritized}(\texttt{request} and \texttt{grant} have been abbreviated into \texttt{r} and \texttt{g}).

We note here that the requests of the non-prioritized processes are not taken into account, thereby making this an unviable solution.
Our tool provides a machine depicted in \cref{fig:PrioritizedNoExamples} which is slightly better, but still has traces which provide unsolicited grants.
One such example of a trace would be \texttt{\{rm \& !r0 \& r1\}\{!g0 \& !g1 \& gm\} \# \{!rm \& !r0 \& !r1\}\{g0 \& !g1 \& gm\}}, where process 1 granted access but was never requested.

To obtain a satisfactory solution, we we additionally provide the following scenarios of executions.
We start with traces which follow the pattern of two requests in the first step and no requests in the second step. We resolve this trace by granting one process in the first step and the second process in the second step.:
\begin{enumerate}
    \item \begin{verbatim}{!rm & r0 & r1}{g0 & !g1 & !gm} # {!rm & !r0 & !r1}
                                  {!g0 & g1 & !gm} \end{verbatim}
    \item \begin{verbatim}{rm & !r0 & r1}{!g0 & !g1 & gm} # {!rm & !r0 & !r1}
                                  {!g0 & g1 & !gm}\end{verbatim}
\end{enumerate}
Now, we handle one example of a trace where all three process are requested at once: 
\begin{enumerate}
    \item \begin{verbatim}{rm & r0 & r1}{!g0 & !g1 & gm} # {!rm & !r0 & !r1}
                                 {g0 & !g1 & !gm} \end{verbatim}
\end{enumerate}
Finally, we examine traces where both \texttt{rm} and \texttt{r0} are requested at once:
\begin{enumerate}
    \item \begin{verbatim}{rm & r0 & !r1}{!g0 & !g1 & gm} # {!rm & r0 & r1}{g0 & !g1 & !gm}
 # {!rm & !r0 & !r1}{!g0 & g1 & !gm} \end{verbatim}
    \item \begin{verbatim}{rm & r0 & !r1}{!g0 & !g1 & gm} # {!rm & !r0 & !r1}{g0 & !g1 & !gm} 
 # {!rm & r0 & !r1}{g0 & !g1 & !gm} \end{verbatim}
    \item \begin{verbatim}{rm & r0 & !r1}{!g0 & !g1 & gm} # {!rm & !r0 & r1}{g0 & !g1 & !gm} 
 # {!rm & !r0 & !r1}{!g0 & g1 & !gm} \end{verbatim}
\end{enumerate}

\begin{figure}[t]
    \centering
    \begin{tikzpicture}
        \node[state, initial] (q0) {$q_0$};
        \node[state, right of=q0, xshift=3cm] (q1) {$q_1$};
        \draw
            (q0) edge[loop above] node{$\mathsf{r_m}/!\mathsf{g_0}\;\&\;!\mathsf{g_1}\;\&\;\mathsf{g_m}$} (q0)
            (q0) edge[bend right, below] node{$!\mathsf{r_m}/\mathsf{g_0}\;\&\;!\mathsf{g_1}\;\&\;!\mathsf{g_m}$} (q1)
            (q1) edge[loop above] node{$\mathsf{r_m}/!\mathsf{g_0}\;\&\;!\mathsf{g_1}\;\&\;\mathsf{g_m}$} (q1)
            (q1) edge[bend right, above] node{$!\mathsf{r_m}/\mathsf{g_0}\;\&\;!\mathsf{g_1}\;\&\;!\mathsf{g_m}$} (q0);
    \end{tikzpicture}
    \caption{Mealy Machine for Prioritized Arbiter returned by Strix}
    \label{fig:StrixPrioritized}
\end{figure}

\begin{figure}[t]
    \centering
    \begin{tikzpicture}
        \node[state, initial] (q0) {$q_0$};
        \node[state, below of=q0] (q1) {$q_1$};
        \node[state, below right of=q1] (q2) {$q_2$};
        \node[state, below left of=q1] (q3) {$q_3$};
        \draw
            (q0) edge[loop above, align=center] node{$!\mathsf{r_0}\;\&\;!\mathsf{r_1}\;\&\;!\mathsf{r_m}/!\mathsf{g_0}\;\&\;!\mathsf{g_1}\;\&\;!\mathsf{g_m}$\\$!\mathsf{r_0}\;\&\;!\mathsf{r_1}\;\&\;\mathsf{r_m}/!\mathsf{g_0}\;\&\;!\mathsf{g_1}\;\&\;\mathsf{g_m}$\\$\mathsf{r_0}\;\&\;!\mathsf{r_1}\;\&\;!\mathsf{r_m}/\mathsf{g_0}\;\&\;!\mathsf{g_1}\;\&\;!\mathsf{g_m}$\\$!\mathsf{r_0}\;\&\;\mathsf{r_1}\;\&\;!\mathsf{r_m}/!\mathsf{g_0}\;\&\;\mathsf{g_1}\;\&\;!\mathsf{g_m}$} (q0)
            (q0) edge[left] node{$(\mathsf{r_0}\;\&\;\mathsf{r_1})|(\mathsf{r_0}\;\&\;\mathsf{r_m})|(\mathsf{r_m}\;\&\;\mathsf{r_1})/!\mathsf{g_0}\;\&\;!\mathsf{g_1}\;\&\;\mathsf{g_m}$} (q1)
            (q1) edge[loop left, left] node{$\mathsf{r_m}/!\mathsf{g_0}\;\&\;!\mathsf{g_1}\;\&\;\mathsf{g_m}$} (q1)
            (q1) edge[bend right, left] node{$!\mathsf{r_m}/!\mathsf{g_0}\;\&\;\mathsf{g_1}\;\&\;!\mathsf{g_m}$} (q2)
            (q2) edge[bend right, right] node{$!\mathsf{r_m}\;\&\;\mathsf{r_1}/!\mathsf{g_0}\;\&\;\mathsf{g_1}\;\&\;!\mathsf{g_m}$} (q1)
            (q2) edge[bend right, right] node{$!\mathsf{r_m}\;\&\;!\mathsf{r_1}/\mathsf{g_0}\;\&\;!\mathsf{g_1}\;\&\;!\mathsf{g_m}$} (q0)
            (q2) edge[below] node{$\mathsf{r_m}\;\&\;\mathsf{r_1}/!\mathsf{g_0}\;\&\;!\mathsf{g_1}\;\&\;\mathsf{g_m}$} (q3)
            (q3) edge[loop below] node{$\mathsf{r_m}/!\mathsf{g_0}\;\&\;!\mathsf{g_1}\;\&\;\mathsf{g_m}$} (q3)
            (q3) edge[left] node{$!\mathsf{r_m}/\mathsf{g_0}\;\&\;!\mathsf{g_1}\;\&\;!\mathsf{g_m}$} (q1);
    \end{tikzpicture}
    \caption{Preliminary machine obtained by our tool (without additional examples) on the Prioritized Arbiter specification}
    \label{fig:PrioritizedNoExamples}
\end{figure}

We then obtain the machine in \cref{fig:PrioritizedTarget}.
We note that there are no spurious grants and the order of requests of process \texttt{0} and \texttt{1} are noted and respected. 
This ensures fairness amongst the non-prioritized processes.

\begin{figure}[!p]
\hspace{-2.8cm}
    \begin{tikzpicture}[->,>=stealth',shorten >=1pt,auto,node distance=5.3cm,
                    thick,inner sep=0pt,minimum size=0pt]
  \tikzstyle{every state}=[fill=gray!30,text=black,inner
  sep=2pt,minimum size=12pt]
        \node[state, initial, fill=red!30] (q0) {$q_0$};
        \node[state, below left of=q0, fill=red!30] (q1) {$q_1$};
        \node[state, below right of=q0, fill=red!30] (q2) {$q_2$};
        \node[state, below of=q1] (q3) {$q_3$};
        \node[state, below of=q2] (q4) {$q_4$};
        \draw
            (q0) edge[loop above, align=center, color=red!70] node{$!\mathsf{r_0}\;\&\;!\mathsf{r_1}\;\&\;\mathsf{r_m}/!\mathsf{g_0}\;\&\;!\mathsf{g_1}\;\&\;\mathsf{g_m}$\\$\mathsf{r_0}\;\&\;!\mathsf{r_1}\;\&\;!\mathsf{r_m}/\mathsf{g_0}\;\&\;!\mathsf{g_1}\;\&\;!\mathsf{g_m}$\\$!\mathsf{r_0}\;\&\;\mathsf{r_1}\;\&\;!\mathsf{r_m}/!\mathsf{g_0}\;\&\;\mathsf{g_1}\;\&\;!\mathsf{g_m}$} (q0)
            (q0) edge[bend right, left, align=center, color=red!70] node[above, sloped]{$(\mathsf{r_0}\;\&\;!\mathsf{r_1}\;\&\;\mathsf{r_m})/!\mathsf{g_0}\;\&\;!\mathsf{g_1}\;\&\;\mathsf{g_m}$\\$(\mathsf{r_0}\;\&\;\mathsf{r_1}\;\&\;!\mathsf{r_m})/!\mathsf{g_0}\;\&\;\mathsf{g_1}\;\&\;!\mathsf{g_m}$} (q1)
            (q0) edge[bend right=100, looseness=1.5, left, color=red!70] node[pos=0.4, rotate=30,sloped, xshift=-1cm, yshift=1.2cm]{$\mathsf{r_0}\;\&\;\mathsf{r_1}\;\&\;\mathsf{r_m}/!\mathsf{g_0}\;\&\;!\mathsf{g_1}\;\&\;\mathsf{g_m}$} (q3)
            (q0) edge[bend right, above, color=red!70] node[above, sloped, yshift=2mm]{$(\mathsf{r_1}\;\&\;\mathsf{r_m})/!\mathsf{g_0}\;\&\;!\mathsf{g_1}\;\&\;\mathsf{g_m}$} (q2)
            (q1) edge[loop left, left] node[above, sloped]{$\mathsf{r_m}\;\&\;!\mathsf{r_1}/!\mathsf{g_0}\;\&\;!\mathsf{g_1}\;\&\;\mathsf{g_m}$} (q1)
            (q2) edge[loop right, right] node[above, sloped]{$\mathsf{r_m}\;\&\;!\mathsf{r_0}/!\mathsf{g_0}\;\&\;!\mathsf{g_1}\;\&\;\mathsf{g_m}$} (q2)
            (q1) edge[bend right, above] node[yshift=2mm]{$!\mathsf{r_m}\;\&\;\mathsf{r_1}/\mathsf{g_0}\;\&\;!\mathsf{g_1}\;\&\;!\mathsf{g_m}$} (q2)
            (q1) edge[bend right, left, color=red!70] node[above, sloped, yshift=3mm]{$!\mathsf{r_m}\;\&\;!\mathsf{r_1}/!\mathsf{g_0}\;\&\;\mathsf{g_1}\;\&\;!\mathsf{g_m}$} (q0)
            (q2) edge[bend right, below, color=red!70] node[yshift=-2mm]{$!\mathsf{r_m}\;\&\;\mathsf{r_0}/!\mathsf{g_0}\;\&\;\mathsf{g_1}\;\&\;!\mathsf{g_m}$} (q1)
            (q2) edge[bend right, right, color=red!70] node[above, sloped]{$!\mathsf{r_m}\;\&\;!\mathsf{r_0}/\mathsf{g_0}\;\&\;!\mathsf{g_1}\;\&\;!\mathsf{g_m}$} (q0)
            (q3) edge[left, color=red!70] node[above, sloped, xshift=-1.5cm]{$!\mathsf{r_m}/\mathsf{g_0}\;\&\;!\mathsf{g_1}\;\&\;!\mathsf{g_m}$} (q2)
            (q4) edge[right] node[above, sloped, xshift=1.5cm]{$!\mathsf{r_m}/!\mathsf{g_0}\;\&\;\mathsf{g_1}\;\&\;!\mathsf{g_m}$} (q1)
            (q3) edge[loop below] node{$\mathsf{r_m}/!\mathsf{g_0}\;\&\;!\mathsf{g_1}\;\&\;\mathsf{g_m}$} (q3)
            (q4) edge[loop below] node{$\mathsf{r_m}/!\mathsf{g_0}\;\&\;!\mathsf{g_1}\;\&\;\mathsf{g_m}$} (q4)
            (q1) edge[left] node[above, sloped]{$\mathsf{r_m}\;\&\;\mathsf{r_1}/!\mathsf{g_0}\;\&\;!\mathsf{g_1}\;\&\;\mathsf{g_m}$} (q3)
            (q2) edge[right] node[above, sloped]{$\mathsf{r_m}\;\&\;\mathsf{r_0}/\mathsf{g_0}\;\&\;!\mathsf{g_1}\;\&\;!\mathsf{g_m}$} (q4);
    \end{tikzpicture}
    \caption{Mealy machine returned by our tool on the prioritized arbiter specification with additional examples. The preMealy machine obtained after generalizing the examples and before completion is highlighted in red. In state $q0$, there are no pending requests.
    In state $q1$, $r_0$ is pending and correspondingly in state $q2$, $r_1$ is pending.
    In states $q3$ and $q4$, both requests $r_0$ and $r_1$ are pending, but however, the order in which the requests are granted matter.
    In $q3$, we grant $r_0$ and then $r_1$ and in $q4$, we grant $r_1$ and then $r_0$.}
    \label{fig:PrioritizedTarget}
\end{figure}

\newpage

\subsection{Mutual Exclusion with no subsequent grants}\label{app:mutextex}
Here, we aim to synthesize a Mealy machine for mutual exclusion with additional property that the implementation should never grant twice in a row. This could be useful in situations where we would like to give the granting process a bit of a \textit{break} so as to execute other instructions. Here is the high-level specification for this system: 
\begin{enumerate}
    \item \texttt{G(request\_0 -> F grant\_0)} 
    \item \texttt{G(request\_1 -> F grant\_1)}
    \item \texttt{G($!$grant\_0 $|$ $!$grant\_1)}
    \item \texttt{G(grant\_0 -> X ($!$grant\_1\;\&\;$!$grant\_0))}
    \item \texttt{G(grant\_1 -> X ($!$grant\_1\;\&\;$!$grant\_0))}
\end{enumerate}
The first three formulas correspond to a typical mutual exclusion problem and the last two specify that no two subsequent grants can take place.

Providing these high-level specification to Strix returns the machine in \cref{fig:StrixAlternatingGrants} (\texttt{request} and \texttt{grant} have been abbreviated into \texttt{r} and \texttt{g}). 
\begin{figure}[h]
    \centering
    \begin{tikzpicture}
        \node[state, initial] (q0) {$q_0$};
        \node[state, left of=q0, xshift=-3cm] (q1) {$q_1$};
        \node[state, below of=q1] (q2) {$q_2$};
        \node[state, below of=q0] (q3) {$q_3$};
        \draw
            (q0) edge[above] node{$\mathsf{true}/\mathsf{g_0}\;\&\;!\mathsf{g_1}$} (q1)
            (q1) edge[left] node{$\mathsf{true}/!\mathsf{g_0}\;\&\;!\mathsf{g_1}$} (q2)
            (q2) edge[below] node{$\mathsf{true}/!\mathsf{g_0}\;\&\;\mathsf{g_1}$} (q3)
            (q3) edge[right] node{$\mathsf{true}/!\mathsf{g_0}\;\&\;!\mathsf{g_1}$} (q0);
    \end{tikzpicture}
    \caption{Mealy Machine for No Subsequent Grants returned by Strix}
    \label{fig:StrixAlternatingGrants}
\end{figure}
As for the example in Introduction, we can see that the solution proposed by Strix does not take into account the requests, so it cannot be considered as \tocheck{an efficient solution} to our problem. Our tool provides a similar machine when no examples are given.

To obtain a satisfactory solution, we additionally provide the following scenarios of executions. First, let
\begin{verbatim}
 s = {!r_0 & !r_1}{!g_0 & !g_1} # {r_0 & !r_1}{g_0 & !g_1}
\end{verbatim}
Then, we complete this simple scenario by the following five different continuations that exhibit relevant reactions of \tocheck{efficient} solutions to the problem:

\begin{itemize}
    \item this scenario asks to favour $r\_1$ when two requests are made after $s$ (in which process $0$ has been granted):
        \begin{verbatim}s # {!r0}{!g_0 & !g_1} # {r_0 & r_1}{!g_0 & g_1} \end{verbatim}


\item even if $r\_0$ has been granted in $s$, if the first process after $s$ making a request is process $0$, then it should have priority over the other:
\begin{verbatim}s # {r_0 & !r_1}{!g_0 & !g_1} # {true}{g_0 & !g_1} \end{verbatim}


\item however, if both processes make a request simultaneously after $s$, then process $1$ gets the priority:
\begin{verbatim}s # {r_0 & r_1}{!g_0 & !g_1} # {!r_0 & !r_1}{!g_0 & g_1} \end{verbatim}


\end{itemize}

\begin{figure}[!p]
    \centering
    \begin{tikzpicture}[->,>=stealth',shorten >=1pt,auto,node distance=2cm,
                    thick,inner sep=0pt,minimum size=0pt]
  \tikzstyle{every state}=[fill=gray!30,text=black,inner
  sep=2pt,minimum size=12pt]
        \node[state, initial] (q0) {$q_0$};
        \node[state, below of=q0] (q1) {$q_1$};
        \node[state, below right of=q1, xshift=3cm] (q2) {$q_2$};
        \node[state, below left of=q1, xshift=-3cm] (q3) {$q_3$};
        \node[state, below of=q1, yshift=-1cm] (q4) {$q_4$};
        \node[state, below of=q4, yshift=-1cm] (q7) {$q_7$};
        \node[state, below of=q7, yshift=-1cm] (q5) {$q_5$};
        \node[state, below of=q5, yshift=-1cm] (q6) {$q_6$};
        \draw
            (q0) edge[loop above] node{$!\mathsf{r\_0}\;\&\;!\mathsf{r\_1}/!\mathsf{g\_0}\;\&\;!\mathsf{g\_1}$} (q0)
            (q0) edge[bend left=20, right, align=center] node{$\mathsf{r\_0}\;\&\;!\mathsf{r\_1}/\mathsf{g\_0}\;\&\;!\mathsf{g\_1}$ \\ $!\mathsf{r\_0}\;\&\;\mathsf{r\_1}/!\mathsf{g\_0}\;\&\;\mathsf{g\_1}$} (q1)
            (q1) edge[bend left=20, left, ] node{$!\mathsf{r\_0}\;\&\;!\mathsf{r\_1}/!\mathsf{g\_0}\;\&\;!\mathsf{g\_1}$} (q0)
            (q0) edge[bend left=80, right] node{$\mathsf{r\_0}\;\&\;\mathsf{r\_1}/\mathsf{g\_0}\;\&\;!\mathsf{g\_1}$} (q6)
            (q1) edge[bend left=20, right] node{$!\mathsf{r\_0}\;\&\;\mathsf{r\_1}/!\mathsf{g\_0}\;\&\;!\mathsf{g\_1}$} (q2)
            (q2) edge[bend left=20, right] node{$!\mathsf{r\_0}/!\mathsf{g\_0}\;\&\;\mathsf{g\_1}$} (q1)
            (q1) edge[bend left=20, left] node{$\mathsf{r\_0}\;\&\;!\mathsf{r\_1}/!\mathsf{g\_0}\;\&\;!\mathsf{g\_1}$} (q3)
            (q3) edge[bend left=20, left] node{$!\mathsf{r\_1}/\mathsf{g\_0}\;\&\;!\mathsf{g\_1}$} (q1)
            (q1) edge[right] node{$\mathsf{r\_0}\;\&\;\mathsf{r\_1}/!\mathsf{g\_0}\;\&\;!\mathsf{g\_1}$} (q4)
            (q4) edge[right] node{$\mathsf{true}/!\mathsf{g\_0}\;\&\;\mathsf{g\_1}$} (q7)
            (q6) edge[bend right=20, right] node{$!\mathsf{r_0}/!\mathsf{g\_0}\;\&\;!\mathsf{g\_1}$} (q2)
            (q3) edge[bend right=20, left] node{$\mathsf{r_1}/\mathsf{g\_0}\;\&\;!\mathsf{g\_1}$} (q6)
            (q7) edge[right] node{$\mathsf{r_1}/!\mathsf{g\_0}\;\&\;!\mathsf{g\_1}$} (q5)
            (q5) edge[right] node{$\mathsf{true}/\mathsf{g\_0}\;\&\;!\mathsf{g\_1}$} (q6)
            (q6) edge[bend left=20, left] node{$\mathsf{r_0}/!\mathsf{g\_0}\;\&\;!\mathsf{g\_1}$} (q4)
            (q2) edge[bend left=20, right] node{$\mathsf{r_0}/!\mathsf{g\_0}\;\&\;\mathsf{g\_1}$} (q7)
            (q7) edge[bend left=20, right] node{!$\mathsf{r_1}/!\mathsf{g\_0}\;\&\;!\mathsf{g\_1}$} (q3);
    \end{tikzpicture}
    \caption{Mealy Machine for the specification of mutual exclusion without consecutive grants produced by our tool.}
    \label{fig:AlternatingGrantsTarget}
\end{figure}

We finally obtain the machine in \cref{fig:AlternatingGrantsTarget} which is a natural solution. Let us describe the states of this machine. State $q_0$ means that there is no pending request and a grant can be executed at the next step. In state $q_1$, there is no pending request but a grant can be executed at the next step. In state $q_2$, $r_1$ is pending and a grant can be done at the next step. Symmetrically, in state $q_3$, $r_0$ is pending and a grant can be executed at the next step. In state $q_4$, both processes are pending and a grant can be done at the next step, and so on.





\section{LTL syntax and semantics}

To be self-contained, we define here the syntax and semantics of the
linear temporal logic (LTL).

 Given a set of atomatic propositions $P$, the formulas of LTL are built according to the following syntax:
$$\varphi:={\sf true} ~|~p~  |~ \neg \varphi ~| ~\varphi_1 \lor \varphi_2 ~|~ {\bf X} \varphi ~|~ \varphi_1 {\sf U} \varphi_2 $$
\noindent
where $p \in P$ is an atomic proposition, $\varphi$, $\varphi_1$ and $\varphi_2$ are LTL formulas, "${\bf X}$" is the {\em next} operator and "${\sf U}$" is the until operator.

The truth value of a LTL formula along an infinite word $w \in (2^P)^{\omega}$ is defined inductively as follows:
  \begin{itemize}
      \item $w \models {\sf true}$
      \item $w \models p$ iff $p \in w[0]$
      \item $w \models \neg \varphi$ iff $w \not\models \varphi$
      \item $w \models \varphi_1 \lor \varphi_2$ iff $w \models \varphi_1$ or $w \models \varphi_2$
      \item $w \models {\sf X} \varphi$ iff $w[1 \dots] \models
        \varphi$ ($w[1 \dots]$ denotes the suffix of $w$ that exclude
        the first letter $w[0]$) 
      \item $w \models \varphi_1 {\sf U} \varphi_2$ iff there exists $i \geq 0$, such that $w[i \dots] \models \varphi_2$, and for all $j$, $0 \leq j < i$, $w [j \dots] \models \varphi_1$ 
  \end{itemize}
\noindent 
We also consider the following abbrevations:
  \begin{itemize}
      \item "Eventually": ${\sf F} \varphi \equiv {\sf True} {\sf U} \varphi$
      \item "Always": ${\sf G} \varphi \equiv \neg {\sf F} \neg \varphi$
      \item "Weak until": $\varphi_1 {\sf W} \varphi_2 \equiv \varphi_1 {\sf U} \varphi_2 \lor {\sf G} \varphi_1$
      \item "Release": $\varphi_1 {\sf R} \varphi_2 \equiv \neg(\neg \varphi_1 {\sf U} \neg \varphi_2)$
  \end{itemize}

\section{Output of Strix on full arbiter example $n=2$}\label{app:nounsoliticed}

The output of Strix on the complete mutual exclusion specification of
the introduction is given on Fig.~\ref{fig:nounsolicited}. 

\begin{figure}[t]
    \resizebox{\textwidth}{!}{
    \includegraphics[]{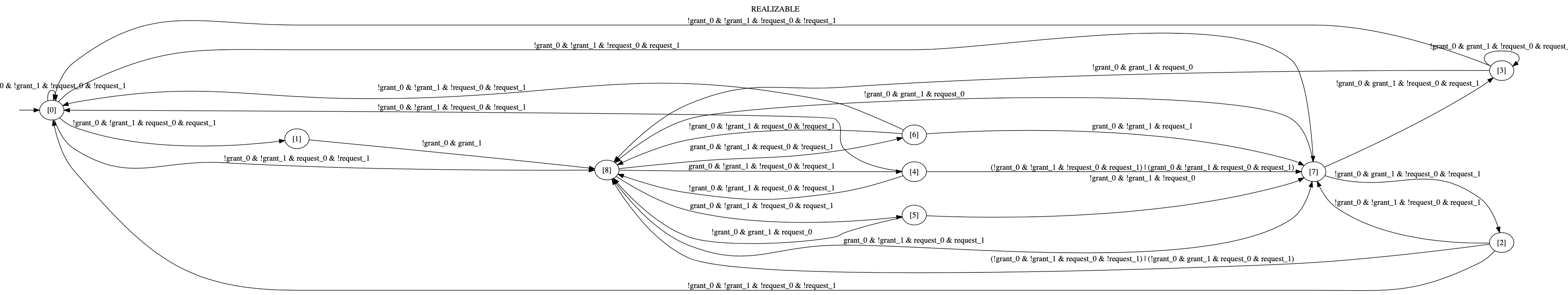}}
    \caption{Output of Strix on the full arbiter specification $n=2$.}
    \label{fig:nounsolicited}
\end{figure}

\section{Details and Proofs for Section~\ref{sec:learningframework}}\label{app:learningframework}

\subsection{Additional notations for (pre)-Mealy machines}\label{app:preMealynotations}

\begin{definition}[Notation $\fIO$]
    Given a preMealy machine $\MM = (M,m_{{\sf init}},\trans)$, we define  the
    (possibly partial)
    function $\fIO : M\times\inputs^*\rightarrow
    (\inputs\outputs)^*$ by $\fIO(m,\inp_1\dots \inp_k) =
    \inp_1\out_1\dots \inp_k\out_k$ such that for $1\leq j\leq k$,
    $\transS_\MM(m,\inp_1\dots\inp_j)$ is defined and $\out_j = \transO_\MM(m,\inp_1\dots\inp_j)$.
    For all $u\in\inputs^*$, we write $\fIO(u)$ instead of
   $\fIO(m_{\sf init},u)$. Note that the
    language accepted by $\MM$ satisfies $L(\MM) = \{
    \fIO(u)\in(\inputs\outputs)^* \mid \fIO(u)\text{
      is defined}\}$.
\end{definition}

\begin{remark}
A Mealy machine is a preMealy machine without holes. If $\MM$ is
Mealy machine then the domain of $\fIO^\MM$ is $M\times
\inputs^*$.
\end{remark}

\subsection{Generalization phase: prefix-tree acceptor and state-merging}\label{app:gen-app}

\begin{example}\label{ex:congruence}
            Consider the preMealy machine $\textsf{PTA}(E_0)$ where $E_0$ has been defined
            in Example~\ref{ex:PTA}. The equivalence relation $\sim_0$
            defined
            by the state partition $\{\{0,2,3\},\{1,4\}\}$ is a
            congruence for $\textsf{PTA}(E_0)$, however it is not a
            Mealy-congruence because we have $0\sim_0 3$ but for input
            $\inp'$, $\transO_{\textsf{PTA}(E_0)}(0,\inp') = \out\neq
            \out' = \transO_{\textsf{PTA}(E_0)}(3,\inp')$.

            However, the equivalence relation $\sim_0'$ defined by the state
            partition $\{\{0,1\},$ $\{2,3,4\}\}$ is a Mealy-congruence for
            $\textsf{PTA}(E_0)$. The quotient
            $\textsf{PTA}(E_0)/_{\sim'_0}$ is depicted on
            the right of Figure~\ref{fig:PTA} and it turns out to be a
            (complete) Mealy machine. Its
            language $L$ is denoted by the regexp
            $(\inp'+\inp)\out(\inp\out+\inp'\out')^*$. Note that
            $E_0\subset L$.
        \end{example}

        The following lemma states that the quotient of a preMealy
        machine $\MM$ by a Mealy-congruence for $\MM$ is a preMealy
        machine which generalizes $\MM$.
        
        \begin{lemma}\label{lem:gen}
            If $\sim$ is a Mealy-congruence for $\MM$, then $\MM/_{\sim}$ is a preMealy machine
            s.t. $L(\MM)\subseteq L(\MM/_{\sim})$.
        \end{lemma}
        \begin{proof}
            It is because any execution of $\MM$ over a
            sequence of inputs $v$ corresponds to a unique execution of $\MM/_{\sim}$
            (whose sequence of states is the sequence of equivalence classes
            of states of $\MM$ over $v$), and 
            which produces the same outputs.\qed
        \end{proof}

A \emph{non-congruent
  point} for an equivalence relation $\sim$ over the states of a
preMealy machine $\MM$ is a triple $(x,x',\inp)\in M\times M\times \inputs$ such that
$x\sim x'$, $\transS(x,\inp)$ and $\transS(x',\inp)$ are both defined
but $\transS(x,\inp)\not\sim \transS(x',\inp)$. Given a
non-congruent point $p = (x,x',\inp)$, the equivalence relation $\sim$ can
be updated to remove $p$, to a coarser equivalence
relation $U(\sim,p)$ defined for all $y,y'\in M$ as $y U(\sim,p) y'$
if $y\sim y'$, or $y\sim \transS(x,\inp)$ and $y'\sim \transS(x',\inp)$, or $y'\sim
\transS(x,\inp)$ and $y\sim \transS(x',\inp)$.

\begin{example}\label{ex:updateequiv}
As an example, consider
the preMealy machine $\textsf{PTA}(E_0)$ defined in
Example~\ref{ex:PTA}, and the equivalence relation $\sim_1$ induced by
the state-partition $\{\{0,2\},\{1\},\{3\},\{4\}\}$. Then,
$p_1 = (0,2,\inp)$ is a non-congruent point, because
$\transS_{\textsf{PTA}(E_0)}(0,\inp)=2\not\sim_13=\transS_{\textsf{PTA}(E_0)}(2,\inp)$. Then, 
$\sim_2 = U(\sim_1,p_1)$ is induced by the partition
$\{\{0,2,3\},\{1\},\{4\}\}$. Then, $p_2 = (0,3,\inp')$ is a
non-congruent point for $\sim_2$, and $U(\sim_2,p_2)$ is exactly
$\sim_0$ as defined in Example~\ref{ex:congruence}, which is a
congruence, so, does not contain any non-congruent point.
\end{example}

We denote that an
equivalence relation $\sim$ is finer than some equivalence relation
$\sim'$ by $\sim\finer \sim'$. The following
proposition is (easily) proved in App.~\ref{app:propequiv}:

\begin{proposition}\label{prop:equiv}
    $U(\sim,p)$ is an equivalence relation such that $\sim\finer U(\sim,p)$.
\end{proposition}

Given an equivalence relation $\sim$ over $M$, we now want to define a
procedure which keeps on removing non-congruent points, i.e., keeps on
applying the function $U$ iteratively until there is no non-congruent
points anymore. Therefore the resulting equivalence relation is a
congruence for $\MM$. We prove in
Lemma~\ref{lem:confluence} that the order in which those points are
removed does not matter. We formalize this via the notion of 
\emph{choice function}, which is a function $ch$ which, given any equivalence relation on the states
of a preMealy machine, outputs a non-congruent point (if it exists),
otherwise it is undefined. Given an equivalence relation $\sim$, a choice function $ch$ and two states
$m,m'$, we denote by $\sim_{ch}^{m,m'}$ the fixpoint of the sequence
$(\sim^{m,m',n})_{n\geq 0}$ where for all $x,x'\in M$, $x\sim^{m,m',0} x'$
if $x\sim x'$, or $x\sim m$ and $x'\sim m'$, or $x\sim m'$ and $x'\sim
m$ (in terms of equivalence classes, $\sim^{m,m',0}$ merges $[m]_\sim$
and $[m']_\sim$). For all $n>0$, $\sim^{m,m',n} =
U(\sim^{m,m',n-1}, ch(\MM,\sim^{m,m',n-1}))$ if
$ch(\MM,\sim^{m,m',n-1})$ is defined, otherwise $\sim^{m,m',n}
=\sim^{m,m',n-1}$.

\begin{example}
This converging sequence is already illustrated in
Example~\ref{ex:updateequiv} for some particular choice function which
first picks $p_1 = (0,2,\inp)$ and then $p_2 = (0,3,\inp')$, starting from the equivalence relation
$\sim$ induced by the partition $\{\{s\}\mid s=0,\dots,4\}$. Then,
$\sim^{0,2,0} = \sim_1$, $\sim^{0,2,1} = U(\sim_1,p_1) = \sim_2$ and
$\sim^{0,2,3} = U(\sim_2,p_2) = \sim_0$ which is the fixpoint of the
sequence. Note that taking a different order, first $p_2$ and then
$p_1$, we would get the same fixed point: $U(\sim_1,p_2)$ is induced by the partition
$\{\{0,2\},\{3\},\{1,4\}\}$ and $U(U(\sim_1,p_1))$ is induced by the
partition $\{\{0,2,3\},\{1,4\}\}$, so, $U(U(\sim_1,p_1),p_2) =
U(U(\sim_1,p_2),p_1)$. This observation can be generalized, as shown by the
following lemma (proved in App.~\ref{app:confluence}). 
\end{example}

\begin{lemma}\label{lem:confluence}
    For any equivalence relation $\sim$ and choice functions $ch_1,ch_2$, we have that $\sim^{m,m'}_{ch_1} = \sim^{m,m'}_{ch_2}$.
\end{lemma}

By the previous lemma, given an equivalence relation $\sim$ over the
states of a preMealy-machine $\MM$, we can define $\sim^{m,m'}$ as
$\sim^{m,m'}_{ch}$ for any choice function $ch$. Note that 
$\sim^{m,m'}$ is a congruence for $\MM$ because it is does not contain
non-congruent points anymore, but it is not necessarily a
Mealy-congruence for $\MM$. We say that two
states $m,m'$ of $\MM$ are \emph{$\sim$-mergeable} if $\sim^{m,m'}$ is
a Mealy-congruence for $\MM$. We say that $m$ and $m'$ are \emph{mergeable} if
they are $diag_\MM$-mergeable, where $diag_M$ is the finest
equivalence relation over $M$, i.e. $diag_M = \{ (m,m)\mid m\in
M\}$.

\begin{example}
As an example, state $0$ and $2$ are not mergeable in
Example~\ref{ex:PTA}, because $diag_{E_0}^{0,2} = \sim_0$ is not a
Mealy-congruence, as illustrated in
Example~\ref{ex:updateequiv}. However, $2$ and $3$ are mergeable, $3$
and $4$, and $1$ and $4$. 
\end{example}

We denote by
$\textsf{Mergeable}(\MM,\sim,m,m')$
(resp. $\textsf{Mergeable}(\MM,m,m')$) the predicate
which holds true whenever $m$ and $m'$ are $\sim$-mergeable
(resp. mergeable). When $m$ and $m'$ are $\sim$-mergeable, we let
$\textsf{MergeClass}(\MM,\sim,m,m') = \sim^{m,m'}$ and
$\textsf{MergeStates}(\MM,\sim,m,m') = \MM/_{\sim^{m,m'}}$.

\begin{lemma}\label{lem:soundnessmerge}
    For all preMealy machine $\MM$, all equivalence relation $\sim$
    over the states of $\MM$, all $\sim$-mergeable states $m,m'$, we
    have $\sim\finer \textsf{MergeClass}(\MM,\sim,m,m')$ and therefore
    $L(\MM)\subseteq L(\textsf{MergeStates}(\MM,\sim,m,m'))$. Moreover, if $m\not\sim
    m'$, then $\sim\finerstrict\textsf{MergeClass}(\MM,\sim,m,m')$ and
    therefore $\textsf{MergeStates}(\MM,\sim,m,m')$ has strictly less states than
    $\MM/_\sim$.
\end{lemma}

\begin{proof}
    Immediate by Lemma~\ref{lem:gen} and the definition of $\sim$-mergeable states.\qed
\end{proof}

\subsection{Termination and correctness of Algorithm \textsc{GEN}(Algo.~\ref{algo:gen})}\label{app:terminationgen} We first prove that algorithm \textsc{Gen} indeed generalizes the
examples while preserving realizability of the specification. 

\begin{lemma}\label{lem:charac}
    For all merging strategy $\sigma_G$,  all finite set of
    examples $E$ and all specification $\spec$ given as a deterministic
    safety automaton $\atm$, if \textsc{Gen}$(E,\spec,\sigma_G)\neq UNREAL$, then 
    $\textsc{Gen}(E,\spec,\sigma_G)$ is a preMealy machine $\MM$ such that
    $\spec$ is $\MM$-realizable and $E\subseteq L(\MM)$. If
    $\textsc{Gen}(E,\spec,\sigma_G)= UNREAL$, then there is no such
    preMealy machine. Moreover, $\textsc{Gen}(E,\spec,\sigma_G)$
    terminates in time polynomial in the size\footnote{The size of $E$
      is defined as the cardinality of $\prefs(E)$} of $E$ and
    exponential in $n$ the number of states of $\atm$.
\end{lemma}

\begin{proof}
    Suppose that $\textsc{Gen}(E,\spec,\sigma_G)\neq UNREAL$. Then $E$ is
    necessarily consistent and $\spec$ is $\textsf{PTA}(E)$-realizable. 
    By Lemma~\ref{lem:soundnessmerge}, the loop at line~\ref{line:loop} computes coarser and coarser
    Mealy-congruences for $\textsf{PTA}(E)$. Therefore, if we denote by
    $\sim_E$ the relation computed by the algorithm after exiting the
    loop, we have $diag_E\finer \sim_E$ and hence
    $L(\textsf{PTA}(E)/_{diag_E}) = L(\textsf{PTA}(E)) = E\subseteq
    L(\textsf{PTA}(E)/_{\sim_E})$. Moreover, line~\ref{line:real} also
    ensures that $\spec$ is $\textsf{PTA}(E)/_{\sim}${-}realizable for all
    equivalence relation $\sim$ computed during iterations of the loop, and in particular for
    $\sim_E$.

    Now, suppose that $\textsc{Gen}(E,\spec,\sigma_G)= UNREAL$, then
    either $E$ is not consistent, in which case it is clear that no
    preMealy machine $\MM$ satisfies $E\subseteq L(\MM)$, or $\spec$ is
    not $\textsf{PTA}(E)$-realizable. For the second case, assume that
    there is a preMealy machine $\MM$ such that $E\subseteq L(\MM)$
    and $\spec$ is $\MM$-realizable by some machine $\PP$, and let us derive a contradiction. We show
    that $\spec$ is $\textsf{PTA}(E)$-realizable by some Mealy machine
    $\PP'$ obtained by taking the
    synchronized product of $\PP$ and $\textsf{PTA}(E)$: a state of
    $\PP'$ is either a state $p$ of $\PP$ or a pair $(p,e)$ where $e$ a state of $\textsf{PTA}(E)$. The initial state of the product is $(p_0,\epsilon)$ where
    $p_0$ is the initial state of $\PP$. From a state $(p,e)$ and an
    input $\inp\in\inputs$, if $\transO_\PP(p,\inp) =
    \transO_{\textsf{PTA}(E)}(e,\inp) = \out$ for some
    $\out\in\outputs$, then the product transitions to
    $(\out, (\transS_\PP(p,\inp),\transS_{\textsf{PTA}(E)}))$,
    otherwise, it transitions to $\trans_\PP(p,\inp)$. This is correct
    since $L(\textsf{PTA}(E)) = \prefs(E)\subseteq L(\PP)$. Moreover
    by definition of the product, $\textsf{PTA}(E)$ is a subgraph of
    $\PP'$ (up to state renaming).

    For the complexity, there are $|\prefs(E)|$ visits to the
    loop. Line~\ref{line:mergeable} takes polynomial time (the
    computation of $\sim^{m,m'}$ for any two states $m,m'$ of a preMealy
    machine is in ptime). According to~\ref{thm:sizeSol}, each
    realizability test at Line~\ref{line:real} takes time
    polynomial in the number of states of the preMealy machine and exponential in
    $n$. Here, the number of states of the preMealy machine is smaller
    than the size of $E$. Overall, this gives the claimed complexity.\qed
\end{proof}

\subsection{Proof of Proposition~\ref{prop:equiv}}\label{app:propequiv}

\begin{proof}
    For all $y\in M$, $y\sim y$ hence $y\ U(\sim,p)\ y$. Symmetry is
    also immediate by definition. Let us prove transitivity. If $y\
    U(\sim,p)\ y'$ and $y'\ U(\sim,p)\ y''$, then there are several
    cases:
    \begin{enumerate}
        \item $y\sim y'$ and $y'\sim y''$: hence $y\sim y''$ and so
          $y\ U(\sim,p)\ y''$.
        \item $y\sim y'$ and $y'\sim \transS(x,\inp)$ and $y''\sim
          \transS(x',\inp)$: hence $y\sim \transS(x,\inp)$ and $y''\sim
          \transS(y,\inp)$, so $y\ U(\sim,p)\ y''$.
        \item $y\sim y'$ and $y''\sim \transS(x,\inp)$ and $y'\sim
          \transS(x',\inp)$: hence $y''\sim \transS(x,\inp)$ and $y\sim
          \transS(x',\inp)$, so $y\ U(\sim,p)\ y''$.

        \item $y\sim \transS(x,\inp)$ and $y'\sim\transS(x',\inp)$ and
          $y'\sim y''$: therefore $y''\sim\transS(x',\inp)$ and so $y\ U(\sim,p)\ y''$.

        \item $y\sim \transS(x,\inp)$ and $y'\sim\transS(x',\inp)$ and $y'\sim \transS(x,\inp)$ and $y''\sim
          \transS(x',\inp)$: so, $y\sim y'$ and $y'\sim y''$, which
          implies $y\sim y''$ and so $y\ U(\sim,p)\ y''$.
        \item $y\sim \transS(x,\inp)$ and $y'\sim\transS(x',\inp)$ and $y''\sim \transS(x,\inp)$ and $y'\sim
          \transS(x',\inp)$: we directly get $y\sim y''$ and so $y\ U(\sim,p)\ y''$.

          The remaining cases are symmetrical to the cases already
          proved by substituting $y$ by $y''$ and $y''$ by $y$. 
      \end{enumerate}

\qed

\end{proof}

\subsection{Proof of Lemma~\ref{lem:confluence}}\label{app:confluence}

\begin{proof}

    We define a relation $\rightarrow$ between equivalence relations as
    follows: $\sim\rightarrow \sim'$ if $\sim' = U(\sim, p)$ for some
    some non-congruent point $p$ (if it exists). We prove that
    $\rightarrow$ is locally confluent, which by Newman's lemma
    implies that $\rightarrow$ is globally confluent. In other words,
    the reflexive and transitive closure $\rightarrow^*$ of
    $\rightarrow$ satisfies that whenever $\sim\rightarrow^* \sim_1$
    and $\sim\rightarrow^* \sim_2$, then there exists $\sim'$ such
    that $\sim_1\rightarrow^*\sim'$ and $\sim_2\rightarrow^*
    \sim'$. The proof of local confluence is not difficult but
    technical, as many cases have to be considered. The main idea is
    to show that if $p$ and $p'$ are two different non-congruent points
    of $\sim$, then $p'$ is a non-congruent point of $U(\sim,p)$ and
    $p'$ is a non-congruent point of $U(\sim,p')$, and then we show that
    $U(U(\sim,p),p') = U(U(\sim,p'),p)$.

    Formally, let $p = (x,y,\inp),p=(x',y',\inp')$ be two different non-congruent points for $\sim$. Let
    $\sim_p = U(\sim,p)$ and $\sim_{p'} = U(\sim,p')$. So,
    $\sim\rightarrow \sim_p$ and $\sim_\rightarrow \sim_{p'}$. Assume
    that $\sim_p\neq \sim_{p'}$. We prove that $p'$ is a non-congruent
    point of $\sim_p$ (and symmetrically $p$ is a non-congruent point
    of $\sim_{p'}$). Suppose that $p'$ is not a non-congruent point of
    $\sim_p$. Then, $\transS_\MM(x',\inp)\sim_p
    \transS_\MM(y',\inp)$. We also know that
    $\transS_\MM(x',\inp)\not\sim \transS_\MM(y',\inp)$ because $p'$
    is a non-congruent point of $\sim$. By definition of $\sim_p$, it
    implies that the following symmetrical two cases can happen:
    \begin{enumerate}
        \item $\transS_\MM(x',\inp')\sim \transS_\MM(x,\inp)$ and
          $\transS_\MM(y',\inp')\sim \transS_\MM(y,\inp)$, or
        \item $\transS_\MM(x',\inp')\sim \transS_\MM(y,\inp)$ and
          $\transS_\MM(y',\inp')\sim \transS_\MM(x,\inp)$.
    \end{enumerate}
    We show that both cases imply that $\sim_p = \sim_{p'}$ which is a
    contradiction. Consider the first case. Then,

    $u\sim_p v$ iff

    $u\sim v$, or $u\sim \transS_\MM(x,\inp)$ and
    $v\sim \transS_\MM(y,\inp)$, or $v\sim \transS_\MM(x,\inp)$ and
    $u\sim \transS_\MM(y,\inp)$ iff

    $u\sim v$, or $u\sim \transS_\MM(x',\inp')$ and $v\sim
    \transS_\MM(y',\inp')$, or $v\sim \transS_\MM(x',\inp')$ and
    $u\sim \transS_\MM(y',\inp')$ iff

    $u\sim_{p'} v$.

    The second case is symmetrical. We have just shown that $p'$ is a
    non-congruent point of $\sim_{p}$ and by symmetry, $p$ is a
    non-congruent point of $\sim_{p'}$. We finally prove that
    $U(\sim_p,p') = U(\sim_{p'},p)$, concluding that $\rightarrow$ is
    locally confluent. We show that $U(\sim_p,p')$ is finer than
    $U(\sim_{p'},p)$, the other direction being completly
    symmetrical. Suppose that $u\ U(\sim_p,p')\ v$ and let us show
    that $u\ U(\sim_{p'},p)\ v$. Therefore, we have one the following
    cases:
    \begin{enumerate}
      \item $u\sim_p v$,
        \item $u\sim_p \transS_\MM(x',\inp')$ and
          $v\sim_p\transS_\MM(y',\inp')$,
          \item $u\sim_p \transS_\MM(y',\inp')$ and
    $v\sim_p\transS_\MM(x',\inp')$.
    \end{enumerate}
    Let us consider the first two cases (the third being symmetrical
    to the second):
    \begin{enumerate}
        \item $u\sim_p v$ implies that one of the following three
          cases holds:
          \begin{enumerate}
              \item $u\sim v$: then $u\sim_{p'} v$, and $u
                U(\sim_{p'},p) v$.
              \item $u\sim \transS_\MM(x,\inp)$ and
                $v\sim\transS_\MM(y,\inp)$: then
                $u\sim_{p'}\transS_\MM(x,\inp)$ and
                $v\sim_{p'}\transS_\MM(y,\inp)$, and so $u\
                U(\sim_{p'},p)\ v$.
              \item $u\sim \transS_\MM(y,\inp)$ and
                $v\sim\transS_\MM(x,\inp)$: this case is symmetric to
                the former. 
            \end{enumerate}
          \item $u\sim_p \transS_\MM(x',\inp')$ and
          $v\sim_p\transS_\MM(y',\inp')$ implies that one of the following
          cases hold:
          \begin{enumerate}
            \item $u\sim \transS_\MM(x',\inp')$ and $v\sim
              \transS_\MM(y',\inp')$: so, $u\sim_{p'}  v$ and hence
              $u\ U(\sim_{p'},p)\ v$. 
            \item $u\sim \transS_\MM(x',\inp')$ and $v\sim
              \transS_\MM(x,\inp)$ and $\transS_\MM(y',\inp')\sim
              \transS_\MM(y,\inp)$: hence $v\sim_{p'}
              \transS_\MM(x,i)$ and  $\transS_\MM(y',\inp')\sim_{p'}
              \transS_\MM(y,\inp)$. By definition of $\sim_{p'}$, we
              also have
              $\transS_\MM(y',\inp')\sim_{p'}\transS_\MM(x',\inp')$. Since
              $u\sim \transS_\MM(x',\inp')$, we have $u\sim_{p'}
              \transS_\MM(x',\inp')$ and from the latter statement, we
              get $u\sim_{p'} \transS_\MM(y',\inp')$ and hence $u
              \sim_{p'} \transS_\MM(y,\inp)$. All this imply that
              $u\ U(\sim_{p'},p)\ v$. 
            \item $u\sim \transS_\MM(x',\inp')$ and $v\sim
              \transS_\MM(y,\inp)$ and $\transS_\MM(y',\inp')\sim
              \transS_\MM(x,\inp)$: this case is symmetrical to the
              latter by substituting $y$ by $x$ and $x$ by $y$.

            \item $u\sim \transS_\MM(x,\inp)$ and $\transS_\MM(x',\inp')\sim
              \transS_\MM(y,\inp)$ and $v\sim \transS_\MM(y',\inp')$:
              by definition of $\sim_{p'}$, we have
              $\transS_\MM(x',\inp')\sim_{p'} \transS_\MM(y',\inp')$
              from which we get $v\sim \transS_\MM(y,\inp)$ and
              therefore $v\sim_{p'} \transS_\MM(y,\inp)$. From $u\sim
              \transS_\MM(x,\inp)$ we get $u\sim_{p'}
              \transS_\MM(x,\inp)$. Therefore, $u\ U(\sim_{p'},p)\ v$.
              
            \item $u\sim \transS_\MM(x,\inp)$ and $\transS_\MM(x',\inp')\sim
              \transS_\MM(y,\inp)$ and $v\sim \transS_\MM(x,\inp)$ and $\transS_\MM(y',\inp')\sim
              \transS_\MM(y,\inp)$: we immediately get $u\sim v$, so
              $u\sim_{p'} v$ and hence $u\ U(\sim_{p'},p)\ v$.

            \item $u\sim \transS_\MM(x,\inp)$ and $\transS_\MM(x',\inp')\sim
              \transS_\MM(y,\inp)$ and $v\sim \transS_\MM(y,\inp)$ and $\transS_\MM(y',\inp')\sim
              \transS_\MM(x,\inp)$: from $u\sim \transS_\MM(x,\inp)$
              we get $u\sim_{p'} \transS_\MM(x,\inp)$ and from $v\sim
              \transS_\MM(y,\inp)$ we get $v\sim_{p'}
              \transS_\MM(y,\inp)$, hence $u\ U(\sim_{p'},p)\ v$.

            \item $u\sim \transS_\MM(y,\inp)$ and $\transS_\MM(x',\inp')\sim
              \transS_\MM(x,\inp)$ and $v\sim \transS_\MM(y',\inp')$:
              symmetric of case $(d)$ by swapping $x$ and $y$.
              
            \item $u\sim \transS_\MM(y,\inp)$ and $\transS_\MM(x',\inp')\sim
              \transS_\MM(x,\inp)$ and $v\sim \transS_\MM(x,\inp)$ and $\transS_\MM(y',\inp')\sim
              \transS_\MM(y,\inp)$:          symmetric of case $(e)$ by swapping $x$ and $y$.
            \item $u\sim \transS_\MM(y,\inp)$ and $\transS_\MM(x',\inp')\sim
              \transS_\MM(x,\inp)$ and $v\sim \transS_\MM(y,\inp)$ and $\transS_\MM(y',\inp')\sim
              \transS_\MM(x,\inp)$:          symmetric of case $(f)$ by swapping $x$ and $y$.

          \end{enumerate}
          \end{enumerate}
          \qed
\end{proof}

\subsection{Details and results on the completion phase}\label{app:comple}

Our goal in this section is to prove correctness of
\textsc{Comp}($\MM_0$, $\spec$, $\sigma_C$)  and
termination. The completion procedure  may
not terminate for some completion strategies. It is because the
completion strategy could for instance keep on selecting a pairs of
the form $(\out,m')$ where $m'$ is a fresh state. However we prove
that it always terminates for \emph{lazy} completion strategies, as
defined in Section~\ref{subsec:overview}. Recall that lazy strategies always favour
existing states. We first start by proving correctness.

\begin{lemma}\label{lem:correctnesscompletion}
If the algorithm $\textsc{Comp}(\MM_0,\spec,\sigma_C)$ terminates and returns a Mealy
machine $\MM$, then $\spec$ is $\MM_0$-realized by $\MM$., i.e., $\spec$ is
realizable by $\MM$ and $\MM_0$ is a subgraph of $\MM$ ($\MM_0\preceq
\MM$). 
\end{lemma}

\begin{proof}
    Let $\MM_i$ be the machine
    computed after the $i$th iteration of the \textbf{while}-loop.
    As explained before, the tests at lines~\ref{line:test1}
    and~\ref{line:test2} ensures the invariant that each $\MM_i$ $\MM_0$-realizes $\spec$. It is trivial for the
    test at line~\ref{line:test1}. Since iteration $i+1$ of the
    algorithm completes the machine $\MM_i$ into a machine
    $\MM_{i+1}$, we get $\MM_0\preceq \MM_1\preceq \MM_2\dots$.
    Therefore, if $\spec$ is $\MM_i$-realizable, it is also
    $\MM_0$-realizable by $\MM_i$, so the test at
    line~\ref{line:test1} guarantees that all the machines $\MM_i$
    $\MM_0$-realize $\spec$. Moreover, the list of
    candidates at line~\ref{line:selection} is guaranteed to be
    non-empty. Indeed, by the invariant, $\MM_i$ can always be
    completed into a Mealy machine realizing $\spec$. So the selection
    at line~\ref{line:selection} is well-defined. Hence, if the
    algorithm returns a machine, this machine is necessarily a Mealy
    machine, because it has no holes, and moreover it $\MM_0$-realizes $\spec$.\qed
\end{proof}

In the sequel, our goal is to prove termination for lazy
strategies. The following technical lemma is a key lemma towards showing
termination. It gives a sufficient condition for which a state of a
preMealy machine can be reused to complete a hole. We need some
notation. Given a preMealy machine $\MM = (M,m_0,\trans)$ and 
 a deterministic parity automaton $A = (Q, q_0, \inputs\cup\outputs,
 \delta_A,d)$, for all $m\in M$, we let $R_m^{A,\MM}\subseteq Q$ (or just
 $R_m$ if $A$ and $\MM$ are clear from the context, all the states of $A$
 reachable from its initial state when reading words that reach state
 $m$ when read by $\MM$. Formally,
 $$
 R_m^{A,\MM} = \{\transS^*_A(q_0,u)\mid u\in
      (\inputs\outputs)^*, \transS_\MM^*(m_0,u)=m\}
      $$
      Given a subset $Q'\subseteq Q$ and some input
      $u\in(\inputs\outputs)^*$, we let $\transS_A^*(Q,u) = \{
      \transS_A^*(q,u)\mid q\in Q'\}$. 

\begin{lemma}\label{lem:keylemmatermination}
    Let $\spec$ be a safety specification given as a
    (complete\footnote{Automata in this paper are complete by
      definition, i.e. there is always a transition from any state on
      any input, but we stress it here as it is a necessary
      requirement for the statement to hold.}) deterministic safety
    automaton $A = (Q, Q_{usf}, q_0, \delta_A)$. Let $\MM
    = (M,m_0,\delta)$
    be a preMealy-machine such that $\spec$ is $\MM$-realizable. Let
    $(m,\inp)$ be a hole of $\MM$ (if it exists). For all $\out\in\outputs$ and
    $m'\in M$, if $\transS_A^*(R_m,\inp\out) \subseteq R_{m'}$, then
    $\spec$ is $\MM'$-realizable for $\MM' = (M, m_0,\delta\cup \{(m,\inp)\mapsto
    (\out,m')\})$.
\end{lemma}

\begin{proof}
    We keep the same notations as in the statement of the lemma.
    Let $\MM_{t} = (M_t,m_0,\trans_t)$ be a Mealy machine which realizes $\spec$ and such that
    $\MM$ is a subgraph of $\MM_t$, i.e., $\MM\preceq \MM_t$. The
    subscript $t$ stands for the fact that the transition function of
    $\MM_t$ is total. We assume without loss of generality that
    $\MM_t$ has a special form: when $M\subseteq M_t$ is left,
    it is never visited again. Formally, for all $u\in\inputs^*$, if
    $\transS_{\MM_t}^*(m_0,u)\not\in M$, then for all $v\in\inputs^*$, 
    $\transS_{\MM_t}^*(m_0,uv)\not\in M$. The machine $\MM_t$ can be
    modified, using one additional bit of memory, so that it satisfies
    this assumption.

    Let $\out = \transO_{\MM_t}(m,\inp)$ and let $\MM' =
    (M,m_0,\delta' := \delta\cup \{(m,\inp)\mapsto (\out,m')\})$.
    We modify $\MM_t$ into a Mealy machine $\MM'_t$ by redirecting
    the transition from $(m,\inp)$ to $(\out,m')$. Formally, $\MM'_t =
    (M_t, m_0, \trans_t')$ where
    $$
    \trans_t' = (\trans_t \setminus \{
    (m,\inp)\rightarrow \trans_t(m,\inp)\})\cup \{(m,\inp)\mapsto
    (\out,m')\}
    $$
    Clearly, $\MM'$ is a subgraph of $\MM_t'$ because $\MM$ is a
    subgraph of $\MM_t$. It remains to show that $\MM_t'$ realizes
    $\spec$.

    Assume that it is not the case and let us derive a
    contradiction. In other words, there exists $w\in
    L_\omega(\MM'_t)$ ($u$ is a finite word) such that $w\not\in L(A)$.
    Let $u\in(\inputs\outputs)^*$ be the shortest unsafe prefix of $w$,
    i.e. the prefix of $w$ such that the
    states visited by the execution of $A$ on $u$ are safe but the
    last one. We decompose $u$ according to its execution in $\MM'_t$
    and the visit to the new transition $(m,\inp)\mapsto
    (\out,m')$ (call it $t_{new}$). By our assumption on the
    form of $\MM_t$, whenever $\MM$ is left, it is never visited
    again. By definition of $\MM'_t$, we therefore have that whenever
    $\MM'$ is left, it is never visited again. So, the execution of $\MM'_t$ on $u$ visits $t_{new}$ a
    couple of times (at least once) while staying in $\MM$ and
    after the last visit to $t_{new}$, is continued by the execution
    of $\MM_t$ on the remaining suffix of $u$, ending in an unsafe
    state. Formally, there exist $u_1,\dots,u_k\in
    (\inputs\outputs)^*$ and $p_1,p'_1,\dots,p_{k-1},p'_{k-1},p_k\in
    Q$ such that $u = u_1\inp\out u_2\inp\out\dots \inp\out u_k$ and
    $$
    \begin{array}{rllllllllllllllll}
      q_0  & \xrightarrow{u_1}_{A} & p_1 &
                                                  \xrightarrow{\inp\out}_A
      & p'_1 & \xrightarrow{u_2}_\mathsf{assist} & p_2 & \dots & p'_{k-1} &
                                                              \xrightarrow{u_k}
      & p_k\in Q_{usf} \\
      m_0 & \xrightarrow{u_1}_\MM & m & \xrightarrow{\inp|\out}
      & m' & \xrightarrow{u_2}_\MM & m & \dots & m' &
                                                       \xrightarrow{u_k}_{\MM_t}
      & m''
    \end{array}
    $$
    We prove that for all $1\leq j\leq k-1$, there exists
    $x_j\in(\inputs\outputs)^*$ such that $\transS_\MM(m_0,x_j) = m'$
    and $\transS_A(p_0,x_j) = p'_j$. We prove by induction on $j$. For
    $j=1$, note that we have $p_1\in R_m^{A,\MM}$, and so $p'_1\in
\transS_A(R_m^{A,\MM},\inp\out)$. So, $p'_1\in R_{m'}^{A,\MM}$. It
implies that there exists $x_1\in(\inputs\outputs)^*$ satisfying the
claim. Suppose it is true at rank $j-1$. So, there exists $x_{j-1}$
such that $\transS_\MM(m_0,x_{j-1}) = m'$
    and $\transS_A$ $(p_0,x_{j-1}) = p'_{j-1}$. Therefore,
    $\transS_\MM(m_0,x_{j-1}u_j) = m$ and $\transS_A(p_0,x_{j-1}u_j) =
    p_j$. This implies that $p_j\in R_m^{A,\MM}$, and so $p'_j\in
\transS_A(R_m^{A,\MM},\inp\out)$, and so there exists $x_j\in(\inputs\outputs)^*$ satisfying the
claim at rank $j$.

Now, consider the word $x_{k-1}u_k$ :  $\transS_A(p_0,x_{k-1}u_k) =
p_k\in Q_{usf}$ and $\transS_{\MM_t}$ $(m_0,x_{k-1}u_k)$ $= m''$. This contradicts that $\MM_t$ realizes $L(A)$. \qed
\end{proof}

We are now ready to prove termination of Algo~\ref{algo:comp}. To
establish the complexity, we need to introduce some notions about
chains and antichains of subsets. Let $Y$ be a finite set of
cardinality $n$. It is well-known that the set $2^Y$ is partially
ordered by inclusion. Therefore, an antichain of elements of $2^Y$ is
a set $\mathcal{Y}\subseteq 2^Y$ such that for all $Y_1,Y_2\in
\mathcal{Y}$, $Y_1$ and $Y_2$ are incomparable by $\subseteq$. We
denote by $\anti_\subseteq(Y)$ the set of $\subseteq$-antichains over $Y$. The set
$\anti_\subseteq(Y)$ can be partially ordered by the partial order
denoted $\trianglelefteq$: for all
$\mathcal{Y}_1,\mathcal{Y}_2\in \anti_\subseteq(Y)$, $\mathcal{Y}_1
\trianglelefteq \mathcal{Y}_2$ if for all $Y_1\in \mathcal{Y}_1$,
there exists $Y_2\in \mathcal{Y}_2$ such that $Y_1\subseteq Y_2$. It
is well-known that $(\anti_\subseteq(Y),\trianglelefteq)$ is a lattice. A
chain in $(\anti_\subseteq(Y),\trianglelefteq)$ is a sequence $\mathcal{Y}_1\vartriangleleft
\mathcal{Y}_2 \vartriangleleft\dots \vartriangleleft
\mathcal{Y}_m$ (note that all relations are strict). The following
lemma is key to bound the
termination time:

\begin{lemma}\label{lem:anti}
    Let $Y$ be a set of cardinality $n$. 
    Any $\vartriangleleft$-chain in $(\anti_\subseteq(Y),\trianglelefteq)$
     has length at most $2^n$. 
\end{lemma}

\begin{proof}
    Given an antichain $\mathcal{Y}\in \anti_\subseteq(Y)$, we let
    $\downarrow \mathcal{Y}$ be the downward closure of $\mathcal{Y}$,
    i.e., $
    \downarrow \mathcal{Y} = \{ X\subseteq Y\mid \exists X'\in
    \mathcal{Y},X\subseteq X'\}
    $.
    Now, observe that for all $\mathcal{Y}_1,\mathcal{Y}_2\in
    \anti_\subseteq(Y)$, we have
    $\mathcal{Y}_1\vartriangleleft\mathcal{Y}_2$ iff $\downarrow
    \mathcal{Y}_1\subsetneq  \downarrow
    \mathcal{Y}_2$. Indeed, for the 'only if' direction, assume that 
    $\mathcal{Y}_1\vartriangleleft\mathcal{Y}_2$ and take $X_1\in 
{\downarrow} \mathcal{Y}_1$. Hence there exists $Y_1\in \mathcal{Y}_1$
s.t. $X_1\subseteq Y_1$. So, there exists $Y_2\in \mathcal{Y}_2$
s.t. $Y_1\subseteq Y_2$, which implies that $X_1\in {\downarrow}
\mathcal{Y}_2$. Conversely, suppose that $\downarrow
    \mathcal{Y}_1\subsetneq  \downarrow
    \mathcal{Y}_2$. Let $Y_1\in \mathcal{Y}_1$. Then, $Y_1\in
    {\downarrow} \mathcal{Y}_1$, so $Y_1\in {\downarrow} \mathcal{Y}_2$,
    which means that there exists $Y_2\in\mathcal{Y}_2$ s.t.
    $Y_1\subseteq Y_2$. This observation implies that the length of any $\vartriangleleft$-chain of
    antichains is at most the length of a maximal $\subsetneq$-chain
    of subsets of ${2^Y}$, which is at most $2^{|Y|}$. \qed
\end{proof}

We now prove termination (assuming the completion strategy $\sigma_C$
is computable in exptime). 

\begin{lemma}\label{lem:terminationintermediate}
   If $\sigma_C$ is lazy, \textsc{Comp}$(\MM_0,S,\sigma_C)$ terminates
   in time polynomial in $|\inputs|$ and the number of holes in
   $\MM_0$, and exponential in the number of states of the
   deterministic safety automaton defining $\spec$, and the number of
   states of $\MM_0$. 
\end{lemma}

\begin{proof}
    Let $A = (Q,q_0,\delta,d)$ be a deterministic parity automaton
    defining the specification. Let $\MM_0 = (M_0,m_0,\trans_0)$ be a preMealy machine such that $\spec$
    is $\MM_0$-realizable (otherwise the algorithm terminates at
    line~\ref{line:test1}). Suppose that the algorithm does not
    terminate and for all $i\geq 1$, let $\MM_i = (M_i,m_0,\trans_i)$ be the preMealy-machine
    computed at the $i$th-iteration of the \textbf{while}-loop.
    For all $i\geq 0$ and $m\in M_i$, we define $R_m^i$ as a shortcut
    for $R_m^{A,\MM_i}$. Since $m$ is a state of $\MM_i$ and $\MM_i$ is a subgraph of
    $\MM_{i+1}$, the set of words reaching $m$ in $\MM_i$
    is included in the set of words reaching $m$ in
    $\MM_{i+1}$. Therefore we obtain the following
    monotonicity property: $R_m^i\subseteq
    R_m^{i+1}$.

    Let $\mathcal{X}_i = \{ R_m^i\mid m\in M_i\}$. Let denote by
    $\lceil \mathcal{X}_i \rceil$ the maximal elements of
    $\mathcal{X}_i$ for inclusion. By the monotonicity
    property, we get that the sequence $(\lceil \mathcal{X}_i
    \rceil)_{i\geq 0}$ eventually stabilizes: there exists $\alpha$
    such that for all $i\geq \alpha$, $\lceil \mathcal{X}_i
    \rceil = \lceil \mathcal{X}_{i+1}
    \rceil$. 

    Consider a machine $\MM_j$ for $j\geq \alpha$ and the hole $(m,\inp)$ of
    $\MM_j$ completed at iteration $j$ by the algorithm,
    i.e. $\trans_{j+1}(m,\inp)=(\out,m')$ for some
    $\out\in\outputs$ and $m'\in M_{j+1}$. We claim that $m'\in M_j$,
    i.e., the algorithm has reused some existing state of
    $\MM_j$. Indeed, consider the set $R_{m,\inp,\out}^j =
    \transS_A^*(R_m^j,\inp\out)$. We have that
    $R_{m,\inp,\out}^j\subseteq R_{m'}^{j+1}$ because
    $\MM_{j+1}$ transitions from $(m,\inp)$ to $(\out,m')$. Clearly,
    there exists $X\in \lceil \mathcal{X}_{j+1} \rceil$ such that
    $R_{m,\inp,\out}^j \subseteq R_{m'}^{j+1} \subseteq X$ and  since $\lceil \mathcal{X}_j \rceil =
    \lceil \mathcal{X}_{j+1} \rceil$, $X\in \lceil \mathcal{X}_{j}
    \rceil$. By definition of $\mathcal{X}_{j}$, $X = R_{m''}^j$ for
    some $m''\in M_j$. Hence, $R_{m,\inp,\out}^j\subseteq R_{m''}^j$.
    Let $\MM' =  (M_j, m_0, \delta_i\cup \{(m,\inp)\mapsto
    (\out,m'')\}$. By Lemma~\ref{lem:keylemmatermination}, $\spec$ is
    $\MM'$-realizable, therefore at iteration $j$, the set
    $candidates$ at line~\ref{line:selection} contains the pair
    $(\out,m'')$. Since the selection strategy is lazy, if $m'\not\in
    M_j$, it would favour $m''$. Hence, $m'\in M_j$.

    We have just proved that $\MM_{j+1}$ has strictly one less hole than
    $\MM_j$, for all $j\geq \alpha$. It implies that if $\MM_\alpha$
    has $k$ holes, then $\MM_{\alpha+k}$ has no holes, and the
    algorithm terminates, contradiction.

    We now bound the number of iterations of the \textbf{while}-loop
    before termination. We only give the main ideas. 
    For all iteration $j$ of the algorithm,
    we let $H_j$ be the number of holes of $\MM_j$. The following
    claim is proved in App.~\ref{app:terminationsafetyclaim}:

    \noindent \textit{Claim.} For all $j{<}\alpha$, either $H_{j+1} = H_j - 1$ and $\lceil \mathcal{X}_j \rceil
    \trianglelefteq \lceil \mathcal{X}_{j+1} \rceil$, or
    $H_{j+1} = H_j + |\inputs|-1$ and $\lceil \mathcal{X}_j \rceil
    \vartriangleleft \lceil \mathcal{X}_{j+1} \rceil$.

Call the first case
    of the claim a \emph{decrease step} and the other case an
    \emph{increase step}. Thanks to Lemma~\ref{lem:anti}, the maximal
    number of increase steps is bounded by $2^{|Q|}$. This allows us
    to bound the number of decrease steps as well. A simple
    calculation detailed in App.~\ref{app:terminationsafetyup} entails that the
    number of iterations before termination is bounded by $
    2^{|Q|}+ k_0+ (2^{|Q|}+1)(|\inputs|-1) \in
    O(k_0+|\inputs|.2^{|Q|})$.

    The overall time complexity for $\textsc{Comp}$ to terminate is
    then the number of iterations bounded by the expression above,
    multiplied by the time complexity of each
    inner-computation of the \textbf{while}-loop, which is dominated by the
    complexity of checking $\MM_{\out,m'}$-realizability, which is
    polynomial in the number of states of $\MM_{\out,m'}$ and
    exponential in the number of states of $A$, by
    Theorem~\ref{thm:sizeSol}.\qed
\end{proof}

\subsection{Proof of the claim used in Lemma~\ref{lem:terminationintermediate}}\label{app:terminationsafetyclaim}
        
We first start by proving the claim:

    \textit{Claim} At each iteration $j<\alpha$, either $H_{j+1} = H_j - 1$ and $\lceil \mathcal{X}_j \rceil
    \trianglelefteq \lceil \mathcal{X}_{j+1} \rceil$, or
    $H_{j+1} = H_j + |\inputs|-1$ and $\lceil \mathcal{X}_j \rceil
    \vartriangleleft \lceil \mathcal{X}_{j+1} \rceil$.

    \begin{proof}
    By the
    monotonicity property, we always have that $\lceil \mathcal{X}_j \rceil
    \trianglelefteq \lceil \mathcal{X}_{j+1} \rceil$. Now, suppose that
    at step     $j$ some of the states of $\MM_j$ can be reused to complete the
    selected hole, then $H_{j+1} = H_j - 1$.

    Suppose now that at step $j$, a new state has been created, and
    let $h = (m,\inp)$ be the hole selected for completion. Clearly,
    $H_{j+1} = H_j+|\inputs|-1$ since adding a new state creates
    $|\inputs|$ holes. Since no state could be reused, it implies by
    Lemma~\ref{lem:keylemmatermination} that for all $\out\in\outputs$
    and all $m'\in M_j$, $\transS^*_A(R_m^j,\inp\out)\not\subseteq
    R_{m'}^j$. So, in particular,
    $\transS^*_A(R_m^j,\inp\out)\not\subseteq X$ for all $X\in \lceil
    \mathcal{X}_j \rceil$.  Let $(m,\inp)\mapsto (\out,f)$
    be the new transition added to $\MM_j$, where $f\not\in M_j$. Note
    that $R_f^{j+1} = \transS^*_A(R_m^j,\inp\out)$. Therefore, $R_f^{j+1}\not\subseteq
    X$ for all $X\in \lceil \mathcal{X}_j \rceil$. Since $R_f^{j+1}\in
    \mathcal{X}_{j+1}$, there exists $Y\in \lceil \mathcal{X}_{j+1}
    \rceil$ such that $R_f^{j+1}\subseteq Y$. Necessarily, $Y\not\in \lceil \mathcal{X}_{j+1}
    \rceil$, proving that $\lceil \mathcal{X}_{j} \rceil\neq \lceil
    \mathcal{X}_{j+1} \rceil$. So, $\lceil \mathcal{X}_{j} \rceil\vartriangleleft \lceil
    \mathcal{X}_{j+1} \rceil$.\qed
\end{proof}

\subsection{Proof of the upper-bound given in the proof of Lemma~\ref{lem:terminationintermediate}}\label{app:terminationsafetyup}

The following proof details the calculation done to bound the number
of iterations, call it $\beta$, of the \textbf{while}-loop before termination.

\begin{proof}
    
    Let $j_1<j_2<\dots <j_t\leq \alpha$ be the
    iterations corresponding to an increase step. We also let
    $j_{t+1} = \beta$ and $j_0 = 0$. For all $0\leq \ell\leq t+1$,
    let $x_\ell$ be the number of holes in the machine
    $\MM_{j_\ell}$. In particular, $x_{t+1} = 0$. 
    Then, we have the following relation:
    $$
    x_0 = k_0\qquad x_{\ell+1} = x_{\ell} - (j_{\ell+1} - j_{\ell} - 1) +
    |\inputs| - 1
    $$
    Indeed, in between two increase steps $j_\ell$ and $j_{\ell+1}$,
    there is $j_{\ell+1}-j_\ell-1$ decrease steps and the increase
    step $j_{\ell+1}$ adds $|\inputs|-1$ holes. 
    
    Clearly, $\beta$ is bounded by the maximal number of increase
    steps plus the maximal number of decrease steps. The latter
    corresponds to
    $$
    \sum_{\ell=0}^{t} (j_{\ell + 1} - j_\ell - 1) = \sum_{\ell=0}^{t}
    (x_\ell - x_{\ell+1} + |\inputs|-1) $$
    $$ = x_0 - x_{t+1}
    +(t+1)(|\inputs|-1) = k_0 + (t + 1) (|\inputs|-1)
    $$
    Moreover, $t$ is bounded by $2^{|Q|}$, and therefore
    $$
    \beta\leq 2^{|Q|}+ k_0+ (2^{|Q|}+1)(|\inputs|-1) \in O(k_0+|\inputs|.2^{|Q|})
    $$
\qed
\end{proof}

\subsection{Proof of Theorem~\ref{thm:terminationcorrectness}}\label{sec:proofs}

The proof of this theorem is based on a series of Lemmas proved in
App.~\ref{app:learningframework}. Let us give an overview of how
this appendix is structured:

\begin{enumerate}
  \item App.~\ref{app:gen-app} formally defines the notions used
    in algorithm \textsc{Gen}: prefix-tree acceptor, state merging,
    and in particular the notions of mergeable classes and the result
    of merging them, together with
    examples. Lemma~\ref{lem:soundnessmerge} states properties about
    merging. 

  \item Lemma~\ref{lem:charac} in App.~\ref{app:terminationgen} proves correctness properties about the
    generalization phase (algorithm \textsc{Gen}), and provides an
    analysis of its termination time. 

  \item App.~\ref{app:comple} establishes correctness of the
    completion phase (Lemma~\ref{lem:correctnesscompletion}). Then, it
    provides a complexity analysis when the completion strategy is
    lazy (Lemmas~\ref{lem:keylemmatermination},~\ref{lem:anti}
    and~\ref{lem:terminationintermediate}). 
\end{enumerate}

We now have all the ingredients to prove
Theorem~\ref{thm:terminationcorrectness}. First, we prove the
correctness part of the statement. Suppose that
$\textsc{SynthSafe}(E,\spec,\sigma_G,\sigma_C)$ terminates, then there are two
cases:

\begin{enumerate}
    \item $\textsc{SynthSafe}$ $(E,\spec,\sigma_G,\sigma_C) = \textsc{UNREAL}$, then
      either $\textsc{Gen}$ $(E,\spec,\sigma_G) = $ $\textsc{UNREAL}$, and we get the
      result by Lemma~\ref{lem:charac}, or
      $\textsc{Gen}(E,\spec,\sigma_G)$ returns some preMealy machine
      $\MM_0$ such that  $\spec$ is not $\MM_0$-realizable. This case is
      impossible: by Lemma~\ref{lem:charac}, $\spec$ is
      necessarily $\MM_0$-realizable.

    \item $\textsc{SynthSafe}(E,\spec,\sigma_G,\sigma_C)$ returns a Mealy
      machine $\MM$. Let $\MM_0$ be the preMealy machine returned by 
      $\textsc{Gen}(E,\spec,\sigma_G)$. Then by
      Lemma~\ref{lem:charac}, $\spec$ is $\MM_0$-realizable and
      $E\subseteq L(\MM_0)$. From
      Lemma~\ref{lem:correctnesscompletion} we get that $\spec$ is
      $\MM_0$-realizable by $\MM$. Therefore, $E\subseteq L(\MM)$ and
      $\spec$ is realizable by $\MM$. 
  \end{enumerate}

We now prove termination in the case of lazy strategies, together with
the complexity. First, $\textsc{Gen}(E,\spec,\sigma_G)$ always
terminate, in time polynomial in the size of $E$ and exponential in
$n$ the number of states of $\atm$, according to
Lemma~\ref{lem:charac}. If $\textsc{Gen}(E,\spec,\sigma_G) \neq
UNREAL$, then it outputs a preMealy machine $\MM_0$ obtained by merging states of \textsf{PTA}(E),
hence it has less states than the size of $E$. From
Lemma~\ref{lem:terminationintermediate},
$\textsc{Comp}(\MM_0,\spec,\sigma_c)$ terminates, in time polynomial in
$|\inputs|$ and the number of holes of $\MM_0$ (which is bounded by
$|E|$), and exponential in $n$. This yields the claimed
complexity.

\subsection{Proof of Theorem~\ref{thm:mealycompleteness}}
\label{app:proofmealycompleteness}
\input{ProofMealyCompleteness.tex}

\subsection{Proofs of the claims in the proof of Lemma~\ref{lem:learnability}}\label{app:claimscharac}

\paragraph{Proof of Claim 1} First, let us prove that $\sim_T$ is a
congruence. Let $e\sim_T e'$ and $\inp\in\inputs$. Suppose that
$\trans_{\textsf{PTA}(E)}(e,\inp)$ and
$\trans_{\textsf{PTA}(E)}(e',\inp)$ are both defined, i.e., $e\inp\out\in
E$ and $e'\inp\out'\in E$ for some $\out,\out'\in\outputs$. Then, we get $\Phi(e\inp) =
\transS_\TT(\Phi(e),\inp) = \transS_\TT(\Phi(e'),\inp) =
\Phi(e'\inp)$. In other words, $e\inp \sim_T e'\inp$. Let us show that
$\out=\out'$. Clearly, $\out = \transO_\TT(\Phi(e),\inp)$ because $E =
L(\textsf{PTA}(E)) \subseteq L(\TT)$. Similarly, $\out' =
\transO_\TT(\Phi(e'),\inp)$. So, $\out = \out'$ follows since $\Phi(e)
= \Phi(e')$. Clearly, $\sim_T$ has at most $|T|$ equivalence
classes. Since $E$ contains for all states $t\in T$, an example
$s_t$ such that $\Phi(s_t) = t$, it follows that $\sim_T$ has at least
$|T|$ equivalence classes. So far, we have proved that $\sim_T$ is a
Mealy-congruence for $\textsf{PTA}(E)$
 with the same number of equivalence
 classes as the number of states of $\TT$. For $t\in T$, we let
 $\Phi^{-1}(t) = \{ e\in E\mid \Phi(e)=t\}$. Note that $\Phi^{-1}(t) =
 [e]$ for some representative $e\in E$. 
We finally show that $\TT$ and $\textsf{PTA}(E)/_{\sim_T}$ are equal, up to the state renamping $\Phi^{-1}$. 
 First, the initial state of $\TT$ maps to $[\epsilon]$, which is the
 initial state of  $\textsf{PTA}(E)/_{\sim_T}$. Let us show that
 $\Phi^{-1}$ preserves the transitions of $\TT$. Let $t,\inp,\out,t'$ such
 that $\trans_\TT(t,\inp)=(\out,t')$. We show that
 $\trans_{\textsf{PTA}(E)/_{\sim_T}}(\Phi^{-1}(t),\inp) =
 (\out,\Phi^{-1}(t'))$. First, $s_t\in \Phi^{-1}(t)$ since
 $\Phi(s_t)=t$. Moreover, by definition of $E$, $e_{t,\inp} =
 \fIO^\TT(s_t\inp) = \fIO^\TT(s_t)\inp\out\in E$, so
 $\trans_{\textsf{PTA}(E)}(\fIO^\TT(s_t),\inp) =
 (\out,\fIO^\TT(s_t\inp))$. Moreover, $\Phi(s_t\inp) =
 t'$. This concludes that $\trans_{\textsf{PTA}(E)/_{\sim_T}}(\Phi^{-1}(t),\inp) =
 (\out,\Phi^{-1}(t'))$. The converse is a consequence of $\sim_T$
 being a Mealy-congruence for $\textsf{PTA}(E)$ and the fact that all
 outputs picked by $\textsf{PTA}(E)$ are consistent with $\TT$, i.e.,
 $L(\textsf{PTA}(E)) = E\subseteq L(\TT)$.  $\hfill$\textit{End of proof of Claim 1.}\qed

\paragraph{Proof of Claim 2} By definition of
        $\sim^{x, y}$, there exist some non-congruent points
        $p_1,\dots,p_n\in E\times E\times \inputs$ such that
        $$
        \sim^{x, y} = U(U(\dots (U(\sim^{x,y,0},p_1),\dots ), p_{n-1}),p_n)
        $$
        Let $\sim^0 = \sim^{x,y,0}$ and $\sim^j =
        U(\sim^{j-1},p_j)$ for $1\leq j\leq n$. We prove by induction
        on $j$ that $\sim^j\finer\sim_T$.

        At rank $j=0$, $\alpha\sim^0\beta$ means that either $(i)$ $\alpha \sim
        \beta$ or,  $(ii)$ $\alpha \sim x$ and $\beta \sim y$, or
        $(iii)$ $\beta \sim x$ and $\alpha \sim y$. In case $(i)$,
        by hypothesis, $\sim\finer\sim_T$ so we are done.  In case $(ii)$,
        by assumption, $x\sim_T y$
        and as $\sim\finer \sim_T$, $\alpha\sim_T x$ and
        $y\sim_T\beta$. So, $\alpha\sim_T\beta$. Case $(iii)$ is symmetrical to $(ii)$.

        At rank $j>0$, if $p_j = (z_1,z_2,\inp)$, then
        $\alpha\sim^j\beta$ means that either $(i)$ $\alpha\sim^{j-1}
        \beta$, or
        $(ii)$ $\alpha\sim^{j-1}\transS_{\textsf{PTA}(E)}(z_1,\inp)$ and 
        $\beta\sim^{j-1}\transS_{\textsf{PTA}(E)}(z_2,\inp)$, or
        $(iii)$ symmetric of $(ii)$ by swapping $\alpha$ and
        $\beta$. In case $(i)$, we get the statement by IH. In case
        $(ii)$, by IH, we get $\Phi(\alpha) =
        \Phi(\transS_{\textsf{PTA}(E)}(z_1,\inp))$, and $\Phi(\beta)
        = \Phi(\transS_{\textsf{PTA}(E)}(z_2,\inp))$. By definition
        of a non-congruent point, we also have $z_1\sim^{j-1} z_2$,
        so, $\Phi(z_1) = \Phi(z_2)$.  Now,
        $\Phi(\transS_{\textsf{PTA}(E)}(z_1,\inp)) =
        \transS_\TT(\Phi(z_1),\inp) = \transS_\TT(\Phi(z_2),\inp') =
        \Phi(\transS_{\textsf{PTA}(E)}(z_2,\inp))$ from which we get
        $\Phi(\alpha)=\Phi(\beta)$. Case $(iii)$ is symmetrical to
        $(ii)$. $\hfill$\textit{End of Proof of Claim 2}. \qed

\section{Details and results of Section~\ref{sec:omega-reg}}

\subsection{Optimizing $\PP$-realizability checking for specifications $\mathcal{D}(\atm,k)$}\label{app:improve-ppcheck}.

As said in the  main body of the paper, a direct application of
Theorem~\ref{thm:sizeSol} on $\mathcal{D}(\atm,k)$ to check
its $\PP$-realizability would yield a doubly exponential
upper-bound. We prove instead that one exponential can be saved by
exploiting the structure of $\mathcal{D}(\atm,k)$, as summarized by
the following theorem:

\begin{thm}
\label{thm:sizeSolk}
    Given a universal co-B\"uchi automaton $\atm$ with $n$ states and $k \in \mathbb{N}$ defining a safety
    specification $S = L_k^\forall(\atm)=L(\mathcal{D}(\atm,k))$ and a preMealy machine $\PP$
    with $m$ states and $n_h$ holes,
    deciding whether $\spec$ is $\PP$-realizable is 
    {\sc ExpTime-Complete}. 
\end{thm}

In this theorem, $k$ is assumed to be given in binary. 
To establish the upper bound, we exploit the fact that the set of
states of $\mathcal{D}(\atm,k)$ forms a complete lattice with several
interesting properties.

\begin{definition}[Lattice of counting functions]
For all co-B\"uchi automata $\atm$ and $k \in \mathbb{N}$, let $\preceq \subseteq CF(\atm,k) \times CF(\atm,k)$ be defined by $f_1 \preceq f_2$ if and only if $f_1(q) \leq f_2(q)$ for all $q \in Q$. The set $(CF(\atm,k),\preceq)$ forms a complete lattice with minimal elements $\overline{-1}=\langle -1,-1,\dots,-1 \rangle$, which denotes the function that assigns value $-1$ to each state $q \in Q$, and with least upper bound operator $\sqcup$ defined as: $f_1 \sqcup f_2=f$ such that $f(q)=\max(f_1(q),f_2(q))$ for all $q \in Q$. This upper bound operator generalizes to any finite set of counting functions $\mathcal{F}=\{f_1,f_2, \dots,f_n\} \subseteq CF(\atm,k)$, and the least upper bound of $\mathcal{F}$ is denoted $\bigsqcup \mathcal{F}$.
\end{definition}

The essence of the structure in $\mathcal{D}(\atm,k)$ is captured in the following series of results.

\begin{lemma}
\label{lem:CFmono}
For all co-B\"uchi automata $\atm$, for all $k \in \mathbb{N}$, for all counting functions $f_1,f_2 \in CF(\atm,k)$, such that $f_1 \preceq f_2$, we have that  $L(\mathcal{D}(\atm,k)[f_2]) \subseteq L(\mathcal{D}(\atm,k)[f_1])$, and thus if $L(\mathcal{D}(\atm,k)[f_2])$ is realizable then $L(\mathcal{D}(\atm,k)[f_1])$ is realizable. 
\end{lemma}
\begin{proof}
As counting functions record the number of visits to accepting states so far, starting from $f_2$ is more constraining than from $f_1$. So any word accepted from $f_2$ is accepted from $f_1$.\qed
\end{proof}

\begin{lemma}[\cite{DBLP:conf/cav/FiliotJR09}]
For all co-B\"uchi automata $\atm$, for all $k \in \mathbb{N}$, for all counting functions $f \in CF(\atm,k)$, it is {\sc ExpTime-Complete} to decide if $L^{\forall}_k(\atm[f])$ is realizable.
\end{lemma}

\begin{corollary}
\label{cor:computeantichain}
For all co-B\"uchi automata $\atm$, for all $k \in \mathbb{N}$, the set of counting functions $W^{\atm}_k=\{ f \in CF(\atm,k) \mid L(\mathcal{D}(\atm,k)[f]) \mbox{~is~realizable~}\}$ is $\preceq$-downward closed and can be represented by the $\preceq$-antichain $\lceil W^{\atm}_k \rceil$ of maximal elements in $W^{\atm}_k$. This set of maximal elements can be computed in exponential time in the size of $\atm$ and the binary encoding of $k$.
\end{corollary}

We now formulate a lemma that will be instrumental, later in this section, to improve the upper bound of the algorithm that solves the $\PP$-realizability problem for universal coB\"uchi specifications.
\begin{lemma}
\label{lem:interup}
For all co-B\"uchi automata $\atm$, for all $k \in \mathbb{N}$, for all sets of counting functions $\mathcal{F}=\{f_1,f_2, \dots,f_n\} \subseteq CF(\atm,k)$: $$L(\mathcal{D}(\atm,k)[\bigsqcup \mathcal{F}]=\bigcap_{f \in \mathcal{F}}L(\mathcal{D}(\atm,k)[f]).$$
\end{lemma}
\begin{proof}
As $CF(\atm,k)$ is finite, it is sufficient to prove that for all $f_1,f_2 \in CF(\atm,k)$, we have
$L(\mathcal{D}(\atm,k)[f_1 \sqcup f_2]=L(\mathcal{D}(\atm,k)[f_1]) \cap L(\mathcal{D}(\atm,k)[f_1]).$ To establish this property, let us consider a word $w \in (\inputs \outputs)^{\omega}$ and the runs on $w$ from $f_1$, $f_2$ and $f_1 \sqcup f_2$. We denote those runs by $r_1$, $r_2$, and $r_{1,2}$, respectively. Let $r_1=g_0 g_1 \dots g_n \dots$ with $g_0=f_1$, $r_2=h_0 h_1 \dots h_n \dots$ with $h_0=f_2$, and $r_{1,2}=l_0 l_1 \dots l_n \dots$ with $l_0=f_1 \sqcup f_2$. It is easy to show by induction, using the definition of $\delta^{\mathcal{D}}$, that for all positions $i \geq 0$, for all $q \in \atm$, we have that $l_i(q)=\max(g_i(q),h_i(q))$ and so $l_i=f_i \sqcup g_i$. Then clearly, we have that $r_{1,2}$ is accepting if and only if both $r_1$ and $r_2$ are accepting. This is because, for $l=f \sqcup g$, we have for $q \in Q$: $l(q)=k+1$ iff $f(q)=k+1$ or $g(q)=k+1$.\qed
\end{proof}

\paragraph{Proof of Theorem~\ref{thm:sizeSolk}} 
We are now ready to provide a proof to the statement.
Given a preMealy $\PP=(M,m_0,\Delta)$, co-B\"uchi automata
$\atm=(Q,q_{{\sf init}},\Sigma,\delta,d)$,  $k \in \mathbb{N}$, we can
compute according to corollary~\ref{cor:computeantichain} the
$\preceq$-antichain $\lceil W^{\atm}_k \rceil \subseteq CF(\atm,k)$ in
exponential time. Then to decide if $\PP$ can be completed into a
(full) Mealy machine that realizes $L(\mathcal{D}(\atm,k))$, we
construct a labelling of states of $\PP$ defined by the function $F^*
: M\rightarrow CF(\atm,k)$, for all $m\in M$, by
$$
F^*(m) = \bigsqcup \{ f \mid \exists u \in (\inputs \outputs)^* \cdot \transS_{\PP}^*(m_0,u) = m \land \delta^{\mathcal{D}}(f_0,u)=f \}
$$

Our goal is now to show that $F^*$ can be computed in polynomial
time.  To do so, we first define the following sequence of functions $(F_j :M \rightarrow CF(\atm,k))_{j \in \mathbb{N}}$:
  \begin{itemize}
      \item for all states $m \in M$ of $\PP$, let $F_0(m)=f_0$ if $m=m_0$, and $F_0(m)=\overline{-1}$ otherwise.
      \item for $j>0$, for all states $m \in M$ of $\PP$, let $F_j(m)=\bigsqcup_{(m_1,\inp,\out,m) \mid \Delta(m_1,\inp)=(\out,m)} \delta^{\mathcal{D}}(F_{j-1}(m_1),(\inp,\out))$, where $\delta^{\mathcal{D}}$ is the transition function of $\mathcal{D}(\atm,k)$.
    \end{itemize}

The following lemma formalizes properties of this sequence of functions.  

\begin{lemma}
The sequence $(F_j :M \rightarrow CF(\atm,k))_{j \in \mathbb{N}}$ satisfies:
  \begin{enumerate}
      \item The sequence stabilizes after at most $|M| \times |Q| \times (k+1)$ steps. We note $G^*$ the function on which the sequence $(F_j)_{j \in \mathbb{N}}$ stabilizes.
      \item Each iteration is computable in time bounded by ${\bf O}(|M|^2 \times |Q|)$.
      \item For all $m \in M$, $F^*=G^*$.
  \end{enumerate}
\noindent
Thus, there is a polynomial time algorithm in $|\PP|$, $|\atm|$, $k
\in \mathbb{N}$, and the size of $\lceil W^{\atm}_k \rceil$ to check
the $\PP$-realizability of $L^{\forall}(\atm)$.
\end{lemma}
\begin{proof}
For point $(1)$, we first note that for all state $m$ of $\MM$ the sequence of counting function $(F_j(\cdot))_{j \in \mathbb{N}}$ stabilizes after $|\MM| \times |Q| \times (k+1)$. Indeed, for all $j \geq 0$, and for all $m \in M$, we have that $F_{j}(m) \preceq F_{j+1}(m)$. As chains in the lattice of counting function $CF(\atm,k)$ has length at most $|Q| \times (k+1)$, each $m$ can be updated at most this number of times. The total number of iterations before stabilization of the $|M|$ state labels defined by the $F_j$ is thus at most $|M| \times |Q| \times (k+1)$.

For point $(2)$, we note that the counting function $F_j(m)$ is
computed as the least upper bound applied of at most $|M|$ counting
functions obtained by applying the transition function of
$\mathcal{D}(\atm,k)$ on counting functions defined by  $F_{j-1}$. It
is important to note that we de not need to construct the entire
automaton $\mathcal{D}(\atm,k)$ for this purpose as we can compute
transitions on-demand based according to Definition~\ref{def:determinization}(3). So the complexity is bounded by ${\bf O}(|M| \times |Q|)$ for updating one state $m$ and thus the overall complexity of one update of all the states is bounded by ${\bf O}(|M|^2 \times |Q|)$.

For $(3)$ we reason by induction, using the definition of
$\Delta$ and $\delta^{\mathcal{D}}$, to prove for all $j \geq 0$
and $m {\in} M$: $$F_j(m)=\bigsqcup \{ f \mid 0 \leq i \leq j \land \exists u \in (\inputs \outputs)^i \cdot  \transS^*(m_0,u) = m \land \delta^{\mathcal{D}}(f_0,u)=f \}.$$\qed
\end{proof}

We now show  that $L(\mathcal{D}(\atm,k))$ is
$\PP$-realizable if and only if there does not exist $m \in M$ such
that $F^*(m) \not \in W^{\atm}_k$. Following the proof of Theorem~\ref{thm:folklore}, we know that $L(\mathcal{D}(\atm,k))$ is $\PP$-realizable iff,
        $L_\omega(\PP)\subseteq L(\mathcal{D}(\atm,k))$ and for every hole
    $h = (p,\inp)$ of $\PP$, there exists $\out_h\in \outputs$ and a Mealy
    machine $\MM_h$ such that for all $u\in
    \textsf{Left}_{p}$, $\MM_h$ realizes $(u\inp\out_h)^{-1} L(\mathcal{D}(\atm,k))$.

First, checking whether $L_\omega(\PP)\subseteq L(\mathcal{D}(\atm,k))$ can be done by verifying that $F^*(m)(q) \not=k+1$ for all $m \in M$. Second, checking the existence of $\MM_h$ is equivalent to check that $\bigcap_{u \inp \in \textsf{Left}_{p}} (u \inp)^{-1} L(\mathcal{D}(\atm,k))$ is realizable. In turn, this is equivalent, by Lemma~\ref{lem:interup} and point $(3)$, to check that $L(\mathcal{D}(\atm,k)[F^*(p)])$ is realizable, which is equivalent to check if $F^*(p) \in W^{\atm}_k$. Both tests can be done in polynomial time in the size of $\atm$ and in the size of $\lceil W^{\atm}_k \rceil$.

We have established that given a preMealy $\PP$,
co-B\"uchi automata $\atm$,  $k \in \mathbb{N}$, and the
$\preceq$-antichain $\lceil W^{\atm}_k \rceil$, we can compute the
sequence of functions $(F_j :M \rightarrow CF(\atm,k))_{j \in
  \mathbb{N}}$ in polynomial time in the size of those inputs, and thus decide (by point $(4))$ if the specification
$L(\mathcal{D}(\atm,k))$ is $\PP$-realizable. According to
Corollary~\ref{cor:computeantichain}, the antichain $\lceil W^{\atm}_k
\rceil$ can be computed in exptime in the size of $\atm$ and
the encoding of $k$. We have thus established the upper-bound of
Theorem~\ref{thm:sizeSolk}. The lower bound is a direct consequence of
Theorem~\ref{thm:sizeSol} which establishes the lower bound for
any specification given as a universal coB\"uchi automaton $\atm$, and
Theorem~\ref{thm:transferpreal}  which reduces the $\PP$-realizability
of $L^\forall(\atm)$ to the $\PP$-realizability of $L_k^\forall(\atm)$
for a $k$ which is exponential in the number of states of
$\atm$. Since $k$ is in binary, its size remains polynomial, so the
reduction is polynomial. This ends the proof of Theorem~\ref{thm:sizeSolk}.
\qed
~~\\

\end{document}

%% file: abstract.tex
\begin{abstract}
    We study a variant of the problem of synthesizing Mealy machines that enforce LTL specifications against \tocheck{all possible behaviours of the environment including hostile ones}.
    In the variant studied here, the user provides the high level LTL specification $\varphi$ of the system to design, and a set $E$ of examples of executions that the solution must produce. 
    Our synthesis algorithm works in two
    phases.
    First, it generalizes the decisions taken along the examples $E$ using tailored extensions of automata learning algorithms.
    This phase generalizes the user-provided examples in $E$ while preserving realizability of $\varphi$.
    Second, the algorithm turns the (usually) incomplete Mealy machine obtained by the learning phase into a complete Mealy machine that realizes $\varphi$.
    The examples are used to guide the synthesis procedure.
    We provide a completness result that shows that our procedure can learn any Mealy machine $M$ that realizes $\varphi$ with a small (polynomial) set of examples.
    We also show that our problem, that generalizes the classical LTL synthesis problem (i.e. when $E=\emptyset$), matches its worst-case complexity.
    The additional cost of learning from $E$ is even polynomial in the size of $E$ and in the size of a symbolic representation of solutions that realize $\varphi$.
    This symbolic representation is computed by the synthesis algorithm implemented in {\sc Acacia-Bonzai} when solving the plain LTL synthesis problem. 
    We illustrate the practical interest of our approach on a set of examples.
\end{abstract}

%% file: DecidingPRealizablity.tex
\tocheck{Brought proof from Appendix to here.}
Before proving Theorem~\ref{thm:sizeSol}, let us note that the $\PP$-realizability problem generalizes the classical
realizability, as the latter is equivalent to the
$\PP_0$-realizability where $\PP_0$ is the preMealy machine composed
of single state (which is initial) without any transition. So, we
inherit the \textsc{ExpTime} lower bound of
Theorem~\ref{thm:folklore}. However, we prove that the
$\PP$-realizability problem is intrinsically harder: indeed, we
show that the \textsc{ExpTime} hardness holds even if $\PP$ is a fixed preMealy machine
and $\spec$ is given as a \emph{deterministic reachability} automaton. This is
in contrast to the classical realizability problem: deciding the
realizability of a specification given as a deterministic reachability
automaton is in
\textsc{PTime}~\cite{DBLP:journals/jacm/ChandraKS81}. Our synthesis
algorithm from specifications and examples extensively rely on
sucessive calls to a $\PP$-realizability checker, for various preMealy
machines $\PP$. However, we show in Sec~\ref{sec:omega-reg} that modulo
pre-computing, in worst-case exponential time, some symbolic (and in
practice compact) representation of some realizable configurations of
the specification automaton, all those calls can be done in
polynomial time in this representation.

\begin{proof} [Proof of \cref{thm:sizeSol}]
We first prove the upper-bound. Let $Q_\PP$ be the set of states of $\PP$, $\trans_\PP$ its transition function and $p_0$ its initial state. For any $p\in Q_\PP$, we define its \emph{left language} $\textsf{Left}_{p}$ as
    $$
    \textsf{Left}_{p} = \{ u\in (I.O)^*\mid \transS_\PP^*(p_0,u) = p \}
    $$
Then, $\PP$-realizability is characterized by the following property:
\begin{claim} 
$\spec$ is $\PP$-realizable iff,
        $L_\omega(\PP)\subseteq \spec$ and for every hole
    $h = (p,\inp)$ of $P$, there exists $\out_h\in \outputs$ and a Mealy
    machine $\MM_h$ such that for all $u\in
    \textsf{Left}_{p}$, $\MM_h$ realizes $(u\inp\out_h)^{-1}\spec$.\footnote{For an alphabet $\Sigma$, a set
      $A\subseteq \Sigma^\omega$ and $u\in \Sigma^*$, $u^{-1}A = \{
      v\in \Sigma^\omega\mid uv\in A\}$.}
\end{claim}
\begin{proof}[Proof of claim] For the 'if' direction, we prove
        that $\PP$ can be extended
    into a Mealy machine $\MM$ which $\PP$-realizes $\spec$ as
    follows: $\MM$ consists of $\PP$ taken in disjoint union, for all
    holes $h$ of $\PP$, with the Mealy machine $\MM_h$, extended with
    the transition $\trans_\MM(h) = (\out_h,\textsf{init}_{h})$ where
    $\textsf{init}_{h}$ is the initial state of $\MM_h$. Clearly,
    $\PP$ is a subgraph of $\MM$. We prove that $\MM$ realizes
    $\spec$. Let $w\in L_\omega(\MM)$. Suppose that $w\not\in \spec$ and
    let us derive a contradiction. Since $L_\omega(\PP)\subseteq
    L^\forall(\atm)$, $w\not\in L_\omega(\PP)$. It implies that the execution
    of $\MM$ on $w$ necessarily visits a hole $h = (p,\inp)$ of $\PP$. So, $w$ can be
    decomposed as $w = u\inp \out_h v$ where $u$ is the longest prefix
    of $w$ such that $u\in \textsf{Left}_p$. Since $w\not\in \spec$, we
    get that $v\not\in (u\inp\out_h)^{-1}\spec$. By definition of $\MM$,
    we have $v\in L_\omega(\MM_h)$, so $\MM_h$ does not realize 
    $(u\inp\out_h)^{-1}\spec$, which is a contradiction.

    Conversely, suppose that $\mathcal{S}$ is $\PP$-realizable by some Mealy machine
    $\MM$. Since $L_\omega(\PP)\subseteq L_\omega(\MM)$ and
    $L_\omega(\MM)\subseteq \spec$, we get $L_\omega(\PP)\subseteq
    \spec$. Now, consider a hole $h = (p,\inp)$. Since $\PP$ is a subgraph of $\MM$, $p$ is a state
    of $\MM$ and since $\trans_\MM$ is total,
    there exists $\out_h\in \outputs$ such that $\trans_\MM(h) =
    (\out_h, p')$ for some state $p'$ of $\MM$. Consider the machine
    $\MM_{p'}$ which is identical to $\MM$ except that its initial
    state is $p'$: $\MM_{p'}$ is a Mealy machine which
    realizes $(u\inp\out_h)^{-1}\spec$ for all $u\in
    \textsf{Left}_p$. Indeed, let $v\in L_\omega(\MM_{p'})$. By
    definition of $\MM_{p'}$, we have $u\inp\out_h v\in
    L_\omega(\MM)\subseteq \spec$. Hence, $v\in (u\inp\out_h)^{-1}\spec$. 
    $\hfill$ .
    \end{proof}

It remains to show that the characterization of the claim can be decided in
    \textsc{ExpTime}. First, deciding whether $L_\omega(\PP)\subseteq
    L^\forall(\atm) = \spec$ is a standard automata inclusion problem. Indeed,
    $\PP$ is can be viewed as a deterministic B\"uchi automaton all
    states of which are accepting, and $\atm$ is a universal
    co-B\"uchi automaton, which can be complemented in linear-time
    into a non-deterministic B\"uchi automaton $\mathcal{B}$. Then, it
    suffices to test whether $L_\omega(\PP)\cap L^\exists(\mathcal{B}) =
    \varnothing$. This is doable in \textsc{PTime} in the size of both
    machines. So, testing whether $L_\omega(\PP)\subseteq L^\forall(\atm) = \spec$
    can be done in \textsc{PTime}.

Now, we want to decide the second part of the characterization.  Note that given a hole $h = (p,\inp)$ and
    $\out_h\in \outputs$, there exists a Mealy machine $\MM_h$ such
    that for all $u\in\textsf{Left}_p$, $\MM_h$ realizes
    $(u\inp\out_h)^{-1}\spec$, iff the specification $\bigcap_{u\in
      \textsf{Left}_p} (u\inp\out_h)^{-1}\spec$ is realizable. Given $h$
    and $\out_h$, we construct in linear-time a universal co-B\"uchi
    automaton recognizing
    $\bigcap_{u\in \textsf{Left}_p} (u\inp\out_h)^{-1}\spec$. First, we compute the set of states 
$$
R^{\atm,\PP}_p =  \{ q\in Q_\atm\mid \exists u\in (I.O)^*, \transS_\PP^*(p_0,u) = p\wedge  \transS_\atm^*(q_0,u\inp\out_h)= q\}
$$
This can be done in ptime. Then, we define the universal co-B\"uchi
automaton denoted $\atm_p$ which is exactly $\atm$ where the set of initial
states is set to $R^{\atm,\PP}_p$. We have $L^\forall(\atm_p) = \bigcap_{u\in
  \textsf{Left}_p} (u\inp\out_h)^{-1}\spec$, and then we use
Theorem~\ref{thm:folklore} to decide, in \textsc{ExpTime} in the size
of $\atm_p$, which is linear in the size of $\atm$, whether
$L^\forall(\atm[p])$ is realizable.

If $\spec$ is $\PP$-realizable, then it is $\PP$-realizable by the machine
$\MM$ as constructed in the proof of the claim. For each hole $h = (p,\inp)$ of
$\PP$, by Theorem~\ref{thm:folklore}, we can bound the size of the
machine $\MM_h$ by $2^{O(n log_2 n)}$ where $n$ is the number of
states of $\atm_p$, which is exactly the number of states of
$\atm$. So, if $\PP$ has $n_h$ holes, $\spec$ is $\PP$-realizable by a
Mealy machine with $m + n_h 2^{O(n log_2 n)}$ states.

For the lower bound, we reduce the problem of deciding whether the
intersection of $n$ languages of finite trees is non-empty, when those
languages are defined by deterministic top-down tree automata. This
problem is known to be \textsc{ExpTime}-c~\cite{tata}. This allows us
to show the lower bound for $\PP$-realizability even for
specifications given by deterministic reachability automata. This is
in contrast to plain realizability, which is solvable in
\textsc{PTime} for this class of
specifications~\cite{DBLP:reference/mc/2018}. Intuitively, high-level reason
why $\PP$-realizability is harder than realizability is because $\PP$
imposes strong constraints on the solution. In particular, it enforces
that the system which $\PP$-realizes $\spec$ behaves the same after
any prefix which reaches the same state of $\PP$. This is why in the
\textsc{ExpTime} solution above one needs to check realizability of
intersection of specifications of the form $\bigcap_{u\in
  \textsf{Left}_p} (u\inp\out_h)^{-1}\spec$, which is a harder problem
than trying to realize monolithic specifications.

We now give the detailed proof to obtain the lower-bound. It reduces the following \textsc{ExpTime}-c
problem~\cite{tata}: given $n$ deterministic top-down tree automata $(\mathcal
T_i)_{i=1}^n$, decide whether $\bigcap_{i=1}^n L(\mathcal T_i)\neq
\varnothing$. The main idea is already captured by the restricted problem
where the $\mathcal T_i$ are DFA, known to be \textsc{PSpace}-c, so we
first expose that case. Let $(\mathcal D_i = (Q_i,in_i,F_i,\delta_i))_{i=1}^n$ be $n$ DFA over
some alphabet $\Sigma$. We let $\inputs = \{ \inp_1,\dots,\inp_n\}$ and $\outputs =
\Sigma\cup \{ \textsf{skip},\textsf{exit}\}$. For all
$j\in\{1,\dots,n\}$, we let $\inputs\otimes L(\mathcal{D}_j)$ the set
of words of the form $\lambda_1\sigma_1\lambda_2\sigma_2\dots
\lambda_k\sigma_k\in (\inputs\outputs)^*$
such that $\sigma_1\dots \sigma_k\in  L(\mathcal{D}_j)$. Consider the following
specification:
$$
\spec = \bigcup_{j=1}^n \{
\inp_j.\textsf{skip}.u.\inp.\textsf{exit}.x\mid u\in\inputs\otimes L(\mathcal
D_j),\inp\in \inputs, x\in (\inputs\outputs)^\omega\}
$$

We also define the following $2$-states preMealy machine $\PP$: from its
initial state $m_0$, whenever it reads $\inp_j$ for any $j=1,\dots,n$,
it outputs $\textsf{skip}$ and move to its second state $m$, which is a
hole. 

We prove that:

\begin{enumerate}
  \item $\spec$ is recognizable by a deterministic reachability automaton
    $\atm_\spec$ of polynomial size
  \item $\spec$ is $\PP$-realizable iff $\bigcap_{i=1}^n L(\mathcal
    D_i)\neq \varnothing$.
\end{enumerate}

First, note that $\spec$ is recognizable by a deterministic reachability
automaton $\atm_\spec$ of polynomial size. Informally, each automaton
$\mathcal{D}_i$ is modified in such a way that any input symbol from
$\inputs$ can be
read in between two output letters, so that it recognizes
$\inputs\otimes L(\mathcal{D}_i)$. Let us write $\inputs\otimes
\mathcal{D}_i$ the modified automaton, and assume all the 
automata $\inputs\otimes \mathcal{D}_i$ have disjoint sets of states. 
From its single initial state, $\atm_\spec$ can
read for all $j=1,\dots,n$ the sequence of two symbols
$\inp_j.\textsf{skip}$ and go the initial state of $\inputs\otimes
\mathcal{D}_j$. Additionally, we add a single state $q_{reach}$, the unique state
to be accepting (in the sense that it has colour $0$ while any other
state has colour $1$). From $q_{reach}$, any sequence is accepting (it
is a trap). Finally, for all accepting states $q_f$ of $\mathcal{D}_j$, and
all inputs $\inp\in \inputs$, we make $\atm_\spec$ transition to
$q_{reach}$ when reading $\inp.\textsf{exit}$ from state $q_f$.

For the second assertion, the main intuitive idea behind its proof is
that $\PP$ transitions to the same state $m$ for any possible initial input
while $\atm_\spec$ transitions to different states. Therefore, $\PP$
enforces that whatever the initial input $i_j$, the same
strategy should be played afterwards, while on the other hand, the
definition of $\spec$ is dependent on the initial input. Formally, suppose
that $\MM$ is a Mealy machine $\PP$-realizing $\spec$. Then, since $\PP$
is a subgraph of $\MM$, the language of $\MM$ is necessarily of the
form 
$$
L_\omega(\MM) = \inputs.\textsf{skip}.L' \hfill\qquad\qquad (1)
$$
for some $L'$ such that $\inputs.\textsf{skip}.L'\subseteq
\spec$. Let $w\in L_\omega(\MM)$. It is necessarily of the form
$w = \inp_j.\textsf{skip}.u.\inp.\textsf{exit}.x$ for some
$j=1,\dots,n$, $u\in (\inputs\otimes L(\mathcal D_j)$,
$\inp\in\inputs$ and $x\in (\inputs\outputs)^\omega$. From $(1)$,
we get that for any other $j'\neq j$, $w' =
\inp_{j'}.\textsf{skip}.u.\inp.\textsf{exit}.x \in L_\omega(\MM)$ and
therefore, $u\in (\inputs\otimes L(\mathcal D_{j'})$. So,
$\bigcap_{i=1}^nL(\mathcal{D}_j)\neq \varnothing$.

The converse is proved similarly: if $v\in
\bigcap_{i=1}^nL(\mathcal{D}_j)$, then to $\PP$-realize $\spec$, it
suffices for the system to play $\textsf{skip}$, then $v$, and then
\textsf{exit} forever. This strategy can easily be described by a
Mealy machine extending $\PP$.

This shows \textsc{PSpace}-hardness. The extension of the latter
reduction to deterministic top-down tree automata (over finite binary $\Sigma$-trees)
is standard: the environment picks the direction $\{1,2\}$ in the tree while the
system picks the labels. We let $\inputs = \{
\inp_1,\dots,\inp_j\}\cup \{1,2\}$ and $\outputs = \Sigma\cup
\{\textsf{exit},\textsf{skip}\}$ as before. The specification $\spec$ is modified as follows:
$\spec = \bigcup_{j=1}^n \spec_j$ where each $\spec_j$ is  the set of words of the
form $\inp_j.\textsf{skip}.u.\inp.\textsf{exit}.x$ such that there
exists finite binary tree $t\in L(\mathcal T_j)$ such that $u$ is a
root-to-leaf branch of $t$, i.e. $u = d_1\sigma_1\dots d_k\sigma_k$
where each $d_i\in\{1,2\}$ is a direction, and each $\sigma_i$ is the
label of the node of $t$ identified by the root-to-node path $d_1\dots
d_i$. The preMealy machine $\PP$ is the same as before, and it is
easily seen that the new specification $\spec$ is definable by a
deterministic reachability automaton of polynomial size: this is due
to the fact that the tree automata are deterministic top-down, and the
path languages of deterministic top-down tree automata are regular,
recognizable by DFA of polynomial size~\cite{tata}.

Let us sketch the correctness of the construction. If $t\in
\bigcap_{i=1}^n L(\TT_i)$, then $\PP$ can be extended into a full
Mealy machine which after the state $m$ exactly mimics the structure of $t$:
states are paths in $t$ and when getting a new direction as input, it
outputs the label of $t$ reached following that direction. If instead,
the current path is a leaf of $t$, then the Mealy machine, whatever it
receives as input in the future, outputs $\textsf{exit}$
foreover. This machine is guaranteed to realize the specification,
because whatever the initial input, all the branches of the tree
induced by the choices of the environment are accepted by all the tree
automata.

Conversely, if there is a Mealy machine $\MM$ extending $\PP$ and
realizing the specification, then whatever the initial input, it plays
the same strategy afterwards. It is then possible to reconstruct a
tree accepted by all tree automata using the
choices made by the environment (directions), which describe paths in
the tree, and the choices made by the system, which correspond to
the labels of nodes identified by those paths. Since \textsf{exit}
must eventually be output on all outcomes, the tree construct in such
a way is guaranteed to be finite.
\end{proof}

%% file: phase1.tex
This phase exploits the examples by generalizing them as much as possible while maintaining realizability of  the specification. It outputs a preMealy machine which is consistent with the examples and realizes the specification, if it exists. It is an RPNI-like learning algorithm~\cite{DBLP:series/synthesis/2015Heinz,DBLP:journals/sttt/GiantamidisTB21} which includes specific tests to maintain realizability of the specification.

\tocheck{The first step of this phase involves building a tree-shaped preMealy machine whose accepted language is exactly the set of prefixes $\prefs(E)$ of the given set of examples $E$, called a \emph{prefix-tree acceptor} (PTA).
Formally, we define PTA as follows:}

\paragraph{Prefix Tree Acceptor} \tocheck{A set $E\subseteq
(\inputs\outputs)^*$ (not necessarily finite) is \emph{consistent} if for all
$e\in\prefs(E)\cap (\inputs\outputs)^*\inputs$, there exists a unique
output denoted
$\out_E(e)\in\outputs$ such that $e.\out_E(e)\in \prefs(E)$. When $E$ is
consistent and finite, we can canonically associate with $E$ a tree-shaped preMealy
machine denoted $\textsf{PTA}(E)$ such that $L(\textsf{PTA}(E)) =
\prefs(E)\cap (\inputs\outputs)^*$, as follows:
$$
\textsf{PTA}(E) = (\prefs(E)\cap(\inputs\outputs)^*, \epsilon,
(e,\inp)\mapsto (\out_E(e\inp), e\inp\out_E(e\inp)))
$$}
\begin{example}\label{ex:PTA}
      \tocheck{Let $\inputs = \{ \inp,\inp'\}$ and $\outputs = \{\out,\out'\}$
    and consider $E_0 = \{\inp'\out,\inp\out\inp\out\inp'\out'\}$. Then
    $E_0$ is consistent and  $\textsf{PTA}(E_0)$ is depicted on
    the left of Figure~\ref{fig:PTA}. For conciseness, we denote its states by
    $0,\dots,4$ where $0=\epsilon$, $1=\inp'\out$, $2=\inp\out$,
    $3=\inp\out\inp\out$ and $4=\inp\out\inp\out\inp'\out'$.}
\end{example}
\begin{figure}[t]
    \begin{tabular}{cc}
      \begin{minipage}{0.45\linewidth}
          \centering
    \begin{tikzpicture}[->,>=stealth',shorten >=1pt,auto,node distance=1.8cm,
                    thick,inner sep=0pt,minimum size=0pt]
  \tikzstyle{every state}=[fill=gray!30,text=black,inner sep=4pt,minimum size=12pt]

        \node[state,initial] (q0) {$0$};
        \node[state] (q1) at (1.5, -1) {$1$};
        \node[state] (q2) at (1.5, 1) {$2$};
        \node[state] (q3) at (3, 1) {$3$};
        \node[state] (q4) at (4.5, 1) {$4$};
        \path (q0) edge[left] node[xshift=1mm,yshift=-2mm]{$\inp'/\out$} (q1) ;
        \path  (q0) edge[right] node[xshift=-2mm,yshift=-3mm]{$\inp/\out$} (q2) ;
        \path (q2) edge[right] node[xshift=-2mm,yshift=2mm]{$\inp/\out$} (q3) ;
        \path (q3) edge[right] node[xshift=-3mm,yshift=2mm]{$\inp'/\out'$} (q4) ;
    \end{tikzpicture}
\end{minipage}
      &
    \begin{minipage}{0.45\linewidth}
    \begin{tikzpicture}[->,>=stealth',shorten >=1pt,auto,node distance=1.8cm,
                    thick,inner sep=0pt,minimum size=0pt]
  \tikzstyle{every state}=[fill=gray!30,text=black,inner sep=1pt,minimum size=26pt]

        \node[state, initial] (q0) {$0,1$};
        \node[state] (q1) at (3,0) {$2,3,4$};
        \draw
            (q0) edge[loop above] node[yshift=1mm]{$\inp'/\out$} (q0)
            (q0) edge[above] node{$\inp/\out$} (q1)
            (q1) edge[loop above] node[yshift=1mm]{$\inp'/\out'$} (q1)
            (q1) edge[loop below] node[yshift=-1mm]{$\inp/\out$} (q1);
    \end{tikzpicture}
\end{minipage}
\end{tabular}

    \caption{The preMealy machine
      $\textsf{PTA}(\{\inp'\out,\inp\out\inp\out\inp'\out'\})$ of
      Example~\ref{ex:PTA} and its quotient by the equivalence
      relation induced by the partition $\{\{0,1\},\{2,3,4\}\}$ as
      described in Example~\ref{ex:congruence}.\label{fig:PTA}}
\end{figure}
\tocheck{In the next step of this phase, the algorithm tries to merge as many as possible states of the PTA.} 
The strategy used to select a state to merge another given state with, is a parameter of the algorithm, and is called a \emph{merging strategy} $\sigma_G$. 
Formally, a \emph{merging} strategy $\sigma_G$ is defined over $4$-tuples $(\MM,m,E,X)$ where $\MM$ is a preMealy machine, $m$ is a state of $\MM$, $E$ is a set of examples and $X$ is subset of states of $\MM$ (the candidate states to merge $m$ with), and returns a state of $X$, i.e., $\sigma_G(\MM,m,E,X)\in X$. \tocheck{The formal definition is as follows:
\paragraph{State merging} We now define the classical state merging operation of RPNI adapted to Mealy machines. 
An equivalence relation $\sim$ over $M$ is called \emph{a congruence} for $\MM$ if for all $x\sim x'$ and $\inp\in\inputs$, if $\trans_\MM(x,\inp)$ and $\trans_\MM(x',\inp)$ are both defined, then $\transS_\MM(x,\inp)\sim \transS_\MM(x',\inp)$. 
It is \emph{Mealy-congruence} for $\MM$ if additionally, $\transO_\MM(x,\inp)=\transO_\MM(x',\inp)$.
When $\MM$ is clear from the context, we simply say congruence and Mealy-congruence.
If $\sim$ is an Mealy-congruence, then the following preMealy machine (called the quotient of $\MM$ by $\sim$) is a well-defined preMealy machine (it does not depend on the choice of representatives): $\MM/_{\sim} = (M/_{\sim}, [m_{\textsf{init}}], ([s],\inp)\mapsto (\transO(s,\inp), [\transS(s,\inp)]))$. In this  definition, $[s]$ denotes the class of $s$ by $\sim$, and we take a representative $s$ such that $\trans(s,\inp)$ is defined. If no such representative exists, the transition is undefined on $\inp$.}

The pseudo-code for Phase~1 is given by Algo~\ref{algo:gen}. 
\tocheck{We provide here a running example to better illustrate the working of algorithm.}
Initially, the algorithm tests whether the set of examples $E$ is consistent\footnote{$E$ is consistent if outputs uniquely depends on prefixes. 
Formally, it means for all prefixes $u\in\prefs(E)\cap (\inputs\outputs)^*\inputs$, there is a unique output $\out\in \outputs$ such that $u\out\in \prefs(E)$.} and if that is the case, whether $\textsf{PTA}(E)$ can be completed into a Mealy machine realizing the given specification $\spec$, thanks to Theorem~\ref{thm:sizeSol}.
\begin{figure}[t]
    \begin{tabular}{cc}
    \begin{minipage}{\linewidth}
    \centering
    \begin{tikzpicture}[->,>=stealth',shorten >=1pt,auto,node distance=1.8cm,thick,inner sep=0pt,minimum size=0pt]
        \tikzstyle{every state}=[fill=gray!30,text=black,inner sep=4pt,minimum size=12pt]

        \node[state,initial] (q0) {$0$};
        \node[state, below left of=q0] (q1) {$1$};
        \node[state, below right of=q0] (q2) {$2$};
        \node[state, below of=q1] (q3) {$3$};
        \node[state, below of=q2] (q4) {$4$};
        \node[state, below of=q3] (q5) {$5$};
        \path (q0) edge[left] node[xshift=-2mm] {$\neg r_1 \land \neg r_2/\neg g_1 \land \neg g_2$} (q1);
        \path  (q0) edge[right] node[xshift=2mm] {$r_1 \land r_2/ g_1 \land \neg g_2$} (q2);
        \path (q1) edge[left] node[xshift=-2mm] {$r_1 \land \neg r_2/g_1 \land \neg g_2$} (q3) ;
        \path (q2) edge[right] node[xshift=2mm] {$\neg r_1 \land \neg r_2/\neg g_1 \land g_2$} (q4);
        \path (q3) edge[left] node[xshift=-2mm] {$\neg r_1 \land r_2/\neg g_1 \land g_2$} (q5) ;
    \end{tikzpicture}
    \end{minipage}
    \end{tabular}
    \caption{The preMealy machine
      $\textsf{PTA}$ of Example~\ref{ex:mutexPTA}. Here, we find that $\varphi^{{\sf ME}}_{{\sf CORE}}$ is $\textsf{PTA-realizable}$ \label{fig:mutexPTA}}
\end{figure}
\begin{example} [Synthesis from $\varphi^{{\sf ME}}_{{\sf CORE}}$ and examples]
    \label{ex:mutexPTA}
    \tocheck{
    Let us consider the classical problem of mutual exclusion described in Example~\ref{ex:mutexTraces} with the LTL specification, $\varphi^{{\sf ME}}_{{\sf CORE}}$, and the prefixes of executions: 
    \begin{itemize}
        \item[$(1)$] $\{!r_1,!r_2\} . \{!g_1,!g_2\} \#\{r_1,!r_2\} . \{g_1,!g_2\} \# \{!r_1,r_2\} . \{!g_1,g_2\}$
        \item[$(2)$] $\{r_1,r_2\} . \{g_1,!g_2\} \# \{!r_1,!r_2\} . \{!g_1,g_2\}$
    \end{itemize}
    We begin by building the $\textsf{PTA}$ as shown in Fig.~\ref{fig:mutexPTA} and then check if  $\varphi^{{\sf ME}}_{{\sf CORE}}$ is ${\sf PTA-realizable}$.}
\end{example}
If that is the case, then it takes all prefixes of $E$ as the set of examples, and enters a loop which consists in iteratively coarsening again and again some congruence $\sim$ over the states of $\textsf{PTA}(E)$, by merging some of its classes. 
The congruence $\sim$ is initially the finest equivalence relation.
It does the coarsening in a specific order: examples (which are states of $\textsf{PTA}(E)$) are taken in length-lexicographic order.
When entering the loop with example $e$, the algorithm computes at line~\ref{line:mergeable} all the states, i.e., all the examples $e'$ which have been processed already by the loop ($e'\prec_{ll} e$) and whose current class can be merged with the class of $e$ (predicate $\textsf{Mergeable}(\textsf{PTA}(E),\sim,e,e')$).
\tocheck{State merging is a standard operation in automata learning algorithms} which intuitively means that merging the $\sim$-class of $e$ and the $\sim$-class of $e'$, and propagating this merge to the descendants of $e$ and $e'$, does not result any conflict. 
At line~\ref{line:real}, it filters the previous set by keeping only the states which, when merged with $e$, produce a preMealy machine which can be completed into a Mealy machine realizing $\spec$ (again by Theorem~\ref{thm:sizeSol}).
If after the filtering there are still several candidates for merge, one of them is selected with the merging strategy $\sigma_G$ and the equivalence relation is then coarsened via class merging (operation $\textsf{MergeClass}(\textsf{PTA}(E),\sim,e,e')$). 
At the end, the algorithm returns the quotient of $\textsf{PTA}(E)$ by the computed Mealy-congruence. 
As a side remark, when $\spec$ is universal, i.e. $\spec = (\inputs\outputs)^\omega$, then it is realizable by \emph{any} Mealy machine and therefore line~\ref{line:real} does not filter any of the candidates for merge.
So, when $\spec$ is universal, Algo~\ref{algo:gen} can be seen as an RPNI variant for learning preMealy machines. 

\begin{figure}[t]
    \centering
    \begin{subfigure}{0.45\linewidth}
    \centering
    \begin{tikzpicture}[->,>=stealth',shorten >=1pt,auto,node distance=1.8cm,thick,inner sep=0pt,minimum size=0pt]
        \tikzstyle{every state}=[fill=gray!30,text=black,inner sep=4pt,minimum size=12pt]

        \node[state,initial] (q0) {$0, 1$};
        \node[state, below left of=q0] (q3) {$3$};
        \node[state, below right of=q0] (q2) {$2$};
        \node[state, below of=q2] (q4) {$4$};
        \node[state, below of=q3] (q5) {$5$};
        \path (q0) edge[loop above, above, align=center] node[xshift=-2mm] {$\neg r_1 \land \neg r_2/$ \\ $\neg g_1 \land \neg g_2$} (q0);
        \path  (q0) edge[right, align=center] node[xshift=2mm, yshift=2mm] {$r_1 \land r_2/$ \\ $g_1 \land \neg g_2$} (q2);
        \path (q0) edge[left, align=center] node[xshift=-2mm, yshift=2mm] {$r_1 \land \neg r_2/$ \\ $g_1 \land \neg g_2$} (q3);
        \path (q2) edge[right, align=center, yshift=2mm] node[xshift=2mm] {$\neg r_1 \land \neg r_2/$ \\ $\neg g_1 \land g_2$} (q4);
        \path (q3) edge[left, align=center] node[xshift=-2mm] {$\neg r_1 \land r_2/$ \\ $\neg g_1 \land g_2$} (q5);
    \end{tikzpicture}
    \caption{We begin by merging states $0$ and $1$ of the preMealy machine ${\sf PTA}$, i.e., we merge classes $[\epsilon]$ and $[\{\neg r_1 \land \neg r_2\}\{\neg g_1 \land \neg g_2\}]$. We then check for ${\sf PTA-realizability}$ which is found to be true. The resulting machine is shown here.}
    \end{subfigure}
    \hfill
    \begin{subfigure}{0.45\linewidth}
    \centering
    \begin{tikzpicture}[->,>=stealth',shorten >=1pt,auto,node distance=1.8cm,thick,inner sep=0pt,minimum size=0pt]
        \tikzstyle{every state}=[fill=gray!30,text=black,inner sep=4pt,minimum size=12pt]

        \node[state,initial] (q0) {$0, 1, 2$};
        \node[state, below left of=q0] (q3) {$3$};
        \node[state, below right of=q0] (q4) {$4$};
        \node[state, below of=q3] (q5) {$5$};
        \path (q0) edge[loop above, above, align=center] node[xshift=-2mm] {$\neg r_1 \land \neg r_2/\neg g_1 \land \neg g_2$ \\ $r_1 \land r_2/ g_1 \land \neg g_2$} (q0);
        \path (q0) edge[left, align=center] node[xshift=-2mm, yshift=2mm] {$r_1 \land \neg r_2/$ \\ $g_1 \land \neg g_2$} (q3) ;
        \path (q0) edge[right, align=center, yshift=2mm] node[xshift=2mm, yshift=2mm] {$\neg r_1 \land \neg r_2/$ \\ $\neg g_1 \land g_2$} (q4);
        \path (q3) edge[left, align=center] node[xshift=-2mm] {$\neg r_1 \land r_2/$ \\ $\neg g_1 \land g_2$} (q5);
    \end{tikzpicture}
    \caption{We then proceed by merging states $\{0, 1\}$ and $2$ of the preMealy machine ${\sf PTA}$, i.e., we merge classes $[\epsilon]$ and $[r_1 \land r_2 / g_1 \land \neg g_2]$. We then check for ${\sf PTA-realizability}$ which is found to be false. We corraborate by observing the trace $(r_1 \land r_2 / g_1 \land \neg g_2)^\omega$ does not satisfy the LTL subformula $G(r_2 \implies F g_2)$. Thus the merge is unsuccesful and is reversed.}
    \end{subfigure}
    \caption{The merging phase of preMealy machine
      $\textsf{PTA}$ of Example~\ref{ex:mutexPTA}.}
    \label{fig:mutexMerge}
\end{figure}
\begin{figure}[t]
    \centering
    \begin{subfigure}{0.32\linewidth}
    \centering
    \resizebox{1\linewidth}{!}{
    \begin{tikzpicture}[->,>=stealth',shorten >=1pt,auto,node distance=1.8cm,thick,inner sep=0pt,minimum size=0pt, scale=0.5]
        \tikzstyle{every state}=[fill=gray!30,text=black,inner sep=4pt,minimum size=12pt]

        \node[state,initial] (q0) {$0, 1, 3$};
        \node[state, below left of=q0] (q5) {$5$};
        \node[state, below right of=q0] (q2) {$2$};
        \node[state, below of=q2] (q4) {$4$};
        \path (q0) edge[loop above, above, align=center] node[xshift=-2mm] {$\neg r_1 \land \neg r_2/$ \\ $\neg g_1 \land \neg g_2$ \\ $r_1 \land \neg r_2/$ \\ $g_1 \land \neg g_2$} (q0);
        \path  (q0) edge[right, align=center] node[xshift=2mm, yshift=2mm] {$r_1 \land r_2/$ \\ $g_1 \land \neg g_2$} (q2);
        \path (q2) edge[right, align=center, yshift=2mm] node[xshift=2mm] {$\neg r_1 \land \neg r_2/$ \\ $\neg g_1 \land g_2$} (q4);
        \path (q0) edge[left, align=center] node[xshift=-2mm, yshift=2mm] {$\neg r_1 \land r_2/$ \\ $\neg g_1 \land g_2$} (q5);
    \end{tikzpicture}}
    \caption{We then merge states $\{0, 1\}$ and $3$ of the preMealy machine ${\sf PTA}$, i.e., we merge classes $[\epsilon]$ and $[\{\neg r_1 \land \neg r_2\}\{\neg g_1 \land \neg g_2\}\#\{r_1 \land \neg r_2\}\{g_1 \land \neg g_2\}]$. We then check for ${\sf PTA-realizability}$ which is found to be true. The resulting machine is shown here.}
    \end{subfigure}
    \hfill
    \begin{subfigure}{0.32\linewidth}
    \centering
    \resizebox{1\linewidth}{!}{
    \begin{tikzpicture}[->,>=stealth',shorten >=1pt,auto,node distance=1.8cm,thick,inner sep=0pt,minimum size=0pt]
        \tikzstyle{every state}=[fill=gray!30,text=black,inner sep=4pt,minimum size=12pt]

        \node[state,initial] (q0) {$0, 1, 3$};
        \node[state, below left of=q0] (q5) {$5$};
        \node[state, below right of=q0] (q2) {$2$};
        \path (q0) edge[loop above, above, align=center] node[xshift=-2mm] {$\neg r_1 \land \neg r_2/$ \\ $\neg g_1 \land \neg g_2$ \\ $r_1 \land \neg r_2/$ \\ $g_1 \land \neg g_2$} (q0);
        \path  (q0) edge[bend right, left, align=center] node {$r_1 \land r_2/$ \\ $g_1 \land \neg g_2$} (q2);
        \path (q2) edge[bend right, right, align=center] node {$\neg r_1 \land \neg r_2/$ \\ $\neg g_1 \land g_2$} (q0);
        \path (q0) edge[left, align=center] node[xshift=-2mm, yshift=2mm] {$\neg r_1 \land r_2/$ \\ $\neg g_1 \land g_2$} (q5);
    \end{tikzpicture}}
    \caption{We then merge states $\{0, 1, 3\}$ and $4$ of the preMealy machine ${\sf PTA}$, i.e., we merge classes $[\epsilon]$ and $[\{r_1 \land r_2\}\{g_1 \land \neg g_2\}\#\{\neg r_1 \land \neg r_2\}\{\neg g_1 \land g_2\}]$. We then check for ${\sf PTA-realizability}$ which is found to be true. The resulting machine is shown here.}
    \end{subfigure}
    \hfill
    \begin{subfigure}{0.32\linewidth}
    \centering
    \resizebox{1\linewidth}{!}{
    \begin{tikzpicture}[->,>=stealth',shorten >=1pt,auto,node distance=1.8cm,thick,inner sep=0pt,minimum size=0pt]
        \tikzstyle{every state}=[fill=gray!30,text=black,inner sep=4pt,minimum size=12pt]

        \node[state,initial, align=center] (q0) {$0, 1$ \\ $3, 4$};
        \node[state, right of=q0] (q2) {$2$};
        \path (q0) edge[loop above, above, align=center] node[xshift=-2mm] {$\neg r_1 \land \neg r_2/\neg g_1 \land \neg g_2$ \\ $r_1 \land \neg r_2/g_1 \land \neg g_2$ \\ $\neg r_1 \land r_2/\neg g_1 \land g_2$} (q0);
        \path  (q0) edge[bend right, below, align=center] node {$r_1 \land r_2/$ \\ $g_1 \land \neg g_2$} (q2);
        \path (q2) edge[bend right, above, align=center] node {$\neg r_1 \land \neg r_2/$ \\ $\neg g_1 \land g_2$} (q0);
    \end{tikzpicture}}
    \caption{We finally merge states $\{0, 1, 3, 4\}$ and $5$ of the preMealy machine ${\sf PTA}$, i.e., we merge classes $[\epsilon]$ and $[\{\neg r_1 \land \neg r_2\}\{\neg g_1 \land \neg g_2\}\#\{r_1 \land \neg r_2\}\{g_1 \land \neg g_2\}]\#\{\neg r_1 \land r_2\}\{\neg g_1 \land g_2\}]$. We then check for ${\sf PTA-realizability}$ which is found to be true. The resulting machine is shown here.}
    \end{subfigure}
    \caption{The merging phase of preMealy machine
      $\textsf{PTA}$ of Example~\ref{ex:mutexPTA} contd.}
    \label{fig:mutexMerge2}
\end{figure}
\paragraph{Example~\ref{ex:mutexPTA} contd: Synthesis from $\varphi^{{\sf ME}}_{{\sf CORE}}$ and examples}
    \tocheck{
    We note that each state $m$ of the ${\sf PTA}$ in Fig.~\ref{fig:mutexPTA} are $\sim$-class $e$, where $e$ is the shortest prefix such that $\Delta(q_{\sf init}, e) = m$.
    We then check if $\varphi^{{\sf ME}}_{{\sf CORE}}$ is ${\sf PTA-realizable}$\footnote{Refer \textit{Checking ${\sf PTA-realizablity}$ of a specification $S$} in Section~\ref{sec:omega-reg}} which we find to be the case.
    We note that each state is labelled in the length-lexicographic order.
    We then begin the process of merging states in the aforementioned order as shown in Fig.~\ref{fig:mutexMerge} and Fig.~\ref{fig:mutexMerge2}.}

\begin{algorithm}[ht]
    \DontPrintSemicolon
    \KwIn{A finite set of examples $E\subseteq (\inputs.\outputs)^*$, a
      specification $\spec\subseteq (\inputs.\outputs)^\omega$ given as a
      deterministic safety automaton, a merging
      strategy $\sigma_G$}
    \KwOut{A preMealy machine $\MM$ s.t. $E\subseteq L(\MM)$ and
      $\spec$ is $\MM$-realizable, if it exists, otherwise UNREAL. }

\textbf{if} $E$ is not consistent or $\spec$ is not
\textsf{PTA}($E$)-realizable\textbf{ then return} UNREAL

    $E \gets \prefs(E)\cap (\inputs\outputs)^*$;\label{line:closure}

    $\sim \gets \{ (e,e)\mid e\in E\}$;\tcp*{$\sim = diag_E$}

    \For{$e\in E$ in length-lexicographic order $\preceq_{ll}$\label{line:loop}}{

      $mergeCand\gets \{ e'\mid \textsf{Mergeable}(\textsf{PTA}(E),\sim,e,e')\wedge e'\prec_{ll}
      e\}$ \label{line:mergeable}

      $mergeCand\gets \{e'\in mergeCand\mid \spec\text{ is } \textsf{MergeStates}(\textsf{PTA}(E),\sim,e,e'){-}realizable\}$\label{line:real}

      \If{$mergeCand\neq\varnothing$}{
      $e'\gets \sigma_G(\MM,e,mergeCand)$\label{line:choice}
      
      $\sim\gets \textsf{MergeClass}(\textsf{PTA}(E),\sim,e,e')$\label{line:merge}
    }
  }

    \textbf{return} $\textsf{PTA}(E)/_{\sim}$

    \caption{GEN($E$,$\spec$,$\sigma_G$)  -- generalization algorithm}
 \label{algo:gen}
\end{algorithm}

%% file: phase2.tex
 \begin{algorithm}[h]

     \DontPrintSemicolon
    \KwIn{A specification $\spec\subseteq (\inputs.\outputs)^*$ given as a
      deterministic safety automaton, a finite
      set of examples $E\subseteq (\inputs.\outputs)^*$, a
      generalizing and a completion
      strategies $\sigma_G, \sigma_C$}
    \KwOut{A Mealy machine $\MM$ such that $E\subseteq L(\MM)$ and
      $\MM$ realizes $\spec$ if it exists, otherwise UNREAL. }

    \If{$\textnormal{\text{\textsc{Gen}}}(E,\spec,\sigma_G)\neq$UNREAL}{
    $\MM_0\gets \text{\textsc{Gen}}(E,\spec,\sigma_G)$ \tcp*{Returns a preMealy
      machine generalizing the set of examples according to $\sigma_G$ and such that $\spec$ is $\MM_0$-realizable}
  }
  \Else
  {
    \textbf{return} UNREAL
  }

  \If{$\textnormal{\text{\textsc{Comp}}}(\MM_0,\spec,\sigma_C)\neq $UNREAL}{
    $\MM\gets \text{\textsc{Comp}}(\MM_0,\spec,\sigma_C)$\tcp*{Complete $\MM_0$
      by creating new states or reusing states according to $\sigma_C$
   }

    \textbf{return} $\MM$

  }
  \Else
  {
    \textbf{return} UNREAL
  }
\caption{\textsc{SynthSafe}($E$,$\spec$,$\sigma_G$,$\sigma_C$) -- synthesis algorithm from specification and examples      \label{algo:learning}}
  \end{algorithm}
The two-phase synthesis algorithm
for safety specifications and examples, called \textsc{SynthSafe}$(E, \spec, \sigma_G,\sigma_C)$
works as follows: it takes as input a set of examples $E$, a specification
$\spec$ given as a deterministic safety automaton,  a generalizing and
completion strategies $\sigma_G,\sigma_C$ respectively. It returns a
Mealy machine $\MM$ which realizes $\spec$ and such that $E\subseteq
L(\MM)$ if it exists. In a first
steps, it calls \textsc{Gen}$(E,\spec,\sigma_G)$. If this calls
returns UNREAL, then \textsc{SynthSafe} return UNREAL as
well. Otherwise, the call to \textsc{Gen} returns a preMealy machine
$\MM_0$. In a second step, \textsc{SynthSafe} calls
\textsc{Comp}$(\MM_0,\spec,\sigma_C)$. If this call returns UNREAL, so
does \textsc{SynthSafe}, otherwise \textsc{SynthSafe} returns the
Mealy machine computed by \textsc{Comp}. The pseudo-code of
\textsc{SynthSafe} can be found in Algo.~\ref{algo:learning}.

The completion procedure  may
not terminate for some completion strategies. It is because the
completion strategy could for instance keep on selecting pairs of
the form $(\out,m')$ where $m'$ is a fresh state. However we prove
that it always terminates for \emph{lazy} completion strategies. A completion strategy $\sigma_C$ is said to be \emph{lazy} if it
favours existing states, which formally means that
if $X\setminus (\outputs\times \{\textsf{fresh}\})\neq \varnothing$, then 
$\sigma_C(\MM,m,\inp,X)\not\in \outputs\times
\{\textsf{fresh}\}$. The first theorem establishes  correctness and
termination of the algorithm for lazy completion strategies (we assume
that the functions $\sigma_G$ and $\sigma_C$ are computable in
worst-case exponential time \tocheck{in the size of their inputs}).

\begin{thm}[termination and correctness]\label{thm:terminationcorrectness}
    For all finite sets of examples $E\subseteq (\inputs.\outputs)^*$,
    all specifications $\spec\subseteq
    (\inputs.\outputs)^\omega$ given as a deterministic safety
    automaton $\atm$ with $n$ states, all merging strategies $\sigma_G$ and all completion
    strategies $\sigma_C$, if \textsc{SynthSafe}($E,\spec,\sigma_G,\sigma_C$) terminates
    then, it returns a Mealy machine $\MM$ such that $E\subseteq L(\MM)$ and
    $\MM$ realizes $\spec$, if it exists, otherwise it returns UNREAL.
    Moreover, \textsc{SynthSafe}($E,\spec,\sigma_G,\sigma_C$)
    terminates if $\sigma_C$ is lazy, in worst-case
    exponential time (polynomial in the
    size\footnote{The size of $E$ is the sum of the lengths of the
      examples of $E$.}
    of $E$ and
    exponential in $n$).  
\end{thm}

The proof of the latter theorem is a consequence of several results
proved on the generalization and completion phases, and is given in
App.~\ref{sec:proofs}. Intuitively, the complexity is dominated by
the complexity of checking $\PP$-realizability
(Theorem~\ref{thm:sizeSol}) and the termination time of the completion
procedure, which we prove to be worst-case exponential in $n$. The
assumption that the specification is a determinsitic safety automaton
$\atm$ is used
when proving termination of the completion algorithm. Intuitively, to any state $m$
of the so far constructed preMealy machine $\MM$, we associate the subset of
states $Q_m$ of $\atm$ which are reachable in $\atm$ when reading
prefixes that reach $m$ in $\MM$. We prove that when a transition to a
fresh state $m'$ is added to $\MM$ and $Q_{m'}\subseteq Q_m$, then
$m$ could have been reused instead of $m'$
(Lemma~\ref{lem:keylemmatermination} in App.~\ref{app:comple}). This is possible as such
subsets are sufficient to summarize the behaviour of $\atm$ on
infinite suffixes, because it is a safety condition. We also show some
monotonicity property of the subsets $Q_m$ when more transitions are
added to $\MM$, allowing to bound the termination time by the length
of the longest chain of $\subseteq$-antichains of subsets, which is
worst-case exponential in the number of states of $\atm$ (Lemma~\ref{lem:terminationintermediate}
in App.~\ref{app:comple}). 

A Mealy machine $\TT$ is minimal if for all Mealy machine $\MM$ such that $L(\TT) = L(\MM)$, the number of states of $\MM$ is at least that of $\TT$.
The next result, proved in App.~\ref{app:proofmealycompleteness}, states that any minimal Mealy machine realizing a specification $\spec$ can be returned by our synthesis algorithm, providing representative examples.

\begin{thm}[Mealy completeness]\label{thm:mealycompleteness}
    For all specifications $\spec\subseteq
    (\inputs.\outputs)^\omega$ given as a deterministic safety
    automaton, for all minimal Mealy machines $\MM$ realizing $\spec$, there exists a finite set of examples
    $E\subseteq (\inputs.\outputs)^*$, of size polynomial in the size
    of $\MM$, such that for all generalizing strategies $\sigma_G$ and
    completion strategies $\sigma_C$, and all sets of examples $E'$ s.t. $E\subseteq E'\subseteq L(\MM)$, 
    \textsc{SynthSafe}($E',\spec,\sigma_G,\sigma_C) = \MM$.
\end{thm}

The polynomial upper bound given in the statement of
Theorem~\ref{thm:mealycompleteness} is more precisely the following:
the cardinality of $E$ is $O(m+n^2)$ where $n$ is the number of
states of $\MM$ while $m$ is its number of transitions. Moreover, each
example $e\in E$ has length $O(n^2)$. \tocheck{More details can be found in
Remark~\ref{rem:complexity}.}

\begin{remark}\label{rem:complexity}
    We bound here the size of the characteristic sample $E_\TT$. Let
    $n$ and $m$ be the number of states and transitions of $\TT$
    respectively. Then, for all states $t$, $s_t$ has length at most $n-1$, and so for
    all $p = (t,\inp)$ such that $\trans(p)$ is defined, $e_p = \fIO^\TT(s_t)$ has length at most
    $2n$. Given two different states $t\neq t'$, $d_{t,t'}$ has length
    at most $n^2$. Therefore, $v_{t,t'}$ has length at
    most $2(n+n^2)$. There are at most $m$ words $e_p$ and $n^2$
    words $v_{t,t'}$. So overall, the cardinality of
    $E_\TT$ is bounded by $m+n^2$ and its size is bounded by
    $mn+2(n^3+n^4)$.
\end{remark}

%% file: SynthesisOmegaRegular.tex
\label{sec:omega-reg}

We now consider the case where the specification $\spec$ is given as universal coB\"uchi automaton, in Section~\ref{subsec:omega}.
We consider this class of
specifications as it is complete for $\omega$-regular languages and allow for compact symbolic representations.
\tocheck{Further in this section, we consider the case of LTL specifications.}

\paragraph{Specifications given as universal coB\"uchi automata}\label{subsec:omega}
Our solution for $\omega$-regular specifications relies on a reduction
to the safety case treated in Sec.~\ref{sec:learningframework}. It 
relies on previous works that develop so called Safraless algorithms
for $\omega$-regular reactive
synthesis~\cite{DBLP:conf/focs/KupfermanV05,DBLP:conf/atva/ScheweF07a,DBLP:conf/cav/FiliotJR09}. The
main idea is to strengthen the acceptance condition of the automaton
from coB\"uchi to $K$-coB\"uchi, which is a safety acceptance
condition. It is complete for the plain synthesis problem (w/o
examples) if $K$ is large enough (in the worst-case exponential
in the number of states of the automaton, see for
instance~\cite{DBLP:conf/cav/FiliotJR09}). Moreover, it allows for
incremental synthesis algorithms: if the specification defined by the
automaton with a $k$-coB\"uchi acceptance condition is realizable, for
$k\leq K$, so is the specification defined by taking $K$-coB\"uchi
acceptance.
Here, as we also take examples into account, we need to slightly
adapt the results.

\begin{thm}\label{thm:transferpreal}
Given a universal co-B\"uchi automaton $\atm$ with $n$ states defining a specificaton $\spec = L^{\forall}(\atm)$ and a preMealy machine $\PP$ with $m$ states, we have that $\spec$ is $\PP$-realizable if and only if $\spec'=L^{\forall}_K(A)$ is $\PP$-realizable for $K = nm|\inputs|2^{{\bf O}(n\log_2 n)}$.
\end{thm}
\begin{proof}
According to Theorem~\ref{thm:sizeSol}, given a universal co-B\"uchi
automaton $\atm$ with $n$ states defining a specification $\spec$, 
and a preMealy machine
$\PP$ with $m$ states and $n_h$ holes, $\spec$ is $\PP$-realizable iff it is
$\PP$-realizable by a Mealy machine with $m+n_h2^{O(nlog_2 n)}$
states.  Let $\MM$ be such a Mealy machine. The rest of the proof
relies on the following lemma:

\begin{lemma}[\cite{DBLP:conf/cav/FiliotJR09}]
Let $\atm$ be a universal coB\"uchi automaton with $\alpha$ states and
$\MM$ a Mealy machine with $\beta$ states, we have that $L_{\omega}(\MM) \subseteq L^{\forall}(\atm)$ iff  $L_{\omega}(\MM) \subseteq L_k^{\forall}(\atm)$ for $k= \alpha \times \beta$.
\end{lemma}

Therefore, we get that $\MM$ realizes $L_K^{\forall}(\atm)$ for $K =
n\times (m+n_h2^{O(nlog_2 n)}) \leq nm|\inputs|2^{O(nlog_2 n)}$. Conversely, any machine realizing
$L_k^{\forall}(\atm)$, for any $k$, also realizes $L^{\forall}(\atm)$.
\end{proof}

\tocheck{The below lemma follows immediately:}
\begin{lemma}
\label{lem:mono}
For all co-B\"uchi automata $\atm$, for all preMealy machines $\PP$, for all $k_1 \leq k_2$, we have that $L^\forall_{k_1}(\atm) \subseteq L^\forall_{k_2}(\atm)$ and so if $L^\forall_{k_1}(\atm)$ is $\PP$-realizable then $L^\forall_{k_2}(\atm)$ is $\PP$-realizable.
Furthermore for all $k \geq 0$, if $\spec'=L^{\forall}_k(A)$ is $\PP$-realizable then $\spec = L^{\forall}(\atm)$ is $\PP$-realizable.
\end{lemma}

Thanks to the latter two results applied to $\PP = \textsf{PTA}(E)$
for a set $E$ of examples of size $m$, we can design an algorithm for synthesising
Mealy machines from a specification defined by a universal coB\"uchi
automaton $\atm$ with $n$ states and $E$: it calls $\textsc{SynthSafe}$ on the safety
specification $L^\forall_k(\atm)$ and $E$ for increasing values of
$k$, until it concludes positively, or reach the bound $K = 2^{{\bf
    O}(m n\log_2 mn)}+1$. In the latter case, it returns
\textsf{UNREAL}. However, to apply \textsc{SynthSafe} properly,
$L^\forall_k(\atm)$ must be represented by a deterministic safety
automaton. This is possible as $k$-coB\"uchi automata are
determinizable~\cite{DBLP:conf/cav/FiliotJR09}. 

\paragraph{Determinization} The determinization of $k$-co-B\"uchi automata $\atm$ relies on a simple
generalization of the subset construction: in addition to remembering
the set of states that can be reached by a prefix of a run while
reading an infinite word, the construction counts the maximal number
of times a run prefix that reaches a given state $q$ has visited
states labelled with color $1$ (remember that a run can visit at most
$k$ such states to be accepting). The states of the deterministic automaton are so-called
{\em counting functions}, formally defined for a co-B\"uchi automaton
$\atm=(Q,q_{{\sf init}},\Sigma,\delta,d)$ and $k \in \mathbb{N}$, as
the set noted $CF(\atm,k)$ of functions $f : Q \rightarrow
\{-1,0,1,\dots,k,k+1\}$. If $f(q)=-1$ for some state $q$, it means
that $q$ is inactive (no run of $\atm$ reach $q$ on the current
prefix). The initial counting function $f_{{\sf init}}$ maps all
$1$-colored initial states to $1$, all $0$-colored initial states to
$0$ and all other states to $-1$. We denote by
$\mathcal{D}(\atm,k)=(Q^\mathcal{D} = CF(\atm,k),q^\mathcal{D}_{{\sf
    init}}=f_{{\sf init}},\Sigma,\delta^\mathcal{D},Q^{\mathcal{D}}_{\sf usf})$ the
deterministic automaton obtained by this determinization procedure. 
\tocheck{We now provide a formal description below:}

\begin{definition}[Determinization with $CF(\atm,k)$]
\label{def:determinization}
Let $\atm=(Q,q_{{\sf init}},\Sigma,\delta,d)$ be a co-B\"uchi automaton and $k \in \mathbb{N}$. 
We associate to the pair $(\atm,k)$, the deterministic safety automaton $\mathcal{D}(\atm,k)=(Q^\mathcal{D},q^\mathcal{D}_{{\sf init}},\Sigma,\delta^\mathcal{D},Q^{\mathcal{D}}_{\sf usf})$ where:
  \begin{enumerate}
      \item $Q^\mathcal{D}=CF(\atm,k)$ is the set of $k$-counting functions for $\atm$.
      
      \item $q^\mathcal{D}_{{\sf init}}=f_0$ where $f_0(q)=-1$ for all $q \not= q_{\sf init}$, and $f_0(q)=0$ for $q=q_{\sf init}$ and $d(q_{\sf init})=2$, and $f_0(q)=1$ for $q=q_{\sf init}$ and $d(q_{\sf init})=1$. Informally, the states that have been assigned the value $-1$ are inactive. Initially, only $q=q_{\sf init}$ is active. If it is labelled with color $1$, its counter equals $1$, otherwise it is equal to $0$.
      
      \item For all $f \in CF(\atm,k)$, and $\sigma \in \Sigma$, the
        transition function $\delta^\mathcal{D}$ is defined as
        follows: $\delta^\mathcal{D}(f_1,\sigma)=f_2$ where for all $q
        \in Q$, $f_2(q)=$ $$\min \left ( \left ( \max_{ q' \in Q: f_1(q')
                  \geq 0\land q \in \delta(q',\sigma)} f_1(q') \right
            )+x, k{+}1 \right ),\text{ with $x=1$ if $d(q)=1$, and $x=0$ if $d(q)=2$.}$$
      
      \item The set of unsafe counting functions is defined\footnote{It is easy to check that $Q^{\mathcal{D}}_{\sf usf}$ is a trap as required.} as 
        $Q^{\mathcal{D}}_{\sf usf}=\{ f \mid \exists q \in Q \cdot f(q)=k+1 \}$. 
      
  \end{enumerate}
The language defined by $\mathcal{D}(\atm,k)$ is the set of infinite words $w \in \Sigma^{\omega}$ such that the unique run of $\mathcal{D}(\atm,k)$ on $w$ never visits a state (counting function) $f$ 
such that $d(f)=1$. This (safety) language of infinite words is denoted by $L(\mathcal{D}(\atm,k))$. The size of $\mathcal{D}(\atm,k)$ is bounded by $k^{{\bf O}(|\atm|)}$.
\end{definition}

\begin{lemma}[$\mathcal{D}(\atm,k)$ correctness, \cite{DBLP:conf/cav/FiliotJR09}]
\label{lem:correctdeterm}
For all universal co-B\"uchi automaton $\atm$, for all $k \in \mathbb{N}$, $L^{\forall}_k(\atm)=L(\mathcal{D}(\atm,k))$.
\end{lemma}

We can now give algorithm \textsc{SynthLearn}, in pseudo-code, as
Algo~\ref{algo:synt}.

\setlength{\textfloatsep}{0.1cm}
\begin{algorithm}[h]

     \DontPrintSemicolon
    \KwIn{A universal co-B\"uchi automaton $\atm$ with $n$ states, a finite
      set of examples $E\subseteq (\inputs.\outputs)^*$, a
      generalizing strategy $\sigma_G$ and a completion
      strategy $\sigma_C$.}
    \KwOut{A Mealy machine $\MM$ realizing $L^\forall(\atm)$ and such that $E\subseteq L(\MM)$ if it exists, otherwise UNREAL. }

    $K \gets nm|\inputs|2^{{\bf O}( n\log_2 n )}$;  $k \gets 0$;\tcp*{$m$ is the
      size of $E$}
    
    \While{$k \leq K$}
    {
    \If
    {${\text{\textsc{SynthSafe}}}(E,\mathcal{D}(\atm,k),\sigma_C,\sigma_G)\neq{\it UNREAL}$}
    {$\textbf{return~} \textsc{SynthSafe}(E,\mathcal{D}(\atm,k),\sigma_C,\sigma_G)$ 
      }
      
    $k \gets k+1;$
    
}
$\textbf{return~} \textit{UNREAL}$
   
\caption{{\sc SynthLearn}($E$,$\atm$,$\sigma_G$,$\sigma_C$) --
  synthesis algorithm from $\omega$-regular specification and examples
  by a reduction to safety      \label{algo:synt}}
\end{algorithm}

\paragraph{Complexity considerations and improving the upper-bound}
As the automaton $\mathcal{D}(\atm,k)$ is in the worst-case
exponential in the size of the automaton $\atm$, a direct application
of Theorem~\ref{thm:terminationcorrectness} yields a doubly exponential time procedure. 
This complexity is a consequence of the fact that the  $\PP$-realizability
problem is Exptime in the size of the deterministic automaton as shown
in Theorem~\ref{thm:sizeSol}, and that the termination of the
completion procedure is also worst-case exponential in the size of the
deterministic automaton. 

We show that we can improve the complexity of each call to
\textsc{SynthSafe} and obtain an optimal worst-case (single) exponential
complexity. We provide an
algorithm to check $\PP$-realizability of a specification
$\spec=L^{\forall}_k(\atm)$ that runs in time singly exponential in the size of
$\atm$ and polynomial in $k$ and the size of $\PP$. 
Second, we provide a finer complexity analysis for the termination of
the completion algorithm, which exhibits a worst case exponential time
in $|\atm|$. Those two improvements lead to an overall complexity of
\textsc{SynthLearn} which is exponential in the size of the
specification $\atm$ and polynomial in the set of examples $|E|$. This
is provably worst-case optimal because for $E=\emptyset$ the problem
is already {\sc ExpTime-Complete}. 

We explain next the first improvement, the
upper-bound for termination.

\label{app:termination-kco}

We establish an upper-bound
on the number of iterations needed to complete the preMealy machine
output by the procedure \textsc{Gen} at the end of the first phase of our
synthesis algorithm (in the case of a specification given by a
$k$-coB\"uchi automaton $\atm$ is realizable). To obtain the required
exponential bound, we rely on the maximal length of chains of
antichains of counting functions partially ordered as follows: let $A
\in \anti_{\preceq}(CF(\atm,k))$ and $B \in
\anti_{\preceq}(CF(\atm,k))$, then $A \trianglelefteq_{{\sf CF}} B$ if
and only if $\forall f \in A \cdot \exists g \in B \cdot f \preceq g$.
The length of those chains is bounded by $k^{{\bf O}(n)}$:

\begin{lemma}\label{lem:anticf}
Any $\vartriangleleft_{{\sf CF}}$-chain in $(\anti_{\preceq}(CF(\atm,k)),\trianglelefteq_{{\sf CF}})$
     has length at most $k^{{\bf O}(n)}$ where $n$ is the number of states in $\atm$. 
\end{lemma}

\begin{proof}
    Just as in the proof of Lemma~\ref{lem:anti}, for an antichain $X
    = \{f_1,\dots,f_n\}$ of counting functions, we define
    ${\downarrow} X$ its downward closure with respect to
    $\preceq$. Then, given another antichain $Y$, we get that
    $X\vartriangleleft_{{\sf CF}} Y$ iff ${\downarrow} X\subsetneq
    {\downarrow} Y$. Therefore the maximal length of a 
    $\vartriangleleft_{{\sf CF}}$-chain is bounded by the number of
    counting functions, which is $k^{{\bf O}(n)}$. \qed
\end{proof}

\paragraph{Checking $\PP$-realizability of a specification $\spec=L^{\forall}_k(\atm)$}

To obtain a better complexity, we exploit some structure that exists
in the deterministic automaton $\mathcal{D}(\atm,k)$. First, the set
of counting functions $CF(\atm,k)$ forms a complete lattice for the
partial order $\preceq$ defined by $f_1\preceq f_2$ if $f_1(q)\leq
f_2(q)$ for all states $q$. We denote
by $f_1\bigsqcup f_2$ the least upper-bound of $f_1,f_2$, and by $W_k^\atm$ the set of counting functions $f$ such that the
specification $L(\mathcal{D}(\atm,k)[f])$ is realizable (i.e. the
specification defined by $\mathcal{D}(\atm,k)$ with initial state
$f$). It is known that $W_k^\atm$ is downward-closed for
$\preceq$~\cite{DBLP:conf/cav/FiliotJR09}, because for all $f_1\preceq f_2$, any machine realizing 
$L(\mathcal{D}(\atm,k)[f_2])$ also realizes
$L(\mathcal{D}(\atm,k)[f_1])$. Therefore, $W_k^\atm$ can be
represented compactly by the antichain $\lceil W^{\atm}_k \rceil$ of its $\preceq$-maximal
elements. Now, the first improvement is obtained thanks to the
following result:

\begin{lemma}\label{lem:efficientpprealizability}
Given a preMealy $\PP=(M,m_0,\Delta)$, a co-B\"uchi automata $\atm$,
and $k \in \mathbb{N}$. For all states $m\in M$, we let
$F^*(m)=\bigsqcup \{ f \mid \exists u \in (\inputs \outputs)^* \cdot
\transS_{\PP}^*(m_0,u) = m \land \transS_{\mathcal{D}}(f_0,u)=f
\}$. Then, $L(\mathcal{D}(\atm,k))$ is $\PP$-realizable iff there does not exist $m \in M$ such that $F^*(m) \not \in W^{\atm}_k$.
\end{lemma}

It is easily shown that the operator $F^*$ can be computed in ptime.
Thus, the latter lemma implies that there is a polynomial time algorithm in $|\PP|$, $|\atm|$, $k
\in \mathbb{N}$, and the size of $\lceil W^{\atm}_k \rceil$ to check
the $\PP$-realizability of $L^{\forall}(\atm)$. Formal details can be
found in App.~\ref{app:improve-ppcheck}.

We end this subsection by summarizing the behavior of our synthesis algorithm for $\omega$-regular specifications defined as  universal co-B\"uchi automata.

\begin{thm}
\label{thm:onlypoly}
Given a universal coB\"uchi automaton $\atm$ and a set of examples
$E$, the synthesis algorithm \textsc{SynthLearn} returns,  if it
exists, a Mealy machine $\MM$ such that $E \subseteq L(\MM)$ and
$L_{\omega}(\MM) \subseteq L^{\forall}(\atm)$, in worst-case
exponential time in the size of $\atm$ and polynomial in the size of
$E$. Otherwise, it returns {\it UNREAL}.
\end{thm}

\tocheck{Notice that Alg.~\ref{algo:synt} calls Alg.~\ref{algo:gen} which itself calls the procedure that checks $\PP$-realizability and checking $\PP$-realizability is in polynomial time as we compute the fixpoint and check if it is safe.}

\paragraph{Specifications given as an LTL formula}\label{subsec:ltl}
We are now in position to apply Alg.~\ref{algo:synt} to a
specification given as LTL formula $\varphi$. Indeed, thanks to the
results of the subsection above, to provide an algorithm for LTL
specifications, we only need to translate $\varphi$ into a universal
co-B\"uchi automaton. This can be done according to the next lemma. It
is well-known (see~\cite{DBLP:conf/focs/KupfermanV05}), that given an LTL formula $\varphi$ over two sets of atomic propositions $P_{\inputs}$ and $P_{\outputs}$, we can construct in exponential time a universal co-B\"uchi automaton $\atm_{\varphi}$ such that $L^{\forall}(\atm_{\varphi})=\sem{\varphi}$, i.e. $\atm$ recognizes exactly the set of words $w \in (2^{P_\inputs} 2^{P_\outputs})^{\omega}$ that satisfy $\varphi$.
We then get the following theorem that gives the complexity of our
synthesis algorithm for a set of examples $E$ and an LTL formula
$\varphi$, complexity which is provably worst-case optimal as deciding if $\sem{\varphi}$ is realizable with $E=\emptyset$, i.e. the plain LTL realizability problem, is already {\sc
  2ExpTime-Complete}~\cite{DBLP:conf/icalp/PnueliR89}.

\begin{thm}
\label{thm:complexityLTL}
Given an LTL formula $\varphi$ and a set of examples $E$, the
synthesis algorithm \textsc{SynthLearn} returns a Mealy machine $\MM$
such that $E \subseteq L(\MM)$ and $L_{\omega}(\MM) \subseteq
\sem{\varphi}$ if it exists, in worst-case  doubly exponential time in
the size of $\varphi$ and polynomial in the size of $E$. Otherwise  it
returns {\it UNREAL}.
\end{thm}

%% file: MergeCompletionStrategies.tex
To implement the algorithms of previous sections, we need to fix strategies to choose among candidates for possible merges during the generalization phase and possible choices of outputs during the completion phase. The strategies that we have implemented are as follows.

First, we consider a \emph{merging} strategy $\sigma_G$ which is defined over $4$-tuples $(\MM,m,E,X)$ where $\MM$ is a preMealy machine, $m$ is a state of $\MM$, $E$ is a set of examples and $X$ is subset of states of $\MM$ for which a merge is possible, and returns a state of $X$ with the following properties.
Given an example $e$ that leads in the current preMealy machine to a state $m$ and a set of candidates $\{m_1,m_2, \dots, m_k\}$ for merging as computed in line~7 of Algorithm~\ref{algo:gen}, we associate to each state $m_i$ the counting functions computed by the fixed point $F^*$ on the current preMealy machine. Our merging strategy then choose one state $m_i$ labelled with a $\preceq$-minimal elements in this set. Intuitively, favouring minimal counting functions preserves as much as possible the set of behaviors that are possible after the example $e$. Indeed, by Lemma~\ref{lem:CFmono}, we know that if $f_1 \preceq f_2$ then $L(\mathcal{D}(\atm,k)[f_2]) \subseteq L(\mathcal{D}(\atm,k)[f_1])$.

Second, we consider a \emph{completion strategy} $\sigma_C$ which is a function defined over all triples $(\MM, m, \inp, X)$ where $\MM$ is the current preMealy machine with set of states $M$, $(m,\inp)$ is a hole of $\MM$, and $X\subseteq \outputs\times (M\cup \{\textsf{fresh}\})$ is a list of candidate pairs $(\out,m')$. It returns an element of $X$, i.e., $\sigma_C(\MM,m,\inp,X)\in X$ and it has the following properties. Remember that, for ensuring termination, the completion strategy $\sigma_C$ must be \emph{lazy}, i.e. if $X\setminus(\outputs\times \{\textsf{fresh}\})\neq \varnothing$, then $\sigma_C(\MM,m,\inp,X)\not\in \outputs\times \{\textsf{fresh}\}$.
Then among the set of possible candidates $\{ (\out_1,m_1),
(\out_2,m_2), \dots, (\out_k,m_k) \}$, we again favour states
associated with $\preceq$-minimal counting functions computed by $F^*$
on the current preMealy machine.

%% file: ProofMealyCompleteness.tex
\begin{proof}
    We prove that the generalizing phase of \textsc{SynthLearn} is
    already complete, i.e., given a well-chosen set of examples, it
    already returns $\MM$. So, the completion phase immediately returns $\MM$
    as well, as there is no holes in $\MM$. The Mealy completeness
    result for the generalizing phase is stated in
    Lemma~\ref{lem:learnability} below. \qed
\end{proof}

The next result states that any
minimal Mealy machine realizing a given specification can be learnt
when given as input a set of examples which includes a characteristic
set of examples of polynomial size in the size of the machine.

\begin{lemma}\label{lem:learnability}
    For all specification $\spec$ given as a det. safety automaton and all minimal Mealy machine $\TT$ realizing
    $\spec$, there exists $E_\TT\subseteq (\inputs\outputs)^*$ of polynomial
    size (in the number of states and transitions of $\TT$) such that for all merging
    strategy $\sigma_G$ and all finite set $E$ s.t. $E_\TT\subseteq
    E\subseteq L(\TT)$, $\textsf{GEN}(E,\spec,\sigma_G) = \TT$.
\end{lemma}

\begin{proof}
We start by defining the characteristic sample and we provide an
overview of the proof. Then, we give more formal details.
    Let $\TT = (T,t_0, \trans_\TT)$ and for all $t\in T$, let $s_t\in \inputs^*$ be a
    $\preceq_{ll}$-minimal word to reach $t$, i.e. such that
    $\transS^*_\TT(t_0,s_t) = t$. Note
    that $s_{t_0} = \epsilon$. Since $\TT$ is minimal, then for any two states $t,t'$ such that
    $t\neq t'$, there exists a unique $\preceq_{ll}$-minimal word
    $d_{t,t'} \in \inputs^+$ distinguishing $t$ and $t'$,
    i.e. such that the sequences of outputs produced by $\TT$ from $t$
    and $t'$ respectively, when reading $d_{t,t'}$, are
    different. Formally, it means that if $d_{t,t'} = \inp_1\dots
    \inp_n$, there exists $1\leq j\leq n$ such that
    $\transO_\TT(t,\inp_1\dots\inp_j)\neq
    \transO_\TT(t',\inp_1\dots\inp_j)$. 
    Note that $d_{t,t'} = d_{t',t}$.
    Let us now define $E_\TT$ (the characteristic sample). For
    any pair $p = (t,\inp)$ such that $\trans_\TT(t,\inp)$ is defined,
    let $e_p = \fIO^\TT(s_t\inp)$ (the notation $\fIO$ has been
    defined in App.~\ref{app:preMealynotations}). In other words, we have one
example per transition of the machine. Now, we also define
examples that prevent some states of the PTA to be merged. For all
$t\neq t'\in T$, we define the example $v_{t,t'} = \fIO^{\TT}(s_t
d_{t,t'})$  and finally let 
$$
E_\TT = \{ v_{t,t'}\mid t,t'\in T, t\neq t'\}\cup \{ e_p\mid p\in
T\times\inputs,\trans_\TT(p)\text{ is defined}\}$$
Let $E$ be a finite set such that $E_\TT\subseteq E\subseteq L(\TT)$. To
prove that $\textsf{GEN}(E, \spec, \sigma_G) = \TT$ (up to
state renaming), for any specification $\spec$ and strategy
$\sigma_G$, we prove the following invariant: the equivalence relation $\sim_e$
computed at iteration $e$ of the algorithm is coarser than the
equivalence relation $\sim_\TT$ which identifies two states of
$\textsf{PTA}(E)$, i.e., two examples of $E$, whenever they reach the
same state in $\TT$. Moreover, both equivalence relations coincide
when restricted to all examples $e'\preceq e$. To show this result, we
prove that a $\sim_e$-class $[e_1]$ can be merged with another
$\sim_e$-class $[e_2]$ iff $e_1\sim_\TT e_2$.

We now give the formal proof. In the sequel, we assume that $E$ is prefix-closed, in the
sense that $E = \prefs(E)\cap (\inputs\outputs)^*$. This is wlog as
the algorithm first computes the prefix closure of $E$ at
line~\ref{line:closure}.  We first prove some useful claim.

For all $e\in E$, we let $\Phi(e)\in T$ such that
$\transS_\TT^*(t_0,e) = \Phi(e)$. Given $e,e'\in E$, we say that $e$
and $e'$ are $\TT$-equivalent, denoted $e\sim_T e'$, if $\Phi(e) =
\Phi(e')$. The following claim states that $\sim_T$ is a
Mealy-congruence for $\textsf{PTA}(E)$ and quotienting the latter by
$\sim_T$ yields exactly $\TT$ (up to state renaming).

\textit{Claim 1} $\sim_T$ is a Mealy-congruence for $\textsf{PTA}(E)$ and $\TT =
\textsf{PTA}(E)/_{\sim_T}$ (up to state renaming).

We give a few intuitions for proving that claim. The detailed proof
can be found in App.~\ref{app:claimscharac}. Since $E\subseteq L(\TT)$, it can be proved that 
$\Phi(\transS_{\textsf{PTA}(E)}(e,\inp)) =
\transS_\TT(\Phi(e),\inp)$ (if there exists $\out\in\outputs$ such
that $e\inp\out\in E$, otherwise $\transS_{\textsf{PTA}(E)}(e,\inp)$
is undefined). This entails that $\sim_T$ is a
congruence. Similarly, we also get that
$\transO_{\textsf{PTA}(E)}(e,\inp) = \transO_\TT(\Phi(e),\inp)$ (if
defined) which
entails that $\sim_T$ is a Mealy-congruence for $\textsf{PTA}(E)$. To
show that the quotient of $\textsf{PTA}(E)$ by $\sim_T$ is $\TT$, we
first use the fact that $E$ contains one example per state of $\TT$,
and so $\sim_T$ has as many equivalence classes as the number of
states of $\TT$. We have already seen that the output produced by
a transition of $\textsf{PTA}(E)/_{\sim_T}$ is consistent with the output produced
  by $\TT$, when the transition is defined, because $\sim_T$ is a
  Mealy-congruence. We prove that for all classes of $\sim_T$ and all
  inputs, the transition of $\textsf{PTA}(E)/_{\sim_T}$ is defined,
    because $E$ contains one example $e_p$ per transition of $\TT$.

Now, let us come back to the proof of the lemma. For all $e\in E$, let $\sim_e$ be the Mealy-congruence
computed after iteration $e$ of the loop at
line~\ref{line:loop}. We prove that $\sim_T$ is coarser than any $\sim_e$ for
all $e$ and $\sim_T$ is equal to $\sim_e$ if restricted to the
$\preceq_{ll}$-downward closure of $e$. This will be sufficient to
conclude that our algorithm returns $\TT$ (up to state
renaming). Formally, given $e\in E$, we let $\downarrow e = \{ e'\in
(\inputs\outputs)^*\mid e'\preceq_{ll} e\}$. Note that $\downarrow
e\subseteq E$. We prove the following two invariants, which states
that $\sim_e$ is always finer than $\sim_T$ and that $\sim_T$ and
$\sim_e$ coincides when restricted to $\downarrow e$.

      \begin{itemize}

        \item \textbf{INV 1} For all $e\in E$, $\sim_e\finer \sim_T$.
        \item \textbf{INV 2} For all $e$, $\sim_e\cap (\downarrow e)^2
          = \sim_T\cap (\downarrow e)^2$.
    
\end{itemize}

Before proving the invariants, let us show that \textbf{INV 2} implies the
statement of the lemma. Indeed, let $e^*=
\text{max}_{\preceq_{ll}}(E)$. Then, $\downarrow e^* = E$, so
$\sim_{e^*}\cap (\downarrow e^*)^2
= \sim_{e^*} = \sim_T$. Therefore, the machine returned by the
algorithm is $\textsf{PTA}(E)/_{\sim_T}$, so by \textit{Claim 1} we get the desired
result.

We rely on a useful claim which states that if $\sim$ is finer
than $\sim_T$, and $x\sim_T y$, then merging the $\sim$-class of $x$
and the $\sim$-class of $y$ yields an equivalence relation finer than
$\sim_T$. Intuitively, it is because $[x]_\sim$ and $[y]_\sim$ are
subsets of the same $\sim_T$-class, and any merge occurring recursively
when computing $\sim^{x,y}$ also preserves this property. The detailed
proof can be found in App.~\ref{app:claimscharac}.

        \textit{Claim 2}: for all $x,y\in E$, if $x\sim_T y$ and
        $\sim\finer\sim_T$, then $\sim^{x,y} \finer \sim_T$.

It remains to prove \textbf{INV 1} and \textbf{INV 2}. We prove them
together by induction.

\textit{Initialisation} The initial step $e = \epsilon$ is
simple. Indeed, $\sim_{\epsilon} = diag_E$, so $\sim_\epsilon\finer
\sim_T$. Moreover, $\downarrow
\epsilon = \{\epsilon\}$, so, $\sim_\epsilon\cap (\downarrow
\epsilon)^2 = \{(\epsilon,\epsilon)\}
          = \sim_T\cap (\downarrow \epsilon)^2$.

\textit{Induction step} We now prove that the invariants are preserved after one
    iteration. Suppose they are true for $e\in E$ and let us show
    it is true for $f\in E$ such that $f$ is the immediate
    successor of $e$ in $E$ in llex-order. Let us give some intuitions
    before the formal details. Intuitively, we prove that a
    merge is possible between $[f]_{\sim_e}$ and some $[y]_{\sim_e}$ such that
    $y\preceq_{ll} e$, iff $f\sim_T y$. To prove the ``only if''
    direction, we exploit the fact that when $f\not\sim_T y$, merging
    their $\sim_e$-classes would produce a congruence which is not a
    Mealy congruence, because $\Phi(f)$ and $\Phi(y)$ can be
    distinguished by $d_{\Phi(f),\Phi(y)}$.

    Let us proceed with the formal proof. We distinguish between two
    cases, depending on whether $mergeCand$ at
    line~\ref{line:mergeable} is empty or not. If it is empty, then we
    prove that for all $y\preceq_{ll} e$, $\Phi(y)\neq \Phi(f)$. If it
    is non-empty, we prove that any $y\in mergeCand$ satisfies $\Phi(y)=\Phi(f)$.

    \begin{itemize}

      \item CASE 1: $mergeCand = \varnothing$. In this case, there is
        no merge, therefore $\sim_e = \sim_f$. The induction hypothesis immediately gives \textbf{INV 1}. To
        prove \textbf{INV 2}, we need to show that for all $x,y
        \in\downarrow f$, if $x\sim_T y$ then $x\sim_f y$, i.e.,
        $x\sim_e y$. It is the case by induction hypothesis whenever
        $x,y\in \downarrow e$. If $x = y = f$, then it is true by
        reflexivity of $\sim_e$. The remaining case is $x=f$ and
        $y\preceq_{ll} e$. We show that this case is actually
        impossible, because
        $\textsf{Mergeable}(\textsf{PTA}(E),\sim_e,f,y)$ would hold
        otherwise (and hence $mergeCand \neq \varnothing$). So, assume
        that $\Phi(f) = \Phi(y)$ and let us prove that 
        $\sim_e^{f, y}$ is a Mealy-congruence for $\textsf{PTA}(E)$.

        By Claim 2, we have that for all
        $\alpha\sim_e^{f,y} \beta$, $\Phi(\alpha) = \Phi(\beta)$. Now,
        suppose that there exist $\alpha\inp\out_1\in E$ and
        $\beta\inp\out_2\in E$, then, $\out_1 =
        \transO_\TT(\Phi(\alpha),\inp) = \transO_\TT(\Phi(\beta),\inp)
        = \out_2$ (because $E\subseteq L(\TT)$). It means that
        $\sim_e^{f,y}$ is a Mealy-congruence for $\textsf{PTA}(E)$, and hence
        $mergeCand\neq \varnothing$, contradiction.

      \item CASE 2: $mergeCand \neq \varnothing$. In that case, there
        is a merge between $f$ and some $y\preceq_{ll} e$. Therefore
        $\sim_f = \sim_e^{f,y}$.

        We first prove \textbf{INV 2}. We need to show that for all
        $\alpha,\beta\preceq_{ll} f$ such that $\alpha\sim_T\beta$, we
        have $\alpha\sim_f \beta$. If $\alpha,\beta\preceq_{ll} e$,
        then by IH, $\alpha\sim_e\beta$ and since $\sim_e\finer
        \sim_f$, we get $\alpha\sim_f\beta$. If $\alpha = \beta = f$,
        then we are done by reflexivity. So, assume that $\alpha = f$
        and $\beta\preceq_{ll} e$. We have $f\sim_T\beta$ and
        $y\sim_Tf$, so $y\sim_T\beta$ and since $y,\beta\preceq_{ll}
        e$, by induction hypothesis, $y\sim_e \beta$.  Now and
        informally,  since $\sim_f$ merges the $\sim_e$-class of
        $f$ and that of $y$ and propagates this merge, we get $\beta\sim_f f\sim_f y$. Formally, by definition of
        $\sim_e^{f,y,0}$, we get that $y\sim_e^{f,y,0}
        \beta$. Moreover, $f \sim_e^{f,y,0} y$, therefore
        $f\sim_e^{f,y,0} \beta$. By Proposition~\ref{prop:equiv}, the
        $U$ is increasing for the order $\finer$, hence,
        $\sim_e^{f,y,0}\finer \sim_f$ and therefore $f\sim_f \beta$
        and we are done proving \textbf{INV 2}.

        Let us now prove \textbf{INV 1}. First, if $\Phi(f) = \Phi(y)$, then by Claim 2, we get that for all $\alpha\sim_f\beta$,
        $\Phi(\alpha) = \Phi(\beta)$, i.e., $\alpha\sim_T\beta$. This
        shows $\sim_f\finer \sim_T$ (\textbf{INV 1}). So, it remains
        to show that $\Phi(f) = \Phi(y)$. Suppose that $\Phi(f) \neq
        \Phi(y)$. We prove that $y\not\in mergeCand$, which is a
        contradiction, i.e. that $\sim_e^{f,y}$ cannot be a
        Mealy-congruence for $\textsf{PTA}(E)$. Remind that $s_{\Phi(y)}$ is the
        minimal word reaching $\Phi(y)$ in $\TT$ and $s_{\Phi(f)}$ is
        the minimal word reaching $\Phi(f)$ in $\TT$. Hence
        $s_{\Phi(y)}\sim_T y$ and $s_{\Phi(f)}\sim_T f$. Since
        $s_{\Phi(f)}\preceq_{ll} y\preceq_{ll} e$ and
        $s_{\Phi(f)}\preceq_{ll} f$, by $\textbf{INV 2}$ we get that
        $s_{\Phi(f}\sim_f f$ and $s_{\Phi(y)}\sim_f y$. Moreover, by
        definition of $\sim_f$, $y\sim_f y$. Hence, $s_{\Phi(f)}\sim_f
        s_{\Phi(y)}$. Consider the input word $d_{\Phi(f),\Phi(y)}$
        distinguishing $\Phi(f)$ and $\Phi(y)$ and decompose it as
        $d\inp$. By definition of $E$, $\fIO^\TT(s_{\Phi(f)}d\inp)\in
        E$ and $\fIO^\TT(s_{\Phi(y)}d\inp)\in E$ ((the notation $\fIO$ has been
    defined in App.~\ref{app:preMealynotations})). There exists
        $\out_1\neq\out_2\in \outputs$ such that
        $\fIO^\TT(s_{\Phi(f)}d\inp) =
        \fIO^\TT(s_{\Phi(f)}d)\inp\out_1$ and $\fIO^\TT(s_{\Phi(y)}d\inp) =
        \fIO^\TT(s_{\Phi(y)}d)\inp\out_2$. Since $\sim_f$ is a
        congruence, we get $\fIO^\TT(s_{\Phi(f)}d)\sim_f
        \fIO^\TT(s_{\Phi(y)}d)$ but
        $\transO_{\textsf{PTA}(E)}(\fIO^\TT(s_{\Phi(f)}d),\inp) =
        \out_1\neq \out_2=
        \transO_{\textsf{PTA}(E)}(\fIO^\TT(s_{\Phi(y)}d),\inp)$. This
        shows that $\sim_f$ is not a Mealy-congruence for $\textsf{PTA}(E)$ and
        hence $y$ and $f$ cannot be merged. Therefore, $\Phi(f) =
        \Phi(y)$ and we are done. 
\end{itemize}

Note that the proof does not rely on any particular specification $\spec$
nor any merging strategy $\sigma_G$.
\end{proof}